\documentclass[a4paper,12pt,numbers,sort&compress]{article}

\pdfoutput=1
\usepackage[paper=letterpaper,margin=1in]{geometry}
\usepackage{amsmath,array,amsfonts,graphicx,wrapfig,lscape,float,mathtools,multirow,longtable}
\usepackage[dvipsnames]{xcolor}
\usepackage{hyperref,caption,setspace}
\usepackage[makeroom]{cancel}
\usepackage{cite}
\usepackage{tocloft}

\renewcommand{\baselinestretch}{1.2}
\newcommand{\eref}[1]{(\ref{#1})}
\newcommand{\fref}[1]{Figure~\ref{#1}}
\newcommand{\sref}[1]{\S\ref{#1}}
\newcommand{\tref}[1]{Table~\ref{#1}}
\newcommand{\nn}{\nonumber}

\newcommand{\beal}[1]{\begin{eqnarray}\label{#1}}
\newcommand{\eea}{\end{eqnarray}}
\newcommand{\ba}{\begin{array}}
\newcommand{\ea}{\end{array}}
\newcommand{\ud}{\mathrm{d}}
\newcommand{\Tr}{\text{Tr}}

\newcommand{\comment}[1]{}

\newcommand{\vev}[1]{\langle #1 \rangle}

\newcommand{\setall}{\setcounter{equation}{0}
        \setcounter{theorem}{0}}

\newtheorem{theorem}{\bf Theorem}

\newtheorem{conjecture}[theorem]{\bf Conjecture}

\newtheorem{lemma}[theorem]{\bf Lemma}
\newtheorem{proposition}[theorem]{Proposition}
\newtheorem{corollary}[theorem]{Corollary}
\newtheorem{definition}[theorem]{Definition}

\newenvironment{proof}[1][Proof]{\begin{trivlist}
\item[\hskip \labelsep {\bfseries #1}]}{\end{trivlist}}

\newcommand{\qed}{\nobreak \ifvmode \relax \else
      \ifdim\lastskip<1.5em \hskip-\lastskip
      \hskip1.5em plus0em minus0.5em \fi \nobreak
      \vrule height0.75em width0.5em depth0.25em\fi}

\newcommand{\cB}{{\cal B}}
\newcommand{\cC}{{\cal C}}

\newcommand{\cX}{{\cal X}}

\newcommand{\IP}{\mathbb{P}}
\newcommand{\IQ}{\mathbb{Q}}

\newcommand{\IN}{\mathbb{N}}
\newcommand{\IR}{\mathbb{R}}
\newcommand{\IC}{\mathbb{C}}
\newcommand{\IF}{\mathbb{F}}

\newcommand{\IZ}{\mathbb{Z}}

\newcommand{\td}{{\rm Td}}

\begin{document}

\begin{flushright}
UUITP-10/17
\end{flushright}

\vspace{1cm}

\begin{center}
{\Large \bf  Calabi-Yau Volumes and Reflexive Polytopes}
\end{center}
\medskip

\vspace{.4cm}
\centerline{
{\large Yang-Hui He}$^1$, \
{\large Rak-Kyeong Seong}$^2$, \
{\large Shing-Tung Yau}$^3$
}
\vspace*{3.0ex}

\renewcommand{\baselinestretch}{0.5}
\begin{center}
{\it
  {\small
    \begin{tabular}{cl}
      ${}^{1}$ 
      & Merton College, University of Oxford, OX14JD, UK \&\\
      & Department of Mathematics, City, University of London, EC1V 0HB, UK \&\\
      & School of Physics, NanKai University, Tianjin, 300071, P.R.~China\\
      ${}^{2}$ 
      & Department of Physics and Astronomy, Angstrom Laboratory,  \\
       & Uppsala University, Box 516, 75120 Uppsala, Sweden \\
      ${}^3$
      & Department of Mathematics, \&
      Center of Mathematical Sciences and Applications, \& \\ 
      & Jefferson Physical Laboratory,
      Harvard University, Cambridge, MA 02138, USA\\
      &\\
      & \qquad
      {\rm hey@maths.ox.ac.uk, \ rakkyeongseong@gmail.com, \ yau@math.harvard.edu}
    \end{tabular}
  }
}
\end{center}

\renewcommand{\baselinestretch}{1.2}

\vspace*{4.0ex}
\centerline{\textbf{Abstract}} \bigskip

We study various geometrical quantities for Calabi-Yau varieties realized as cones over Gorenstein Fano varieties, obtained as toric varieties from reflexive polytopes in various dimensions. Focus is made on reflexive polytopes up to dimension 4 and the minimized volumes of the Sasaki-Einstein base of the corresponding Calabi-Yau cone are calculated. By doing so, we conjecture new bounds for the Sasaki-Einstein volume with respect to various topological quantities of the corresponding toric varieties. We give interpretations about these volume bounds in the context of associated field theories via the AdS/CFT correspondence.
\\

\newpage

\begin{singlespace}
\tableofcontents
\end{singlespace}


\newpage

\section{Introduction}\setall

Whereas there has been a multitude of constructions for Calabi-Yau varieties over the years, both compact and non-compact, by far the largest number thereof have been realized with toric geometry as the point d'appui.
The combinatorial nature of toric varieties, in translating the relevant geometric quantities to manipulations on lattice polytopes, renders them particularly useful in systematically constructing and studying Calabi-Yau varieties and related phenomena in string theory and mathematics.

Compact smooth Calabi-Yau $(n-1)$-folds have been constructed \cite{bat,BB,Kreuzer:1995cd,Kreuzer:1998vb,Kreuzer:2000xy,Kreuzer:2002uu,Batyrev:1994hm,1998math......1107B,oldcy3,Braun:2011ik,Altman:2014bfa,cy3online} as hypersurfaces corresponding to the anti-canonical divisor within a $n$-dimensional toric variety coming from a reflexive polytope. 
The result was an impressive list of at least half a billion smooth, compact Calabi-Yau 3-folds. 
These, when plotted in terms of their Hodge numbers $h^{1,1} - h^{2,1}$ against $h^{1,1} + h^{2,1}$, gave rise to the famous funnel shape with left-right symmetry signifying mirror symmetry \cite{Candelas:1994bu,Batyrev:1994hm}.

Parallel to constructions of compact smooth Calabi-Yau $(n-1)$-folds, toric varieties also played a significant role in the construction of non-compact Calabi-Yau varieties.
A non-compact Calabi-Yau $n$-fold $\mathcal{X}$ of complex dimension $n$ can be realized as an affine cone over 
a complex base $X$, of complex dimension $n-1$.
By far the largest class of these affine Calabi-Yau $n$-folds is when the base is a toric variety $X(\Delta)$, from an $(n-1)$-dimension polytope $\Delta$.
A key fact is that $\mathcal{X}$ itself is a real cone over a compact, smooth Sasaki-Einstein manifold $Y$ of real dimension $2n -1$.

In string theory, the discovery of the AdS/CFT correspondence on $AdS_5\times S^5$ \cite{Maldacena:1997re}, later generalized to $AdS_5\times Y$ \cite{Morrison:1998cs,Acharya:1998db}, has led to considerable interest in Sasaki-Einstein 5-manifolds $Y$ and the non-compact Calabi-Yau $3$-fold cone over $Y$.
The dual field theory is a $4d$ $\mathcal{N}=1$ superconformal field theory on the worldvolume of a stack of D3-branes probing the Calabi-Yau singularity \cite{Maldacena:1997re,Witten:1998qj,Klebanov:1998hh}.
A large class of examples was obtained by identifying the $4d$ $\mathcal{N}=1$ worldvolume theories corresponding to the probed Calabi-Yau 3-folds, which were initially constructed from known toric varieties such as $T^{1,1}$ and the del Pezzos cones. 
Infinite families of theories corresponding to Sasaki-Einstein manifolds such as the class of $Y_{p,q}$ \cite{Martelli:2004wu,Benvenuti:2004dy}, $L_{p,q,r}$ \cite{Benvenuti:2005ja,Butti:2005sw} and $X_{p,q}$ \cite{Hanany:2005hq} manifolds followed giving a comprehensive picture over the correspondence between $4d$ $\mathcal{N}=1$ supersymmetric theories and toric non-compact Calabi-Yau 3-folds.

This development culminated in a Type IIB brane configuration of D5-branes suspended between a NS5-brane wrapping a holomorphic surface $\Sigma$, which realizes the $4d$ $\mathcal{N}=1$ supersymmetric theories and encodes its vacuum moduli space given by the probed toric Calabi-Yau 3-fold. These brane configurations, T-dual to the D3-branes probing the toric Calabi-Yau 3-fold, can be represented by a bipartite periodic graph on a 2-torus and are known as brane tilings \cite{Hanany:2005ve,Franco:2005rj}. 
Brane tilings describe for every toric Calabi-Yau 3-fold the corresponding $4d$ $\mathcal{N}=1$ theory, including theories that are related by Seiberg duality.

A parallel story was with regards to the volume of the Sasaki-Einstein 5-manifold $Y$, which is related to the central charge  $a$-function of the $4d$ $\mathcal{N}=1$ theory via the AdS/CFT correspondence \cite{Gubser:1998vd,Henningson:1998gx}. 
In \cite{Martelli:2005tp}, it was shown that minimizing the volume of the Sasaki-Einstein manifold determines its Reeb vector, which in turn showed that volume minimization is the geometric dual of a-maximization \cite{Intriligator:2003jj}. 
Following this work, \cite{Martelli:2006yb} found a general formula for the volume function using an equivariant index on the Calabi-Yau cone over $Y$, which counts holomorphic functions. 
This index is the Hilbert series of the Calabi-Yau and has been used in different contexts to count gauge invariant operators of supersymmetric gauge theories \cite{Benvenuti:2006qr,Feng:2007ur}. 
In \cite{Martelli:2006yb}, the Hilbert series and consequently the volume function was obtained using a localization formula relying on the toric data coming from the cone over the Sasaki-Einstein 5-manifold. 
This volume formula, which essentially depends only on the topological fixed point data encoded in the toric geometry of the Calabi-Yau cone, showed that the minimized volume of the Sasaki-Einstein 5-manifold is always an algebraic number representing the geometric counterpart of the maximized $a$-function.

Thus inspired, it is natural to ask how the minimum volume for toric non-compact Calabi-Yau manifolds behaves in dimension other than 3.
In particular, we focus on toric Calabi-Yau $n$-folds whose toric diagram is a $(n-1)$-dimensional reflexive polytope in $\mathbb{Z}^{n-1}$.
Such polytopes have been fully classified up to dimension $4$ by \cite{bat,Batyrev:1994hm,1998math......1107B,Kreuzer:1995cd,Kreuzer:1998vb,Kreuzer:2000xy},
providing a finite set of toric Calabi-Yau $n$-folds for which we can calculate the volume of the corresponding Sasaki-Einstein base.
By computing the volume minima for such a large set of Sasaki-Einstein manifolds as bases of the toric Calabi-Yau $n$-folds,
we gain a compendious understanding of its distribution. 
When compared to topological quantities of the corresponding toric variety, such as the Chern numbers and the Euler number, we remarkably observe that the distribution of the volume minima is not at all random.

We identify bounds of the volume minima determined by topological quantities of the toric variety, which in turn shed light on the properties of the corresponding supersymmetric gauge theories living on the probe branes at the associated Calabi-Yau singularity.
This not only provides new insights into the properties of the $4d$ $\mathcal{N}=1$ worldvolume theories on probe D3-branes at Calabi-Yau 3-folds singularities and their corresponding brane tiling constructions, but also sheds light on more recent developments regarding $3d$ $\mathcal{N}=2$ worldvolume theories on probe M2-branes \cite{Aharony:2008ug,Martelli:2008si,Hanany:2008cd,Hanany:2008fj} and $2d$ $(0,2)$ worldvolume theories on probe D1-branes \cite{Franco:2015tna,Franco:2015tya,Franco:2016qxh,GarciaCompean:1998kh,Mohri:1997ef} at toric Calabi-Yau 4-folds as well as $0d$ supersymmetric matrix models living on Euclidean D(-1)-branes at toric Calabi-Yau 5-folds \cite{Franco:2016tcm}. 
In the case of $4d$ $\mathcal{N}=1$ theories related to toric Calabi-Yau 3-folds, we re-interpret the volume bounds as bounds on the central charge $a$ in the $4d$ theory \cite{Intriligator:2003jj,Butti:2005vn,Butti:2005ps}. We make a similar connection to the free energy $F$ for $3d$ $\mathcal{N}=2$ theories living on M2-branes probing toric Calabi-Yau 4-folds \cite{Herzog:2010hf,Martelli:2011qj}.

The paper is organized as follows. 
Section \sref{sprel} overviews the construction of both compact and non-compact Calabi-Yau manifolds with a review on reflexive polytopes and toric varieties. 
The section also summarizes how non-compact toric Calabi-Yaus are related to supersymmetric gauge theories in different dimensions.
In section \sref{stop}, we give a summary of topological quantities related to toric varieties such as the Euler and Chern numbers. 
By doing so, we present combinatorial formulae that simplify the computation of these topological quantities from the toric variety.
Section \sref{shsvol} gives the construction of non-compact toric Calabi-Yau cones over Sasaki-Einstein manifolds and the computation of Hilbert series of the cone, adhering to the notation of \cite{Martelli:2006yb} which related an equivariant index on the Calabi-Yau cone with the volume function of the Sasaki-Einstein base manifold. 
At the end of the section, we explain volume minimization and how it relates to $a$-maximization for $4d$ $\mathcal{N}=1$ supersymmetric gauge theories. The method of computing this volume from the Hilbert series and Reeb vector, which in turn is exacted from the polytope data, is described in an algorithmic fashion.

Finally, section \sref{s:res} summarizes our results for volume minima for Sasaki-Einstein manifolds related to toric Calabi-Yau 3, 4 and 5-folds with reflexive toric diagrams. 
We plot the volume minima against topological quantities of the corresponding toric varieties and arrive at lower and upper bounds for the volume minima. 
These in the context of toric Calabi-Yau 3-folds and 4-folds with reflexive polygons as toric diagrams can be re-interpreted as bounds on the central charge $a$-function and the Free energy $F$ of the corresponding $4d$ $\mathcal{N}=1$ and $3d$ $\mathcal{N}=2$ theories, respectively.
The appendices summarize the implementations used to efficiently compute the minimum volumes, including tables with explicit values for volume minima for the reflexive toric Calabi-Yau cases considered.

\subsection*{Nomenclature}

\begin{longtable}{lcl}
  $\Delta$
  &:& (reflexive) convex lattice polytope;\\
  && subscript emphasizes dimension: $\Delta_{n-1} \subset \IZ^{n-1}$ \\
  && we distinguish between {\it vertices}: the extremal lattice points \& \\
  && {\it perimeter} points: lattice points (including vertices) lying on edges
  \\
  $X(\Delta_{n-1})$
  &:& toric variety corresponding to $\Delta_{n-1}$; \\
  && $\dim_{\IC} X(\Delta_{n-1}) = n-1$ \\
  FRS
  &:&
  Fine Stellar Regular triangulation of polytope \\
  $\widetilde{X(\Delta_{n-1})}$
  &:&
  smooth toric variety (if it exists) obtained from FRS of $\Delta$ \\
  $\Sigma(\Delta)$ 
  &:&
  (Inner) normal fan associated with $\Delta$ \\
  $\cX = \cC_{\IC}(X(\Delta_n))$
  &:& affine Calabi-Yau (complex) cone over $X(\Delta_{n-1})$;\\
  && $\dim_{\IC}\cX = n$\\
  $Y = \cB_{\IR}(\cX)$
  &:& Sasaki-Einstein manifold, base of $\cX$ which is real one over $Y$;\\
  && $\dim_{\IR} Y = \dim_{\IR} \cB(\cC(X(\Delta_{n-1}))) = 2n-1$ \\
  $b_i$ && components of the Reeb vector, $i=1,\ldots,n$ for $\cX$;\\
  && set the last component $b_{n} = n$\\
  $V(b_i; Y)$ && volume function for $Y$, also denoted $V(b_i; \cX)$ \\
  $b_i^*$ && the value of the Reeb vector with $b_n = n$ which minimizes
  $V(b_i; Y)$ \\
  $g(t_i; \cX)$ && multi-graded ($i=1,\ldots,n$) Hilbert series for $\cX$\\
\end{longtable}

\newpage

\section{Preliminaries \label{sprel}}

We begin with some preliminaries in both the mathematics and the physics, expanding on the nomenclature summarized above.
First, we introduce the concept of reflexive polytopes and explain why they have become so central in the construction of Calabi-Yau varieties.
Next, we focus on the affine Calabi-Yau varieties constructed therefrom and summarize their connection to quiver gauge theories.

Reflexive polygons and more generally reflexive polytopes have appeared both in the context of compact and non-compact Calabi-Yau manifolds.
Batyrev-Borisov \cite{bat,BB} first studied reflexive polyhedra $\Delta_n$ in $n$ dimensions in order to construct families of smooth Calabi-Yau hypersurfaces that are compactification in projective toric varieties given by $\Delta$. This allowed for the construction of dual families of compact $(n-1)$-dimensional Calabi-Yau varieties as hypersurfaces.
This duality is none other than mirror symmetry \cite{Batyrev:1994hm,Kreuzer:1995cd,Kreuzer:1998vb,Kreuzer:2000xy}.

Furthermore, as we will illustrate in section \sref{ss2}, non-compact toric Calabi-Yau $n$-folds have been studied more recently in the context of supersymmetric quiver gauge theories in various dimensions as worldvolume theories of probe D$(9-2n)$-branes.
These non-compact Calabi-Yau $n$-folds can be thought of as rational polyhedral cones generated by the vertices of a convex polytope $\Delta_{n-1}$ in $n-1$ dimensions.
A finite class of such non-compact Calabi-Yau $n$-folds can be generated in every dimension when the convex polytope is taken to be reflexive.
For example, in 2 dimensions, there are exactly 16 reflexive polygons as shown in \fref{f:n=2} which give rise to 16 non-compact toric Calabi-Yau 3-folds.
For these, the full list of corresponding  $4d$ $\mathcal{N}=1$ supersymmetric quiver gauge theories and of the corresponding brane tilings was classified in \cite{Hanany:2012hi}. 

From the work on mirror symmetry for compact Calabi-Yau manifolds \cite{bat,BB,Batyrev:1994hm,Kreuzer:1995cd,Kreuzer:1998vb,Kreuzer:2000xy} and the work on brane tilings for non-compact Calabi-Yau varieties \cite{Hanany:2005ve,Franco:2005rj,Feng:2005gw}, we see that reflexive polytopes are at the crux of studying Calabi-Yau geometry, offering a plenitude of new insights.
In this paper, we will focus exclusively on non-compact Calabi-Yau cones and the reader is referred to a careful exposition of the geometry and combinatorics of lattice polytopes in this context in \cite{nill} as well as an excellent rudimentary treatment in \cite{DW,Skarke:2012zg}.

\subsection{The Reflexive Polytope $\Delta$ \label{sreflexive}} 

Let $\Delta$ be a {\it convex lattice polytope} in 
$\mathbb{Z}^n$, which we can think of either as
\begin{enumerate}
\item a collection of vertices (dimension 0), each of which is a $n$-vector with integer entries, each pair of neighboring vertices defines an edge (dimension 1), each triple, a face (dimension 2), etc., and each $n-1$-tuple, a facet (dimension $n-1$ or codimension 1);
  or as
\item a list of linear inequalities with integer coefficients, each of which slices a facet defined by the hyperplane.
\end{enumerate}
We will only consider those polytopes containing the origin $(0,\ldots,0)$ as their unique interior point and define the {\it dual polytope} $\Delta^\circ$ (sometimes also called the {\it polar polytope}) to $\Delta$ as all vectors in 
$\mathbb{Z}^n$ whose inner product with all interior points of $\Delta$ is greater than or equal to $-1$: 
\begin{equation}
\label{dualp}
\mbox{Dual polytope: }
\Delta^\circ = \{ \underline{v} \in \mathbb{Z}^n ~|~ \underline{m} \cdot \underline{v} \geq -1\;\; \forall \underline{m} \in \Delta \} \ .
\end{equation}
The $-1$ is a vestige of definition (2) of $\Delta$, since any hyperplane not passing the origin can be brought to the form $\sum\limits_i c_i x_i = -1$.
In general, while the polar dual of a lattice polytope may have vertices with rational entries, the case where they are all integer vectors is certainly special and interesting. Consequently, we have that 
\begin{definition}
A convex lattice polytope $\Delta$ is \textbf{reflexive} if its dual polytope $\Delta^\circ$ is also a lattice polytope.
\end{definition}
Indeed, the symmetry of this definition means that the dual polytope $\Delta^\circ$ is also reflexive.
In particular, the origin is always a point in both and is, in fact, the {\em the unique interior point}.
Furthermore, we refer to any $\Delta = \Delta^\circ$ as self-dual or self-reflexive.

\subsection{The Toric Variety $X(\Delta)$} \label{s:toric}

Given a lattice polytope $\Delta_n$, it is standard to construct a compact toric variety $X(\Delta_n)$ of complex dimension $n$ (cf.~\cite{fulton} and an excellent recent textbook \cite{CLS}).
Briefly, one constructs the (inner) normal fan $\Sigma(\Delta)$ consisting of cones whose apex is an interior point - which can be chosen to be the origin - with rays extending to the vertices of each face.
That is, $\Sigma(\Delta)$ is constructed from $\Delta$ as the positive hull of the $n$-cone over the faces $F$ of $\Delta$ as follows
\beal{es80a1}
\Sigma = \{
\text{pos}(F) ~:~ F \in \text{Faces}(\Delta)
\}~,~
\eea
where 
\beal{es80a2}
\text{pos}(F) = 
\left\{
\sum_i \lambda_i \underline{v}_i ~:~
\underline{v}_i \in F ~,~ \lambda_i \geq 0
\right\} ~.~
\eea
We can think of $\Delta_n$ as the fan of $n$-cones subtended from the apex at the origin to the proper faces of $\Delta$, where the origin is always a point in our reflexive $\Delta$.

From the fan $\Sigma$, the construction of the compact toric variety $X(\Sigma)$ follows the standard treatment of \cite{fulton}, with each cone giving an affine patch.
The resulting variety $X(\Sigma)$ is not guaranteed to be smooth.
In order to ensure smoothness, we need a notion of regularity.
\begin{definition}
  The polytope $\Delta_n$ and the fan $\Sigma(\Delta_n)$ are called regular if every cone in the fan has generators that form part of a $\mathbb{Z}$-basis.
\end{definition}
One way to check this is to ensure that the determinant of all $n$-tuples of generators of each cone is $\pm 1$.
Note that each generator/ray of the cone is an integer $n$-vector but the number of such generators will always be $\geq n$. We need to choose all possible $n$-tuples and compute the determinant. The key result we will need is that (cf.~\cite{CLS}):
\begin{theorem}\label{thm:delta1} 
The toric variety $X(\Delta_n)$ is smooth iff $\Delta_n$ is regular. 
\end{theorem}

The above discussions hold for any convex lattice polytope.
Henceforth, we will concentrate only on reflexive polytopes $\Delta$ with a single internal lattice point, viz., the origin. According to \cite{nill}, we have:
\begin{theorem}
Any reflexive polytope $\Delta$ corresponds to a Gorenstein toric Fano variety $X(\Delta)$.
\end{theorem}
Note that $X(\Delta)$ is in general singular and as we will see the smooth ones are very rare and we will need to perform desingularization.
Indeed, when $\Delta_n$ is regular, $X(\Delta_n)$ is a smooth toric Fano $n$-fold and it is a manifold of positive curvature \cite{bat,ewald,WW,1998math......1107B}.
In the general case, by Gorenstein Fano \cite{nill} we mean that $X(\Delta)$ is a normal toric variety whose anticanonical divisor is ample and Cartier.

\subsection{Classification of $\Delta_n$} 

\begin{table}[ht!!]
\centering
  \begin{tabular}{c||c|c|c|c|c}
    \mbox{dimension} & 1 & 2 & 3 & 4 & \ldots \\ \hline
    \mbox{\# Reflexive Polytopes} & 1 & 16 & 4319 & 473,800,776 & \ldots \\
    \hline
    \mbox{\# Regular} & 1 & 5  & 18 & 124 & \ldots
  \end{tabular}
  \caption{{\sf {\small The number of inequivalent reflexive polytopes in different dimensions.
        The reflexive polytopes all give toric Gorenstein Fano varieties and the regular ones, smooth toric Fano $n$-folds.}}
  \label{tcyno}
  }
\end{table}

\begin{figure}[t!!]
\begin{center}
\resizebox{0.95\hsize}{!}{
  \includegraphics[trim=0mm 0mm 0mm 0mm, width=9in]{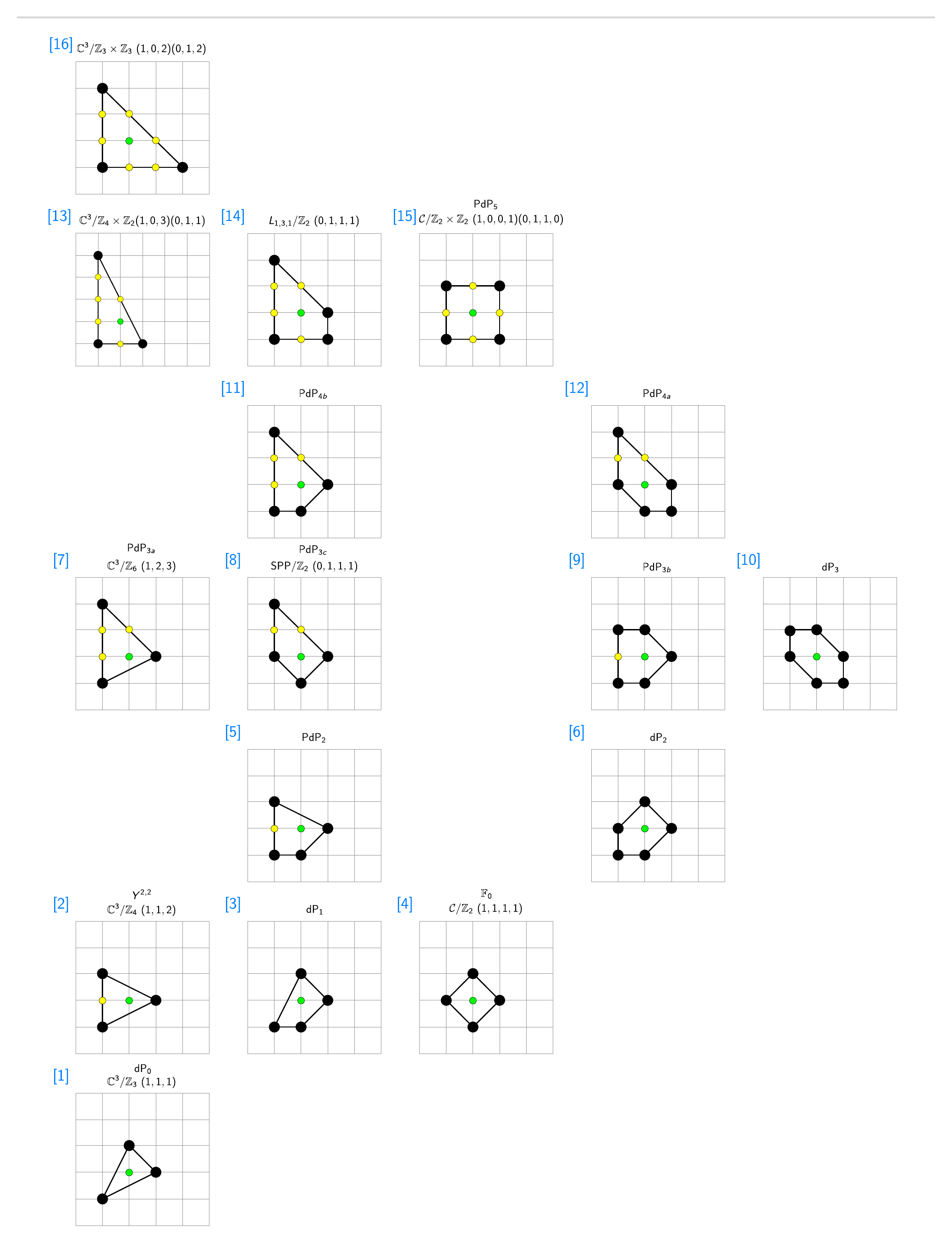}
}
\caption{{\sf {\small
  The 16 inequivalent reflexive polytopes $\Delta_2$ in dimension 2.
  We see that, in particular, this includes the toric del Pezzo surfaces, numbers 1, 4, 3, 6, 10, which are the 5 smooth ones.
  For the naming of these polytopes we refer to the Calabi-Yau cone $\cX = \cC(X(\Delta))$.
  The middle 4 are self-dual while the other 6 polar dual pairs are drawn mirror-symmetrically about this middle line.
}}
\label{f:n=2}}
\end{center}
\end{figure}

Clearly, two polytopes are equivalent if there is a $GL(n;\mathbb{Z})$ matrix relating the vertices of one to the other.
The classification of $GL(2;\IZ)$-inequivalent reflexive polytopes is a non-trivial story.
In dimension 2, there are only 16 such objects and are known classically (cf.~\cite{DW}).
We illustrate these in Figure \ref{f:n=2}, where the polygons have been organized judiciously: the area increases as one goes up and the number of extremal vertices increases as one goes right.
The middle 4 are self-dual while the other dual pairs are drawn mirror-symmetrically about this middle line.
Moreover, the 16 include the 5 toric diagrams for the smooth toric Fano 2-folds, viz., the toric del Pezzo surfaces: $\IP^2$, $\IF_0 = \IP^1 \times \IP^1$, as well as $\IP^2$ blown up at up to 3 generic points, which are shown, respectively, as numbers 1, 4, 3, 6, 10 in the figure.
The remaining 11 are singular but Gorenstein and include orbifolds which can be resolved.
The naming of these polytopes refers to the Calabi-Yau cone $\cX = \cC(X(\Delta))$; for example, number 16 is the $\IZ_3 \times \IZ_3$ orbifold of $\IC^3$.
There is a slight abuse of notation, which context will make clear, that $dP_n$ refers both to the $n$-th del Pezzo surface and the affine Calabi-Yau cone over it.

\begin{figure}[ht!!]
\begin{center}
\resizebox{0.85\hsize}{!}{
  \includegraphics[trim=0mm 0mm 0mm 0mm, width=8in]{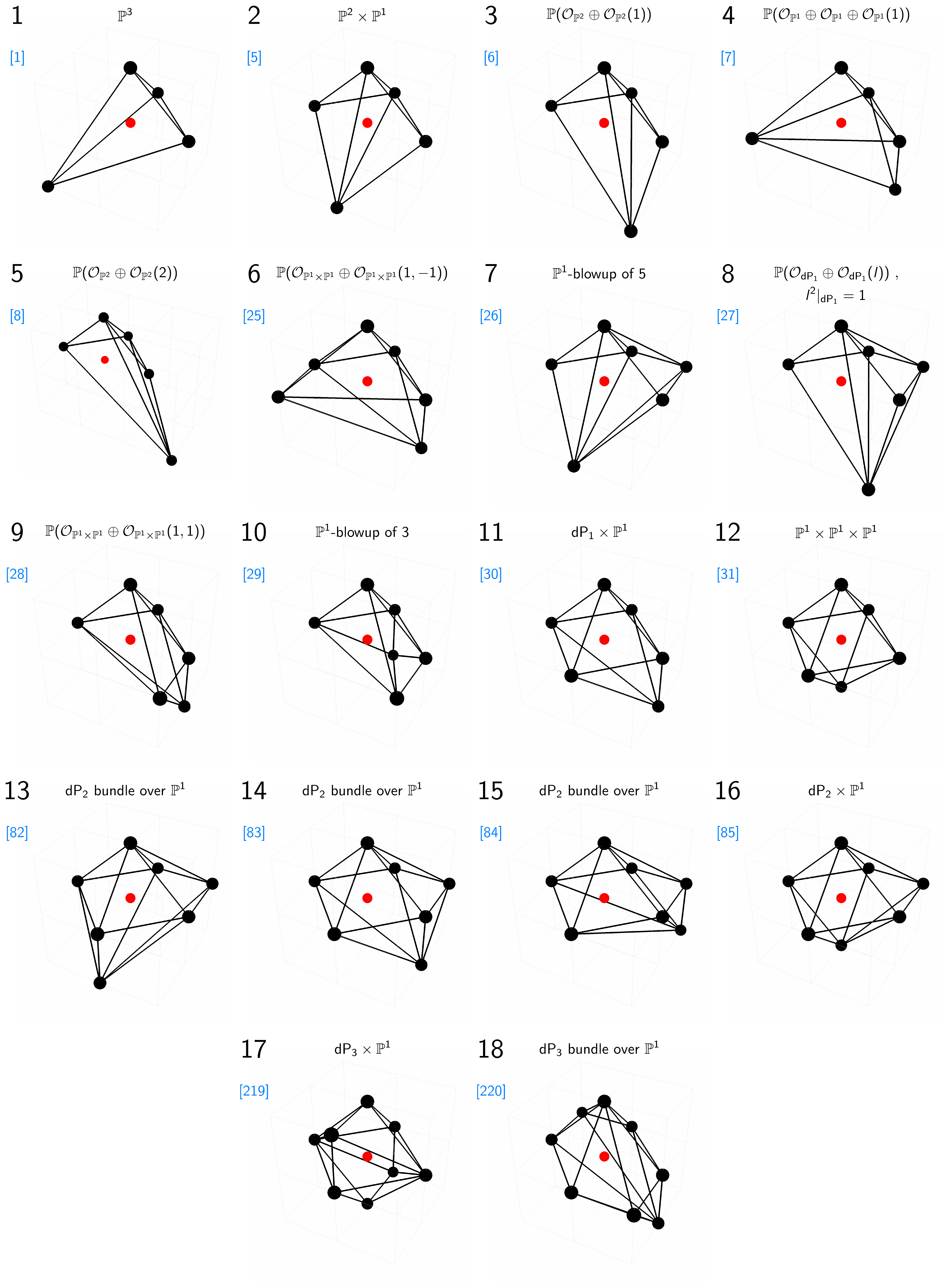}
}
\caption{{\sf {\small
 The 18 reflexive polytopes in dimension 3 corresponding to smooth Fano varieties.
 All of these have a single interior point $(0,0,0)$ which is shown in red.
 The notation $dP_n$ refers to the del Pezzo surface of degree $9-n$, i.e., formed by $\IP^2$ blown up at $n$ generic points. \label{f:n=3smooth}}}}
\end{center}
\end{figure}

In dimension 3, there is already a multitude of reflexive polyhedra.
The classification of these was undertaken in \cite{Kreuzer:1998vb} and 4319 were found of which only 18 are smooth, corresponding to the 18 smooth Fano toric threefolds \cite{bat}.
Of course, it is impossible to draw all 4319 polyhedra here.
For reference, we draw the 18 smooth cases in Figure \ref{f:n=3smooth}.
Moreover, there are 79 which are self-dual.
In the figure and hence forth, the blue labels and identification numbers of the reflexive polyhedra follow the labeling order \footnote{
  Beware that in SAGE, the index of lists begin with 0, thus to access the SAGE data, one actually uses index $0, \ldots, 4318$.
  } of the Sage \cite{sage} database, $1, \ldots, 4319$.
We refer the reader to appendix \ref{appalg} for further discussions on SAGE implementations.

In dimension 4, the number of reflexive $\Delta_4$ increases dramatically to $473,800,776$, the fruit of a tour-de-force computer search by \cite{Kreuzer:1995cd,Kreuzer:2000xy}; of these, 124 give smooth Fano toric 4-folds (and have been explored in the context of heterotic string theory in \cite{He:2009wi}).
In general, it is known that the number of inequivalent reflexive polytopes up to $GL(2;\IZ)$ in each dimension is finite but above $n = 4$, the number is not known.
We summarize the remarkable sequence in \tref{tcyno} (cf.~\cite{oeis}).

\subsection{Calabi-Yau Constructions \label{scy}} 

Stemming from a reflexive polytope are two Calabi-Yau constructions:
\begin{enumerate}
\item
  The theorem of Batyrev-Borisov \cite{bat,BB} and the consequent constructions of Kreuzer-Skarke \cite{Kreuzer:1995cd,Kreuzer:1998vb,Kreuzer:2000xy,Kreuzer:2002uu} is the fact that to a reflexive $\Delta_n$, one can associate a smooth Calabi-Yau $(n-1)$-fold $\cX_{n-1}$ given by the vanishing set of the polynomial
  \footnote{
    The famous quintic in $\IP^4$, for instance, is thus realized in the toric variety $\IP^4$ as follows.
    We have $x_{1,\ldots,5}$ as the (homogeneous) coordinates of $\IP^4$ and can think of the reflexive polytope $\Delta$ as having vertices
    $\underline{m}_1 = (-1,-1,-1,-1)$,
    $\underline{m}_2 = (~4,-1,-1,-1)$,
    $\underline{m}_3 = (-1,~4,-1,-1)$,
    $\underline{m}_4 = (-1,-1,~4,-1)$, and
    $\underline{m}_5 = (-1,-1,-1,~4)$
    as well as all the points interior to these extremal points, including, for example, $(0,0,0,0)$.
    The dual polytope $\Delta^\circ$ is easily checked to have vertices
    $\underline{v}_1 = (1, 0,0,0)$,
    $\underline{v}_2 = (0,1,0,0)$,
    $\underline{v}_3 = (0,0,1,0)$,
    $\underline{v}_4 = (0,0,0,1)$, and
    $\underline{v}_5 = (-1,-1,-1,-1)$.
    Thus each lattice point $\bold{m} \in \Delta$ contributes a quintic monomial in the coordinates $x_{1,\ldots,5}$ to the defining polynomial.
  }
\begin{equation}
\label{deq}
0 = \sum\limits_{\underline{m} \in \Delta} C_{\underline{m}}
\prod\limits_{\rho = 1}^k x_\rho^{(\underline{m} \cdot \underline{v_\rho})  + 1} \ ,
\end{equation}
where $C_{\underline{m}} \in \mathbb{C}$ are numerical coefficients parametrizing the complex structure of $\cX_{n-1}$, the dot being the usual vector dot-product, $x_{\rho = 1, \ldots, k}$ are the projective coordinates of $X(\Delta)$
and $\underline{v}_{\rho = 1, \ldots, k}$ are the vertices of $\Delta^\circ$ ($k$ being the number of vertices of $\Delta^\circ$, or equivalently, the number of facets in the original polytope $\Delta$).

\item
  Alternatively, we can think of the vertices of $\Delta$ as generating a {\it rational polyhedral cone} $\sigma$ with the apex at the origin as discussed in \eqref{es80a2}.
  That is, while our reflexive polytope lives in $\mathbb{Z}^n$, we shall take a point $(0,0,\ldots,0) \in \mathbb{Z}^{n+1} =: N$ and consider the cone generated by the vectors $\underline{u}_i$ from this apex to the vertices of $\Delta$:
\beal{es40a1}
\sigma = \left\{
    \sum\limits_{u_i} \lambda_{i} \underline{u}_i | \lambda_{u_i} \ge 0
    \right\} \subset N_{\mathbb{R}} := N \otimes_{\mathbb{Z}} \mathbb{R} \ .
\eea
  The dual cone then lives in $M_{\mathbb{R}} := M \otimes_{\mathbb{Z}} \mathbb{R}$ with $M := \hom(N, \mathbb{Z})$, as
\beal{es40a2}
    \sigma^\vee = \left\{
    \underline{m} \in M_{\mathbb{R}}
    | \underline{m} \cdot \underline{u} \ge 0 \quad \forall u \in \sigma
    \right\} \ . 
\eea
  Subsequently, we can define an affine algebraic variety $\cX_{n+1}$ associated to this data by taking the maximal spectrum of the group algebra generated by the lattice points which the dual cone covers in $M$: 
  \begin{equation}
    \cX_{n+1} \simeq \mbox{Spec}_{{\rm Max}}\left( \mathbb{C}[\sigma^\vee \cap M]  \right) \ .
  \end{equation}
  The fact that in $N$ the end-points of the vector generators of the cone are co-hyperplanar ensures that $\cX_{n+1}$ is a Gorenstein singularity, and thus admits a resolution to a Calabi-Yau manifold.
\end{enumerate}

\begin{figure}[h!!!]
\begin{center}
\resizebox{1\hsize}{!}{
  \includegraphics[trim=0mm 0mm 0mm 0mm, width=10in]{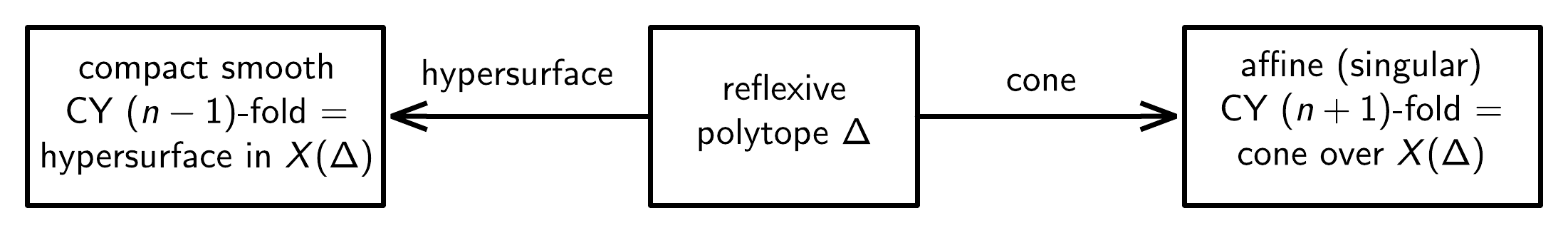}
}
\caption{{\sf {\small
Two different constructions of Calabi-Yau varieties from a reflexive polytope $\Delta_n$. In this work, we will concentrate on non-compact Calabi-Yau cones over $X(\Delta)$.
\label{fsummaryplot}
}}}
\end{center}
\end{figure}

We summarize the above two constructions in \fref{fsummaryplot}.
While there has been a host of activity following Batyrev-Borisov and Kreuzer-Skarke in studying the distribution of the topological quantities of compact Calabi-Yau hypersurfaces \cite{Kreuzer:2002uu,Candelas:2008wb,Anderson:2016cdu,Johnson:2014xpa,Gray:2014fla,Candelas:2016fdy,He:2015fif,Altman:2014bfa,cy3online,oldcy3,Braun:2011ik},
in this paper we will be exclusively concerned with the arrow to the right in the figure and study the affine Calabi-Yau $(n+1)$-fold $\cX$, as a {\it complex} cone over the Gorenstein Fano variety $X(\Delta_n)$ where $\Delta_n$ is reflexive.
For convenience, we will use $\Delta_{n-1}$ so that $\cX$ is of complex dimension $n$.
A key fact is that $\cX$ is itself a {\it real} cone over a Sasaki-Einstein manifold $Y$ of (real) dimension $2n - 1$.
We will study in detail the volumes of these manifolds in relation to the various topological quantities of $X(\Delta)$ and find surprising patterns.

\subsection{Connection to Quiver Gauge Theories \label{ss2}}

\begin{table}[h!!!]
  \centering
  \begin{tabular}{c|c|c|c}
    & Brane Configuration & T-Duality & D-Brane Probe \\ \hline 
    a)
    &
    \begin{tabular}{c|cccccccccc}
      \; & 0 & 1 & 2 & 3 & 4 & 5 & 6 & 7 & 8 & 9 
      \\
      \hline
      D5 & $\times$ & $\times$ & $\times$ & $\times$ &  $\cdot$ &  $\times$ &  $\cdot$ &  $\times$   &  $\cdot$ &  $\cdot$
      \\
      NS5 & $\times$ & $\times$ & $\times$ & $\times$ & \multicolumn{4}{c}{----- \ $\Sigma$ \ -----} &  $\cdot$ &  $\cdot$
    \end{tabular}
    &
    $\stackrel{\mbox{2 times}}{\longleftrightarrow}$
    &
    D3 $\perp$ CY3  
    \\ \hline
    b)
    &
    \hspace{-0.23cm}
    \begin{tabular}{c|cccccccccc}
      \; & 0 & 1 & 2 & 3 & 4 & 5 & 6 & 7 & 8 & 9 
      \\
      \hline
      D4 & $\times$ & $\times$ & $\cdot$ & $\times$ &  $\cdot$ &  $\times$ &  $\cdot$ &  $\times$   &  $\cdot$ &  $\cdot$
      \\
      NS5 & $\times$ & $\times$ & \multicolumn{6}{c}{----------- \ $\Sigma$ \ -----------} &  $\cdot$ &  $\cdot$
    \end{tabular}
    &
    $\stackrel{\mbox{3 times}}{\longleftrightarrow}$
    &
    D1 $\perp$ CY4  
    \\ \hline
    c)
    &
        \hspace{-0.38cm}
    \begin{tabular}{c|cccccccccc}
      \; & 0 & 1 & 2 & 3 & 4 & 5 & 6 & 7 & 8 & 9 
      \\
      \hline
      D3 & $\cdot$ & $\times$ & $\cdot$ & $\times$ &  $\cdot$ &  $\times$ &  $\cdot$ &  $\times$   &  $\cdot$ &  $\cdot$
      \\
      NS5 & \multicolumn{8}{c}{--------------- \ $\Sigma$ \ ---------------} &  $\cdot$ &  $\cdot$
    \end{tabular}
    &
    $\stackrel{\mbox{4 times}}{\longleftrightarrow}$
    &
    D(-1) $\perp$ CY5  
    \\
  \end{tabular}
  \caption{{\sf {\small 
        The various brane configurations for brane tilings and how, under T-duality, they map to D-branes probing affine Calabi-Yau cones in various dimensions.
        (a)
        Brane tilings where D5-branes are suspended between a NS5-brane that wraps a holomorphic surface $\Sigma$. The D5 and NS5-branes meet in a $T^2$ inside $\Sigma$. Under thrice T-duality, the D5-branes are mapped to D3-branes probing CY3;
        (b)
        Brane brick models where D4-branes are suspended between a NS5-brane that wraps a holomorphic 3-cycle $\Sigma$. The D4 and NS5-branes meet in a $T^3$ inside $\Sigma$.
        Under T-duality, the D4-branes are mapped to D1-branes  probing CY4;
        (c)        
        Brane hyper-brick models where Euclidean D3-branes are suspended between a NS5-brane that wraps a holomorphic 4-cycle $\Sigma$. The D3 and NS5-branes meet in a $T^4$ inside $\Sigma$.
        Under T-duality, the D3-branes are mapped to D(-1)-branes probing CY5.        
    }}
  }
\label{t:branes}
\end{table}

 For the last two decades, remarkable progress both in string theory and mathematics has been achieved in connection to the AdS/CFT correspondence \cite{Maldacena:1997re}. Maldacena's discovery that $4d$ $\mathcal{N}=4$ supersymmetric Yang-Mills theory relates to string theory on $AdS_5 \times S^5$ led to numerous discoveries related to conformal field theories and the associated geometries. Soon after the discovery, it was realized \cite{Morrison:1998cs,Acharya:1998db} that the correspondence extends to theories with less supersymmetry. When one replaces the $S^5$ sphere with a Sasaki-Einstein 5-manifold $Y$, the supersymmetry is broken down to $\mathcal{N}=1$. A large class of examples was born: $4d$ $\mathcal{N}=1$ superconformal field theories corresponding to IIB string theory in an $AdS_5 \times Y$ background. Classes of Sasaki-Einstein 5-manifolds such as $Y^{p,q}$ and $L^{p,q,r}$ provide a rich set of theories to study the correspondence.

 These $4d$ $\mathcal{N}=1$ theories can be thought of as worldvolume theories living on D3-branes that probe a Calabi-Yau 3-fold singularity. If the Calabi-Yau 3-fold is toric, the $4d$ $\mathcal{N}=1$ theories can be described in terms of a bipartite periodic graph on a 2-torus known as dimers \cite{1997AnIHP..33..591K,2003math.....10326K} and \textit{brane tilings} \cite{Hanany:2005ve,Franco:2005rj,Feng:2005gw}. These bipartite graphs represent a type IIB brane configuration of D5-branes suspended between a NS5-brane that wraps a holomorphic surface $\Sigma$ originating from the toric geometry. 
 The Newton polynomial $P(x,y)$ of the toric diagram coming from the probed toric Calabi-Yau 3-fold defines the holomorphic surface $\Sigma: P(x,y)=0$.
 This IIB brane configuration is T-dual to the probe D3-branes at the Calabi-Yau singularity.
\tref{t:branes}, part (a) summarizes the type IIB brane configuration. 
Brane tilings provide a natural geometric interpretation of gauge theory phenomena such as Seiberg duality in $4d$ \cite{Seiberg:1994pq,Feng:2000mi,Franco:2005rj}.
  
Non-compact toric Calabi-Yau 4-folds attracted much interest in relation to M-theory. 
When M2-branes probe a toric Calabi-Yau 4-fold, the worldvolume theory of the M2-branes is described by a $3d$ $\mathcal{N}=2$ Chern-Simons theory \cite{Bagger:2006sk,Gustavsson:2007vu,Aharony:2008ug}. 
The probed Calabi-Yau 4-folds have a base which is a Sasaki-Einstein 7-manifold \cite{Martelli:2008rt}. 
A generalized brane tiling that encodes the levels of the $3d$ $\mathcal{N}=2$ Chern-Simons theories was proposed in \cite{Hanany:2008cd,Hanany:2008fj}.

More recently, toric Calabi-Yau 4-folds appeared in the context of \textit{brane brick models} and $2d$ $(0,2)$ quiver gauge theories \cite{Franco:2015tna,Franco:2015tya,Franco:2016qxh}. When D1-branes probe a toric Calabi-Yau 4-fold singularity, the worldvolume theory on the D1-branes is a $2d$ $(0,2)$ theory. 
Brane brick models are brane configurations consisting of D4-branes suspended between a NS5-brane that wraps a holomorphic 3-cycle $\Sigma$. This holomorphic 3-cycle derives from the zero locus of the Newton polynomial of the toric diagram characterizing the toric Calabi-Yau 4-fold. 
The brane configuration is summarized in \tref{t:branes}, part (b). 
In \cite{Franco:2016nwv}, brane brick models were used to provide a brane realization of the recently discovered correspondence between $2d$ $(0,2)$ theories known as Gadde-Gukov-Putrov triality \cite{Gadde:2013lxa}. 

In \cite{Franco:2016tcm}, it was argued from the perspective of the Hori-Vafa mirror of a Calabi-Yau $n$-fold \cite{Hori:2000kt},
that there exists a natural generalization of Seiberg duality in $4d$ and Gadde-Gukov-Putrov triality in $2d$ for theories in $0d$. A new correspondence, called quadrality, relates different $0d$ $\mathcal{N}=1$ gauged matrix models that arise as worldvolume theories of D(-1)-branes probing non-compact toric Calabi-Yau 5-folds. The brane realization is known as \textit{brane hyper-brick models} \cite{Franco:2016tcm}; cf.~\tref{t:branes}, part (c).

In summary, toric Calabi-Yau $n$-folds have appeared in connection to quiver gauge theories in various dimensions.
From this panorama of correspondences, it is tempting to ask whether there are any further geometrical features of Calabi-Yau $n$-folds that can teach us more about quiver gauge theories in various dimensions.
In order to search for these properties, we restrict ourselves to toric Calabi-Yau $n$-folds that exist in every dimension $n \ge 1$, namely the ones which arise from reflexive polytopes: $\cX_n = \cC(X(\Delta_{n-1}))$. This class of toric Calabi-Yau $n$-folds is by far the largest family studied in the literature and is a suitable sample set for our analysis in this paper.

Our strategy is to study the connection between the topology and metric geometry of the Calabi-Yau variety $\cX$.
The toric geometry of the base $X(\Delta)$ -- the compact toric Gorenstein Fano $(n-1)$-fold --  will furnish us with the convenience of extracting the topological data while the Sasaki-Einstein geometry of the base $Y$ will give us the relevant metrical information.

\section{Topological Quantities of $X(\Delta)$ \label{stop}}
We begin by concentrating on the compact toric variety $X(\Delta)$.
These, as discussed in sections \sref{s:toric} and \sref{scy}, can be constructed from a reflexive lattice polytope $\Delta$. The following section describes topological quantities that can be computed from $\Delta$ such as the Euler number and Chern classes\footnote{
  Note that we are not studying the orbifold or stringy topological quantities associated to the singular variety as has been done in \cite{Batyrev:2016vrl,Dixon:1985jw,gottsche}.
}
of the complete resolution $\widetilde{X(\Delta)}$.

\subsection{Triangulations and the Euler Number \label{stopo}}
For the case when the toric variety $X(\Delta)$ is smooth, it is straightforward to compute its Betti numbers from standard formulae \cite{fulton}.
Unfortunately, as mentioned in section \sref{s:toric}, when $\Delta$ is reflexive, the corresponding toric varieties $X(\Delta)$ are not guaranteed to be smooth.
In fact, only a tiny fraction of the reflexive polytopes actually give directly smooth $X(\Delta)$ corresponding to the regular polytopes, as we saw in Table \ref{tcyno}.

We are therefore compelled to consider resolutions under triangulations of $\Delta$. We concentrate on triangulations of $\Delta$ which guarantee that each cone in $\Delta$ is regular.
For reflexive polytopes $\Delta$ , we consider a special type of triangulation known as {\bf FRS triangulations} \cite{Altman:2014bfa}.
FRS stands for the following:
\begin{itemize}
\item \textit{Fine:} all lattice points of $\Delta$ are involved in the triangulation;
\item \textit{Regular:} as discussed in Theorem \ref{thm:delta1}, ensures that $\Delta$ is regular and $X(\Delta)$ is smooth;
\item \textit{Star:} the origin, which is taken to be the interior lattice point of the polytope $\Delta$, is the apex of all the triangulated cones.
\end{itemize}
We point out that our requirement of {\it regular} is actually stronger than was used in \cite{Altman:2014bfa} and leave a detailed discussion of this as well as algorithmic implementation to appendix \sref{appalg}.
Note that for regularity to hold, the generators of the simplicial cones should give determinants $\pm 1$.

Under FRS triangulation of $\Delta$, we achieve a complete resolution $\widetilde{X(\Delta)}$ whose $n$-cones are all regular.  
Note that all cones in $\widetilde{X(\Delta)}$ under the FRS triangulation of a reflexive polytope $\Delta$ are not only regular but also simplicial.
The FRS triangulation of a reflexive polytope can be obtained using a simple but efficient trick from \cite{Braun:2011ik,Long:2014fba,Altman:2014bfa}. One first ignores the single interior point that forms the origin of all the simplicial cones and applies a triangulation of the boundary points of the polytope.
When one then adds lines from the boundary triangles to the origin at the interior point, one obtains the simplicial cones that form the FRS triangulation of $\Delta$.

Now, it is an important result \cite{bat,nill} that
\begin{theorem}
For $n \leq 3$, any $n$-dimensional Gorenstein toric Fano variety admits a coherent crepant resolution.
\end{theorem}
What this means is that our 16 polygons and 4319 polyhedra all have FRS triangulations that completely resolve the singularity, as we can explicitly see case by case.
However, for higher dimensions, we are no longer guaranteed that a smooth $\widetilde{X(\Delta_{n \ge 4})}$ actually exists.
Henceforth, for $n = 4$, we will restrict to cases where a complete desingularization is possible.

Under regularity and smoothness, the Betti numbers $b_i$ and the Euler number $\chi$ can be readily obtained combinatorially \cite{fulton,CLS}:
\begin{theorem}\label{betti}
After performing FRS triangulation $\widetilde{\Delta_n}$ on the reflexive polytope $\Delta_n$, the smooth, compact toric Fano $n$-fold $\widetilde{X(\Delta_n)}$ has Betti numbers
\beal{es80a10}
b_{2k-1} = 0  ~,~
\qquad
b_{2k} = \sum^{n}_{i=k} (-1)^{i-k} \left(\ba{c} i \\ k \ea\right) d_{n_i} ~,~
\eea
where $k=0,1,\dots, n$ and $d_j$ is the number of $j$-dimensional cones in $\widetilde{\Delta}$.
In particular, the Euler number is
\beal{es80a15}
\chi(\widetilde{X(\Delta)}) = d_n \ .
\eea
\end{theorem}

The above formula in fact even further simplifies for reflexive polytopes in dimensions $n=2,3$.
For $n=2$, the number of facets in $\widetilde{X(\Delta)}$ is simply given by the number of perimeter lattice points, $p$.
Note that we distinguish between {\it vertices} and {\it perimeter} lattice points.
Take the example of reflexive polygon number 16 in Figure \ref{f:n=2}, which has 3 vertices and 9 perimeter lattice points.
Hence, we have that
\begin{corollary}
  For all 16 reflexive polygons $\Delta$ in dimension 2, the smooth toric variety $\widetilde{X(\Delta)}$ corresponding to the complete desingularization by FRS triangulation, has Euler number equal to the number of perimeter lattice points, which is in turn the total number of lattice points of $\Delta$ minus 1 (for the unique interior lattice point):
  \[
  \chi(\widetilde{X(\Delta_2)}) = p = (\Delta_2 \cap \IZ^2) - 1 \ .
  \]
\end{corollary}
\paragraph{Remark: }
One can equally arrive at this by Pick's formula, that for convex polygons $\Delta_2$, the area $A$ is given by
\begin{equation}
A(\Delta_2) = i + \frac{p}{2} - 1 \ ,
\end{equation}
where $i$ is the number of interior lattice points, and $p$, that of perimeter (or boundary) lattice points. Here $i=1$, corresponding to the origin and an FRS triangulation means each of the triangles formed by neighboring points on the perimeter is of area $1/2$, hence the number of faces is just $p$.

It is of course helpful to have the topological quantities as simply as the $n=2$ case for higher dimensions.
However, generalizing Pick's formula for higher dimensions is not so straightforward.
The standard method is via Ehrhart polynomials \cite{CLS,DR}.
\begin{theorem} [Ehrhart]
  The generating function in $t \in \IZ$ for an $n$-dimensional polytope $\Delta$ is defined as
  \[
  L(\Delta; t) := \mbox{number of lattice points of } \{t \Delta\} \cap \IZ^n
  \]
  where $t \Delta$ is simply $\Delta$ dilated by a factor of $t$.
  Then for $t \in \IN$, $L(\Delta; t)$ is a polynomial in $t$ known as the Ehrhart polynomial with expansion
  \[
  L(\Delta; t) = \mbox{Vol}(\Delta) t^n + \ldots + \chi_{\Delta} \ .
  \]
\end{theorem}
Thus $L(\Delta; t)$ has its leading term capturing the volume of the polytope and its constant, the combinatorial Euler number $\chi_\Delta$. Note that $\chi_\Delta$ has nothing to do with the algebro-geometric Euler number $\chi(\widetilde{X(\Delta)})$.
For $n=2$, the above reduces to Pick's formula. For $n=3$, because in a FRS triangulation each tetrahedron is of volume $\frac{1}{3!}$, the number of tetrahedra is simply $\text{Vol}(\Delta_3) / (1/3!) = 6\mbox{Vol}(\Delta_3)$.
As a result, we find that the Euler number is linear in the number of total integer points in $\Delta_3$ and that in particular, it is twice the latter subtracted by 6:
\begin{corollary}\label{chi3}
  For all the 4319 reflexive polyhedra in dimension $n=3$, the Euler number of the complete resolution $\widetilde{X(\Delta_3)}$ is
  \begin{align*}
  \chi(\widetilde{X(\Delta_3)}) & = \#(\mbox{tetrahedra in any FRS triangulation}) = 6 Vol(\Delta_3) \\
  & = 2 \#(\Delta_3 \cap \IZ^3) - 6
  = 2 p - 4 \ ,
  \end{align*}
  where $Vol(\Delta_3)$ is the leading term in the Ehrhart polynomial, giving the volume of the polytope
  and $\#(\Delta_3 \cap \IZ^3)$ is the number of lattice points in the polytope.
\end{corollary}
In the last step, analogous to the dimension 2 case, we define $p$ as the number of ``perimeter'' points which includes the vertices as well as any lattice points living on any edge or any face of $\Delta_3$, since there is only a single interior point.

We will make use of these topological results for toric $2$, $3$ and $4$-folds in the following sections.

\subsection{Chern Numbers}\label{s:chern}

Having studied the Euler characteristic, we now turn to the Chern classes.

\subsubsection{Dimension 2: Reflexive Polygons}

In dimension 2, there are two numbers, customarily called the {\bf Chern numbers}, which can be obtained naturally by integration over the toric 2-fold.
Again, to ensure smoothness and compactness so that the Chern classes are well-defined, we work over the complete resolution $M_2 := \widetilde{X(\Delta_2)}$:
\begin{equation}
  M_2 := \widetilde{X(\Delta_2)} \ ;
  \qquad
  C(M_2) = \int_{M_2} c_1(M_2)^2 \ , \qquad
  \chi(M_2) = \int_{M_2} c_2(M_2) \ .
\end{equation}
The integral over the surface of the top Chern class is, of course, none other than the Euler characteristic.
By direct computation, we readily find that $C$ is simply the number of perimeter lattice points of the polar dual of $\Delta$.
This can be thought of as a version of Batyrev's mirror symmetry for Fano 2-folds:
\begin{proposition}\label{n=2chern}
  For the 16 reflexive polygons in dimension 2, the two independent Chern numbers of the complete resolution $M_2 := \widetilde{X(\Delta_2)}$ are
  \[
  \Big(C(M_2), \chi(M_2)\Big) = (p^\circ, p)
  \]
  where $p$ and $p^\circ$ are  the number of lattice points on the perimeter of respectively $\Delta$ and the polar dual $\Delta^\circ$. 
\end{proposition}

\subsubsection{Dimension 3: Reflexive Polyhedra}

\begin{figure}[h!!!]
\begin{center}
\resizebox{1\hsize}{!}{
  \includegraphics[trim=0mm 0mm 0mm 0mm, width=8in]{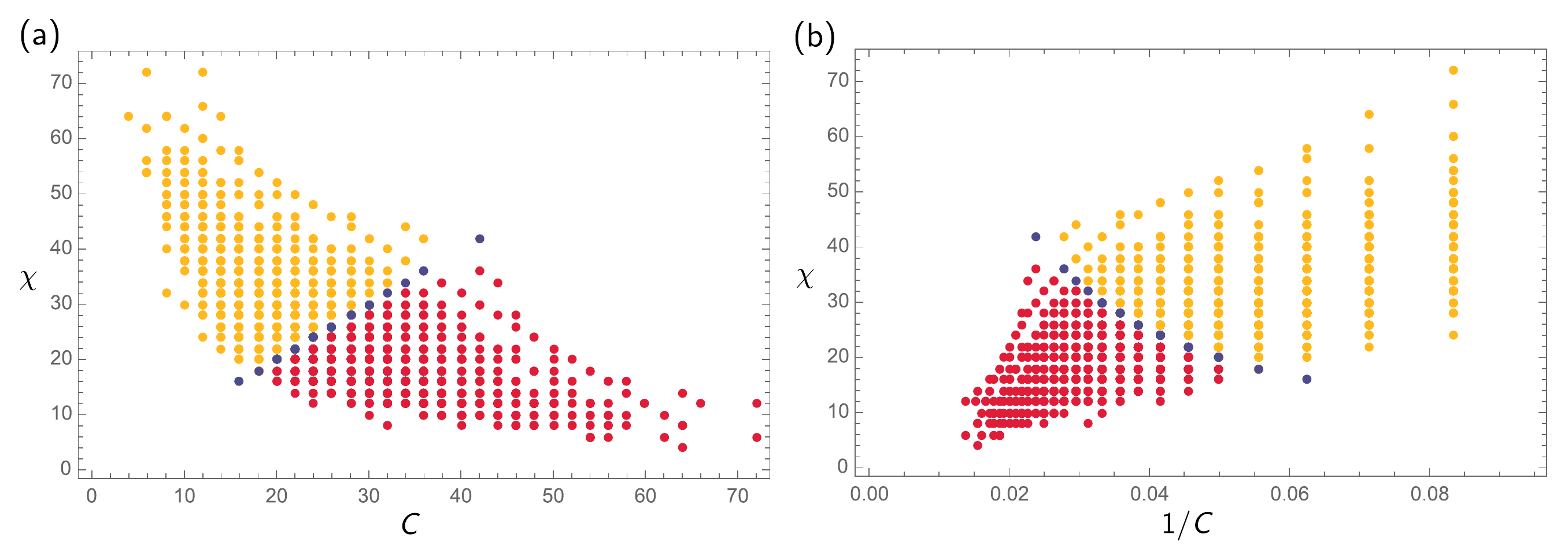}
  }
\caption{{\sf {\small
      (a) The plot of the two non-trivial Chern numbers $C = \int_{M_3} c_1(M_3)^3$ versus the Euler number $\chi =   \int_{M_3} c_3(M_3)$ of the complete desingularization of the toric variety $M_3 = \widetilde{X(\Delta_3)}$ for the 4319 reflexive polyhedra $\Delta_3$. The red points correspond to $\Delta_3$ which are dual polytopes of the ones in yellow. Self-dual $\Delta_3$ are in blue.
      (b) For reference, we include the plot also for the reciprocal of $C$ against $\chi$.
\label{f:chern3d}}}}
\end{center}
\end{figure}

In this subsection, we move onto the reflexive polyhedra in dimension 3.
We have the smooth Fano variety $M_3 := \widetilde{X(\Delta_3)}$ from the FRS triangulation of $\Delta_3$.
We can now form three natural integrals from the Chern classes $c_i(M_3) \in H^{2i}(M_3;\IZ)$, viz., 
$\int_{M_3} c_1(M_3)^3$,
$\int_{M_3} c_1(M_3) c_2(M_3)$ and
$\int_{M_3} c_3(M_3) \in \IZ$.
We can easily see that these are independent and are the only three non-trivially-vanishing classes which can be integrated over the Fano threefold to produce an integer.
Again the third integral is the Euler characteristic, as discussed in Corollary \ref{chi3} and it remains to study the first two integrals.

We can quite efficiently compute these quantities. First off, the second integral is equal to 24 for all 4319 reflexive polytopes in $n=3$, by direct computation \cite{sage}.
The remaining two Chern numbers $C := \int_{M_3} c_1(M_3)^3$ and $\chi$, on the other hand, are intricately related and we show the plot between them in \fref{f:chern3d}.
We see many bands of vertical lines and to ensure these are not due to any round-off errors, we also plot against the reciprocal of the Euler characteristic.
Inspired by Proposition \ref{n=2chern}, we compute this pair of Chern numbers for the dual pairs of reflexive polytopes explicitly and find that they are indeed reversed under ``mirror symmetry''.
Recalling further that the top (3rd) Todd class $\td_3(M_3) = \frac{1}{24}c_1(M_3) c_2(M_3)$, we summarize these results in
\begin{proposition}\label{n=3chern}
  Defining the 3 independent Chern numbers
  \[
  C := \int_{M_3} c_1(M_3)^3 \ ,
  \quad
  T = \int_{M_3} \td_3(M_3) = \frac{1}{24}c_1(M_3) c_2(M_3) \ ,
  \quad
  \chi = \int_{M_3} c_3(M_3)
  \]
  for the complete resolution $M_3 :=\widetilde{X(\Delta_3)}$ for the reflexive polyhedron $\Delta_3$ in dimension 3, we have that
  \begin{itemize}
  \item $T = 1$ for all 4319 cases;
  \item $(C,\chi) =
    \Big(2 \#(\Delta_3^\circ \cap \IZ^3) - 6, \quad
    2 \#(\Delta_3 \cap \IZ^3) - 6 \Big)
    = (2p^\circ-4, \ 2p-4)$
  \end{itemize}
  where $\Delta^\circ$ is the polar dual polyhedron, and $p$ and $p^\circ$ are the number of perimeters lattice points of $\Delta_3$ and $\Delta_3^\circ$ respectively.
\end{proposition}

As before, the perimeter points include the vertices as well as any lattice points living on any edge or any face of $\Delta_3$.
Therefore, we see that once again, while the middle Chern number remains fixed here at 1, the bottom and top Chern numbers are reversed for polar dual pairs.

\section{Hilbert Series and Volume Functions \label{shsvol}}

One of the most important quantities which characterize an algebraic variety $X$ is the Hilbert series.\footnote{Note, however, that the Hilbert series is not a topological invariant and does depend on embedding and choice of grading for the coordinate ring for $X$.} It has been used in numerous ways to enumerate operators in a chiral ring and to study the geometric structure they form. As part of the plethystic program \cite{Benvenuti:2006qr,Feng:2007ur}, Hilbert series has been used extensively to address the problem of counting gauge invariant operators and of characterizing moduli spaces of various different supersymmetric gauge theories.  

We will be concentrating on a specific application \cite{Martelli:2005tp,Martelli:2006yb,Bergman:2001qi} of the Hilbert series in relation to the volume of the Sasaki-Einstein base $Y$ of the toric Calabi-Yau cone $\cX = \cC(X(\Delta))$.
Let us first refresh our memory on the Hilbert series and thence explore the connection between Hilbert series and the volume function of a Sasaki-Einstein manifold.

\subsection{Hilbert series \label{shs}}

For a projective variety $X$ over which $\cX$ is a cone, realized as an affine variety in $\IC^k$, the Hilbert series is the generating function for the dimension of the graded pieces of the coordinate ring $\mathbb{C}[x_1,...,x_k] / \left< f_i \right>$ where $f_i$ are the defining polynomials of $X$:
\begin{equation}
g(t; \cX) = \sum\limits_{i=0}^{\infty} ~\dim_{\mathbb{C}}(X_i) ~t^i \ ,
\end{equation}
which is always a rational function.
Here, the $i$-th graded piece $X_i$ can be thought of as the number of algebraically independent degree $i$ polynomials on the variety $X$. $t$ keeps track of the degree $i$. For multi-graded rings with pieces $X_{\vec{i}}$ and grading $\vec{i}=(i_1,\dots,i_k)$, the Hilbert series takes the form $g(t_1,\dots,t_k; \cX) = \sum\limits_{\vec{i}=0}^{\infty} \dim_{\mathbb{C}}(X_{\vec{i}} ) t_1^{i_1}\dots t_k^{i_k}$.
Note that the Hilbert series can appear in two kinds, with as many $t$ variables as the ambient space as given above, or being fewer, as the dimension of $\cX$; we will use the latter below.

If the variety $\cX$ corresponds to the chiral ring of a supersymmetric gauge theory, the Hilbert series can be thought of as the generating function of gauge invariant operators on the chiral ring. The grading can be chosen as such that fugacities $t_1^{i_1},\dots, t_k^{i_k}$ count charges of the operators under the full global symmetry of the supersymmetric gauge theory. 
Finally, if the variety $X$ is toric, the Hilbert series can be conveniently obtained from the toric diagram itself as explained in the ensuing subsection.

\subsubsection{Hilbert Series from Toric Geometry}

\begin{figure}[h!!!]
\begin{center}
\resizebox{0.7\hsize}{!}{
  \includegraphics[trim=0mm 0mm 0mm 0mm, width=8in]{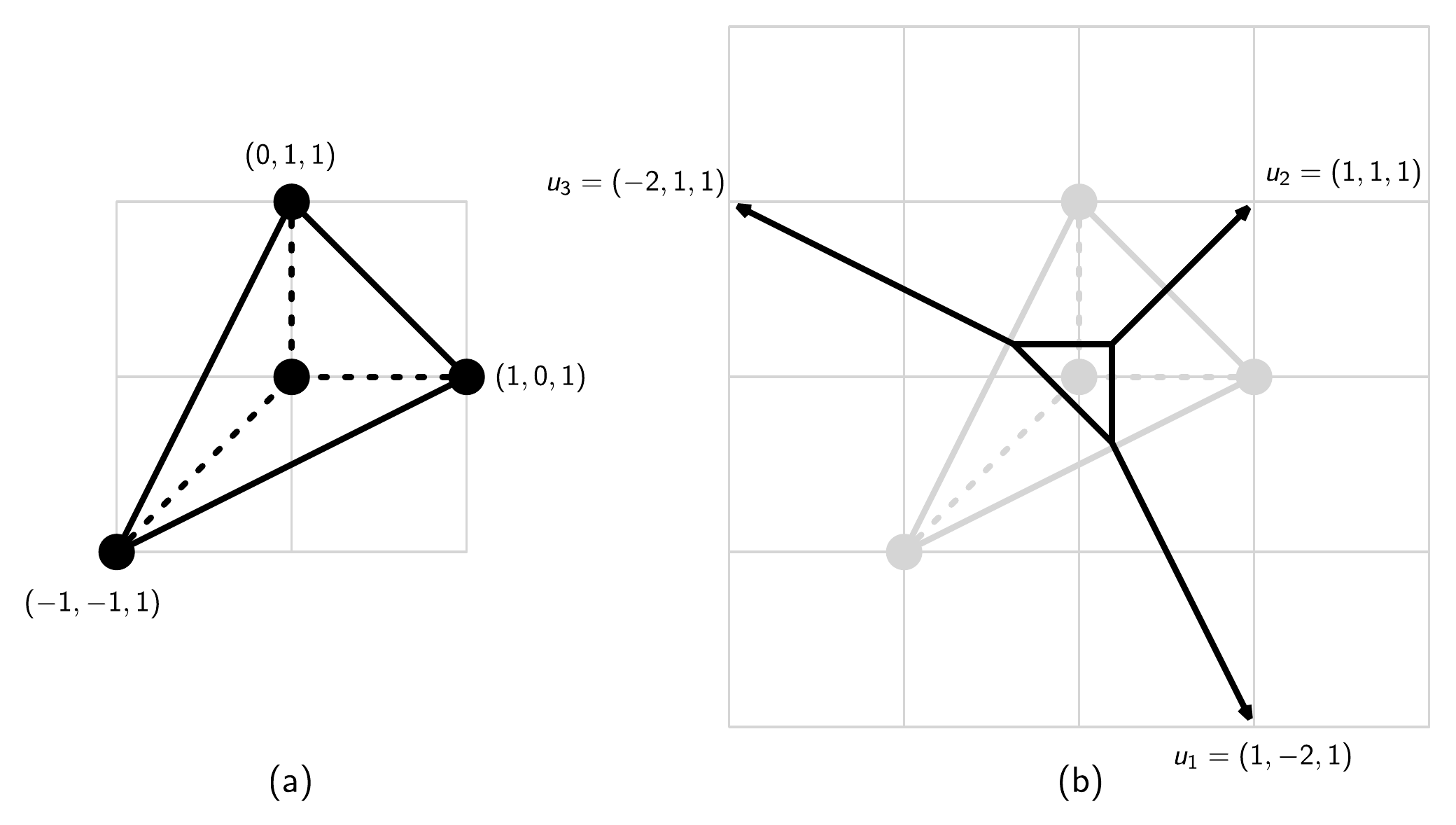}
}
\caption{{\sf {\small
      (a) Toric diagram and (b) dual $(p,q)$-web diagram for $\mathbb{C}^3/\mathbb{Z}_3$.
      The vectors $u_{1,2,3}$ are the outer normals to the fan and furnish the $(p,q)$-charges.
\label{fdp3pq}
}}}
\end{center}
\end{figure}

When the affine variety is toric, the corresponding Hilbert series can be computed directly from the toric diagram \cite{Martelli:2005tp,Benvenuti:2006qr}.
As discussed in section \sref{scy}, the affine toric Calabi-Yau $n$-fold $\cX$ can be defined as a cone over a compact Sasaki-Einstein manifold of real dimension $2n-1$.
The toric diagram of $\cX$ is a convex reflexive lattice polytope of (real) dimension $n-1$, being effectively dimension one less because of the Calabi-Yau condition which ensures that the endpoints of the  generators of the cone are co-hyperplanar.

For $n=3,4$, our toric diagrams $\Delta_{2,3}$ admit at least one FRS triangulation in terms of $(n-1)$-simplices and for $n=5$, we will only consider those which do have such a triangulation.
From an FRS triangulation, one can construct a graph-dual, which for Calabi-Yau 3-folds is known as a $(p,q)$-web diagram.
\fref{fdp3pq} illustrates dual web-diagrams with the corresponding toric diagrams for $\mathbb{C}^3/\mathbb{Z}_3$.

In general, we have that \cite{Martelli:2005tp,Martelli:2006yb}
\begin{theorem}
The Hilbert series of the toric Calabi-Yau $n$-fold cone $\cX$ (and that of a toric variety $\widetilde{X(\Delta_{n-1})}$) is succinctly obtained from the triangulation of $\Delta_{n-1}$ as follows 
\beal{es50a1}
\fbox{$
  g(t_1,\dots,t_n ; \cX) =
  \sum\limits_{i=1}^{r} \prod\limits_{j=1}^{n} (1-\vec{t}^{~\vec{u}_{i,j}})^{-1}
$}
 ~.~
\eea
Here, the index $i=1,\dots,r$ runs over the $n-1$-dimensional simplices in the FRS triangulation and $j=1,\dots,n$ runs over the faces of each such simplex.
Each $\vec{u}_{i,j}$ is an integer $n$-vector, being the outer normal to the $j$-th face of the fan associated to $i$-th simplex. 
$\vec{t}^{~\vec{u}_{i,j}} := \prod\limits_{a=1}^{n} t_a^{u_{i,j}(a)}$ multiplied over the $a$-th component of $\vec{u}$.
\end{theorem}
Note that we introduced an $n$-dimensional origin $(0,0,\ldots,0)$ which is distance 1 away from the hyperplane of the $(n-1)$-dimensional toric diagram.
This sets the $n$-th coordinate of each vector $\vec{u}_{i,j}$ to be 1.
The grading of the Hilbert series in terms of $\vec{t}$ is such that each fugacity $t_a$ corresponds to the $a$-th coordinate of the $\vec{u}_{i,j}$.
We emphasize the fact that all subsequent discussions on volume functions are independent of (1) the $GL(n;\IZ)$ equivalence in the definition of the toric diagram and (2) the different FRS triangulations.

\subsubsection{Hilbert Series for $\mathbb{C}^n/\mathbb{Z}_n$}

Let us consider a few examples of Hilbert series obtained from toric diagrams of toric Calabi-Yau $n$-folds. We consider here Abelian orbifolds of the form $\mathbb{C}^n/\mathbb{Z}_n$ with orbifold action $(1,\dots,1)$ for $n=3,4,5$. The coordinates for vertices in the toric diagram are
\beal{es51a1}
\mathbb{C}^{3}/\mathbb{Z}_3 &~:~&
{\small
\left\{
\begin{array}{c}
  (1,0,1), \\ (0,1,1),\\ (-1,-1,1),\\(0,0,1)
\end{array}
\right.
}
\qquad
\mathbb{C}^{4}/\mathbb{Z}_4 ~:~
{\small
\left\{         
\begin{array}{c}
  (1,0,0,1),\\(0,1,0,1),\\(0,0,1,1),\\(-1,-1,-1,1),\\(0,0,0,1)
\end{array}
\right.
}
\nn\\
\mathbb{C}^{5}/\mathbb{Z}_5 &~:~&
{\small
  \left\{
\begin{array}{c}
  (1,0,0,0,1),\\(0,1,0,0,1),\\(0,0,1,0,1),\\(0,0,0,1,1),\\
  (-1,-1,-1,-1,1),\\(0,0,0,0,1)
\end{array}
\right.
}
\eea
Above, the coordinates of the vertices include the $n$-th coordinate, which will become important in the discussion that follows.
We label the points in the toric diagrams from $1,\dots,n$ in the order they are listed in \eref{es51a1}.
Following this labelling, the unique triangulations of these diagrams can be summarized in terms of the points in the toric diagram as follows (called {\bf incidence data}):
\beal{es51a2}
\mathbb{C}^{3}/\mathbb{Z}_3 &~:~&
\{\{1,2,4\},\{1,3,4\}\,\{2,3,4\}\} 
\nn\\
\mathbb{C}^{4}/\mathbb{Z}_4 &~:~&
\{\{1, 2, 3, 5\}, \{1, 2, 4, 5\}, \{1, 3, 4, 5\}, \{2, 3, 4, 5\}\}
\nn\\
\mathbb{C}^{5}/\mathbb{Z}_5 &~:~&
\{\{1, 2, 3, 4, 6\}, \{1, 2, 3, 5, 6\}, \{1, 2, 4, 5, 6\}, \{1, 3, 4, 5, 6\}, \{2, 3, 4, 5, 6\}\}
\nn\\
\eea
Following \eref{es50a1}, we find the Hilbert series for $\mathbb{C}^3/\mathbb{Z}_3$, $\mathbb{C}^4/\mathbb{Z}_4$ and $\mathbb{C}^5/\mathbb{Z}_5$,
\beal{es51a3}
g(t_i;\mathbb{C}^3/\mathbb{Z}_3) 
&=&
\frac{1}{
(1 - t_2) (1 - t_1^{-1} t_2) (1 - t_1 t_2^{-2} t_3^{-1})
} 
+ 
\frac{1}{
(1 - t_1) (1 - t_1 t_2^{-1}) (1 - t_1^{-2} t_2 t_3)
} 
\nn\\
&&
+ 
\frac{1}{
(1 - t_1^{-1}) (1 - t_2^{-1}) (1 - t_1 t_2 t_3^{-1})
}
\nn\\
g(t_i;\mathbb{C}^4/\mathbb{Z}_4)
&=&
\frac{1}{
(1 - t_1) (1 - t_2) (1 - t_3) (1 - t_1^{-1} t_2^{-1} t_3^{-1} t_4)
} 
\nn\\
&&
+ 
\frac{1}{
(1 - t_1^{-1}) (1 - t_1^{-1} t_2) (1 - t_1^{-1} t_3) (1 - t_1^3 t_2^{-1} t_{3}^{-1} t_4)
}
\nn\\
&&
+ 
\frac{1}{
(1 - t_2^{-1}) (1 - t_1 t_2^{-1}) (1 - t_2^{-1} t_3) (1 - t_1^{-1} t_2^3 t_3^{-1} t_4)
}
\nn\\
&&
+ 
\frac{1}{
(1 - t_3^{-1}) (1 - t_1 t_3^{-1}) (1 - t_2 t_3^{-1}) (1 - t_1^{-1} t_2^{-1} t_3^3 t_4)
}
\nn\\
g(t_i;\mathbb{C}^5/\mathbb{Z}_5)
&=&
\frac{1}{
(1 - t_4) (1 - t_1^{-1} t_4) (1 - t_2^{-1} t_4) (1 - t_3^{-1} t_4) (1 - t_1 t_2 t_3 t_4^{-4} t_5^{-1})
}
\nn\\
&&
+ 
\frac{1}{
(1 - t_3) (1 - t_1^{-1} t_3) (1 - t_2^{-1} t_3) (1 - t_3 t_4^{-1}) (1 - t_1 t_2 t_3^{-4} t_4 t_5^{-1})
} 
\nn\\
&&
+ 
\frac{1}{
(1 - t_2) (1 - t_1^{-1} t_2) (1 - t_2 t_3^{-1}) (1 - t_2 t_4^{-1}) (1 - t_1 t_2^{-4} t_3 t_4 t_5^{-1}) 
}
\nn\\
&&
+ 
\frac{1}{
(1 - t_1) (1 - t_1 t_2^{-1}) (1 - t_1 t_3^{-1}) (1 - t_1 t_4^{-1}) (1 - t_1^{-4} t_2 t_3 t_4 t_5^{-1})
}
\nn\\
&&
+ \frac{1}{
(1 - t_1^{-1}) (1 - t_2^{-1}) (1 - t_3^{-1}) (1 - t_4^{-1}) (1 - t_1 t_2 t_3 t_4 t_5^{-1})
}
~.~
\eea
By setting $t_1,\dots,t_{d-1}=1$ and $t_d=t$, the above Hilbert series can be unrefined to give
\beal{es51a3b}
g(t;\mathbb{C}^3/\mathbb{Z}_3) &=&
\frac{1 + 7 t + t^2}{(1 - t)^3}  ~,~
\nn\\
g(t;\mathbb{C}^4/\mathbb{Z}_4) &=&
\frac{1 + 31 t + 31 t^2 + t^3}{(1 - t)^4} ~,~
\nn\\
g(t;\mathbb{C}^5/\mathbb{Z}_5) &=&
\frac{1 + 121 t + 381 t^2 + 121 t^3 + t^4}{(1 - t)^5} ~.~
\eea
Notice that the numerators of the rational functions in \eref{es51a3b} are all palindromic. According to Stanley \cite{stanley}, this indicates that the corresponding variety of the Hilbert series is Calabi-Yau (Gorenstein).

\subsection{The Calabi-Yau Cone $\cX$ and Sasaki-Einstein Base $Y$ \label{sse}} 
The Hilbert series as computed above, exploiting the relation between the affine variety $\cX$ and the compact real manifold $Y$, has an extraordinary manifestation as a normalized volume. In this section, we briefly summarize this property following the exposition and notation of \cite{Martelli:2005tp,Martelli:2006yb}. 

Consider a K\"ahler cone $\cX$ which has complex dimension $n$ with K\"ahler form $\omega$.
The metric on the cone takes the form 
\beal{es55a1}
\ud s^2(\cX) = \ud r^2 + r^2 \ud s^2(Y) ~,~
\eea
where $Y= \cX|_{r=1}$ is the Sasakian base manifold over which $\cX$ is a real cone.
The K\"ahler form $\omega$ can be written, for $\eta$ a global one-form on $Y$, as
\beal{es55a2}
\omega = - \frac{1}{2} \ud(r^2 \eta) = \frac{1}{2} i \partial \overline{\partial} r^2 ~.~
\eea
Now, $Y$ has a Killing vector field $K$ called the {\bf Reeb vector}, defined, for the complex structure $\mathcal{I}$ on $\cX$, as
\beal{es55a3b}
K := \mathcal{I} \left( 
r \frac{\partial}{\partial r}
\right) ~.~
\eea

As discussed in section \sref{scy}, we require the Calabi-Yau cone $\cX$ to be toric, meaning that we have a torus action $\mathbb{T}^n$ that leaves the K\"ahler form $\omega$ invariant.
We take $\partial /\partial \phi_i$ to be the generators of the torus action and $\phi$ to be the angular coordinates of the action with $\phi \sim \phi + 2\pi$.
The torus action is integrable and allows for the introduction of symplectic coordinates $y_i$ defined as 
\beal{es55a4}
y_i := -  \frac{1}{2} \langle r^2 \eta , \frac{\partial}{\partial \phi_i } \rangle \ ,
\eea
with $\vev{~,~}$ the usual bilinear pairing between forms and vector fields.
Using these symplectic coordinates, the K\"ahler cone can be expressed as a torus fibration over a convex rational polyhedral cone of the form
\beal{es55a5}
\sigma = \{
y \in \mathbb{R}^n ~|~
l_a (y) := (y,v_a) \leq 0 ~,~
a=1,\dots,d
\} ~,~
\eea
where the linear function $l_a (y)$ is defined in terms of the inner product $(~,~)$ between $y$ and $v_a$ which are the inward pointing normal vectors to the $d$ facets of the polyhedral cone.
This cone $\sigma$ is precisely the rational polyhedral cone discussed in \eref{es40a1}.

In terms of the symplectic and angular coordinates, the K\"ahler form on $\cX$ is
\beal{es55a10}
\omega = \ud y_i \wedge \ud \phi_i ~,~
\eea
and the K\"ahler metric on $\cX$ becomes
\beal{es55a11}
\ud s^2 = G_{ij} \ud y_i \ud y_j + G^{ij} \ud \phi_i \ud \phi_j ~,~
\eea
where $G^{ij}$ is the inverse of $G_{ij} := \partial_i \partial_j G$ for some symplectic potential $G$.
This potential $G$ not only determines the metric but also relates to the complex structure on $\cX$ as
$\mathcal{I} = \left[
\ba{cc}
0 & -G^{ij} \\
G_{ij} & 0 
\ea
\right]
~.~$ 
Combining the expression of the complex structure and the original expression of the Killing vector field in \eref{es55a3b}, we can write
\beal{es44a15}
K = b_i \frac{\partial}{\partial \phi_i} ~,~
\eea
where $b\in \mathbb{R}^n$ and $b_i$ are the components of the Reeb vector.
They relate to the symplectic potential as follows
\beal{es44a16}
b_i = 2 G_{ij} y_j ~,~
\eea
and under the canonical form of the metric can be directly related to the vectors $v^a$ as
$b^{\text{can}} = \sum_a v^a$.

The Reeb vector has the following norm
\beal{es44a17}
1 = b_i b_j G^{ij} = 2 b_i G_{jk} y_k G^{ij} = 2 \left<b,y\right> ~,~
\eea
which can be used to show that the base of the cone $\cX$ at $r=1$ defines a hyperplane
\beal{es44a18}
\left\{
y \in \mathbb{R}^n ~|~
\left<b,y\right> = \frac{1}{2}
\right\} ~.~
\eea
This hyperplane intersects $\sigma$ in \eref{es55a5} to form a polytope $\Delta$, which is precisely what we called the toric diagram.
Note that the Reeb vector is always in the interior of the polytope $\Delta$ and is chosen such that along the hyperplane in \eref{es44a18}, one of its components is set to 
\beal{es44a19}
\fbox{$b_n=n$} ~.~
\eea

So far, everything described applies to general K\"ahler cones in the case of $\cX$ being Calabi-Yau, i.e. a toric Gorenstein singularity. Because of this, the normal vectors $v_a$ after some $GL(n,\mathbb{Z})$ transformation can be taken to be
\beal{es44a20}
v_a = (w_a,1) ~,~
\eea
where $w_a \in \mathbb{Z}^{n-1}$.
These are precisely the normal vectors that are used in the computation of the Hilbert series in \eref{es50a1}, when each simplex in the triangulation of the toric diagram is taken to be a polyhedral cone on its own right.
Note also that the vectors $w_a \in \mathbb{Z}^{n-1}$ are precisely pointing to the vertices of the toric diagram of the Calabi-Yau $n$-fold.\footnote{We can also construct $\cX$ symplectically in terms of a GLSM quotient of $\mathbb{C}^d$.
  The quotient is taken in terms of integer-valued charges $Q_a^i$, where here $i=1,\dots,d-n$ and $a = 1,\ldots,d$, as
  \beal{es44a25}
  \cX = 
  \left\{
  (Z_1,\dots, Z_d) \in \mathbb{C}^d ~|~
  \sum_a Q_i^{a} |Z_a|^2 = 0
  \right\} ~,~
  \eea
  with $Z_{a=1,\dots,d}$ the coordinates of $\mathbb{C}^d$.
  Note that the integer-valued charges $Q_i^a$ of the quotient form a map
  $Q_i^a ~:~ \mathbb{Z}^d \rightarrow \mathbb{Z}^{d-n}$.
  The normal vectors form the defining kernel of the map with
  $\sum_a Q_i^{a} v_a = 0$.
  Given that $\cX$ is Gorenstein and the normal vectors $v_a=(w_a,1)$, the charge vectors satisfy
  $\sum_a Q_i^a = 0$,
  which is the Calabi-Yau condition of the quotient.
  Note that since $\cX$ is an affine algebraic variety, the space of holomorphic functions on $\cX$ are elements of a coordinate ring.
  These elements are counted by a generating function which we have studied in section \sref{shs} as the Hilbert series of $\cX$.
}
As discussed in section \sref{scy}, we will study Calabi-Yau $n$-folds whose toric diagrams are in addition {\it reflexive}.

\subsection{The Volume Function \label{svolfct}} 
The volume of the Sasaki-Einstein base $Y$ is a strictly convex function that arises from the Einstein-Hilbert action on $Y$ \cite{Martelli:2005tp,Martelli:2006yb}.
For a Sasakian metric on $Y$, the volume takes the general form
\beal{es48a1}
\text{vol}[Y] = 
\int_Y \ud \mu = 2n \int_{r\leq 1} \frac{\omega^n}{n!} ~,~
\eea
where $\ud\mu$ is the Riemannian measure on the cone $\cX$.
It has been shown in \cite{Martelli:2006yb} that the volume of the base $Y$ is a function of Reeb vector components $b_i$ and in relation to the volume of $S^{2n-1}$, 
\beal{es48a2}
\text{vol}[S^{2n-1}] = \frac{2\pi}{(n-1)!} ~,~
\eea
the ratio
\beal{es48a3}
V(b; Y) := \frac{\text{vol}[Y]}{\text{vol}[S^{2n-1}]} ~,~
\eea
is always an algebraic number, which we define as the {\bf volume function}.
We remark that sometimes we will also denote this by $V(b; \cX)$ without ambiguity, though the volume is, of course, always that of the compact base $Y$.

The volume function can be derived directly from the Hilbert series of the coordinate ring of $\cX$ that was discussed in section \sref{shs},
\beal{es48a5}
\fbox{
  $
  V(b_i; Y) = \lim\limits_{\mu\rightarrow 0}  \mu^n g(t_i = \exp[-\mu b_i]; \cX)
  $
}
~.~
\eea
The above limit picks the leading order in $\mu$ in the expansion of the Hilbert series $g(t_i = \exp[-\mu b_i]; \cX)$, which was shown to be related to volume of the Sasaki-Einstein base $Y$ in \cite{Martelli:2005tp,Martelli:2006yb,Futaki:2008vr}.

\begin{figure}[h!!!]
\begin{center}
\resizebox{0.6\hsize}{!}{
  \includegraphics[trim=0mm 0mm 0mm 0mm, width=8in]{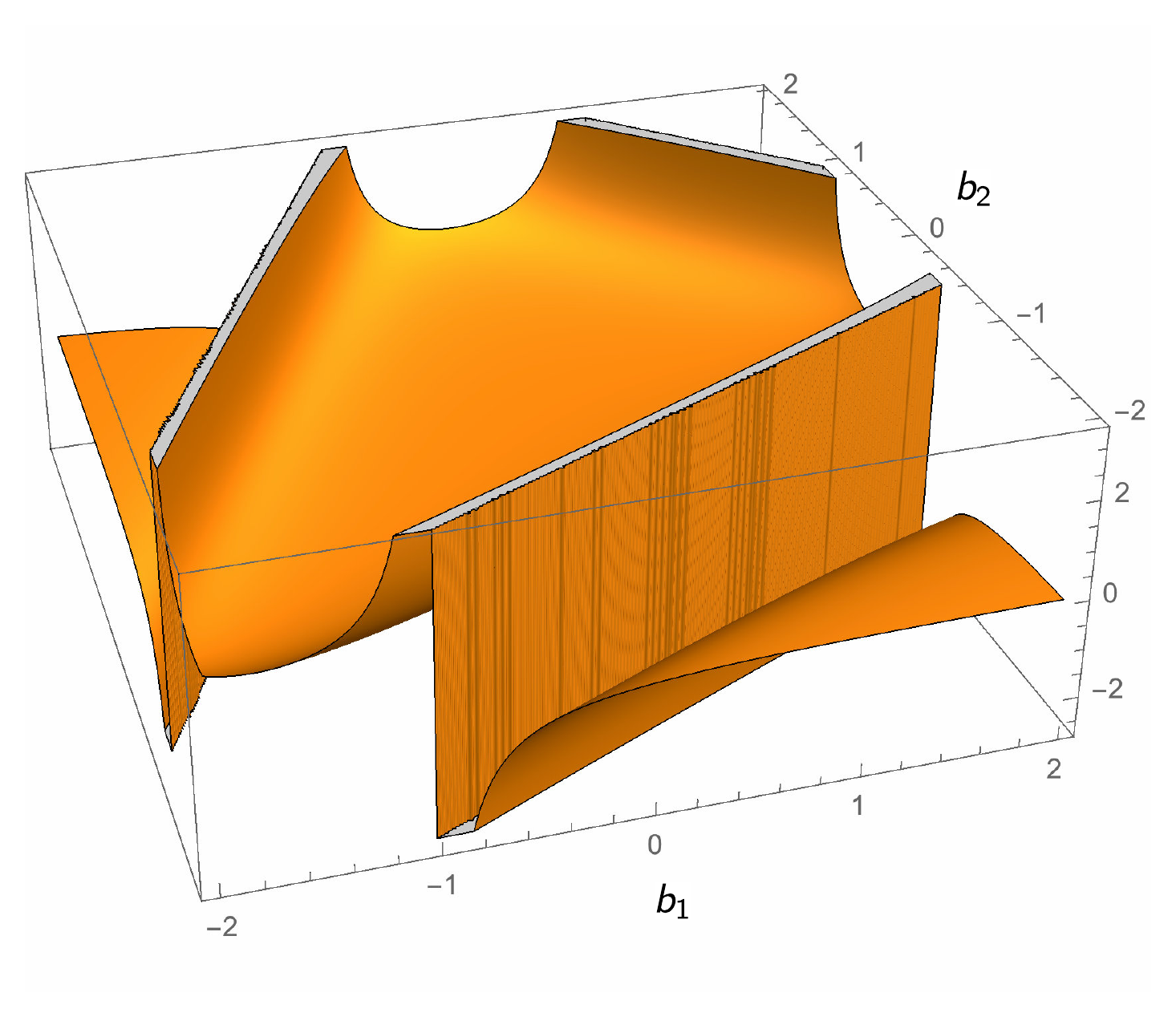}
}
\caption{{\sf {\small
      Volume function plot for $\mathbb{C}^3/\mathbb{Z}_3$.
      We set $b_3 = 3$ and $V$ is plotted against the components $b_{2,3}$ of the Reeb vector.
\label{fc3z3volfct}
}}}
\end{center}
\end{figure}

\subsubsection{Volume Functions for $\mathbb{C}^n/\mathbb{Z}_n$}

Let us take our running example of the Hilbert series that were computed in \eref{es51a3} for $\mathbb{C}^n/\mathbb{Z}_n$ with $n=3,4,5$ and calculate the corresponding volume functions.
Under the limit in \eref{es48a5}, the volume function are
\beal{es48a10}
V(b_i; \mathbb{C}^3/\mathbb{Z}_3) &=& 
\frac{-9}{
(b_1 - 2 b_2 - 3) (- 2 b_1 + b_2 - 3) (b_1 + b_2 - 3) 
} ~,~
\nn\\
V(b_i; \mathbb{C}^4/\mathbb{Z}_4) &=& 
\frac{64}{
(b_1 + b_2 - 3 b_3 - 4) (b_1 - 3 b_2 + b_3 - 4) (- 3 b_1 + b_2 + b_3 - 4)
} 
\nn\\
&&
\times
\frac{1}{
 (b_1 + b_2 + b_3 -  4) 
}~,~
\nn\\
V(b_i; \mathbb{C}^5/\mathbb{Z}_5) &=& 
\frac{
-625
}{
(b_1 + b_2 + b_3 - 4 b_4 - 5) (b_1 + b_2 - 4 b_3 + b_4 - 5) (b_1 - 4 b_2 + b_3 + b_4 - 5)  } 
\nn\\
&&
\times 
\frac{1}{ (-4 b_1 + b_2 + b_3 + b_4 - 5) (b_1 + b_2 + b_3 + b_4 - 5)}~.~
\eea
As discussed in \eqref{es44a19},  we take $b_n=n$.
\fref{fc3z3volfct} shows a $b_i$-plot of the volume function for $\IC^3/\IZ_3$.
Here, $\dim_{\IC} \cX = 3$ so there are three components of the Reeb vector.
We set $b_3 = 3$, and plot $V(b1,b2)$.

\subsection{Volume Minimization \label{svolmin}}

The volume function plays an important role in physics.
In this subsection, we briefly review the story of how the AdS/CFT correspondence gives the remarkable statement that minimizing the volume function is the same as maximizing the central charge in the dual conformal field theory for Calabi-Yau 3-folds.

\subsubsection{AdS/CFT and $a$-maximization \label{samax}}
As discussed in section \sref{ss2}, D3-branes probing Calabi-Yau cones $\cX$ of dimension 3 give rise to $4d$ $\mathcal{N}=1$ supersymmetric gauge theories.
It is expected that theories of this class flow at low energies to a superconformal fixed point. The Reeb vector discussed in \S\ref{sse} generates the $U(1)$ R-symmetry of the theory and the global R-symmetry is part of the superconformal algebra.

The superconformal R-charges of the theory are determined by a procedure known as {\bf $a$-maximization} \cite{Intriligator:2003jj,Butti:2005vn,Butti:2005ps}. It determines the superconformal R-charges uniquely by maximizing the combination of 't Hooft anomalies which completely determine the central charge of the superconformal field theory in 4 dimensions,
\beal{es53a1}
a(R) = \frac{1}{32} (9 \Tr R^3 - 3 \Tr R) ~,~
\eea
where for anomaly free theories $\Tr R = 0$.
The idea is to write a trial function $a_{\text{trial}}(R_t)$ in terms of trial R-charges,
\beal{es53a2}
R_t = R_0 + \sum_i c_i F_i
~,~
\eea
where $F_i$ can be thought as charges coming from the global symmetries which are not the R-symmetry in the theory. When the trial function $a_{\text{trial}}(R_t)$ is maximized, only the superconformal R-charges $R_0$ contribute as shown in \cite{Intriligator:2003jj}. 

\comment{
\begin{figure}[h!!!]
\begin{center}
\resizebox{0.4\hsize}{!}{
  \includegraphics[trim=0mm 0mm 0mm 0mm, width=8in]{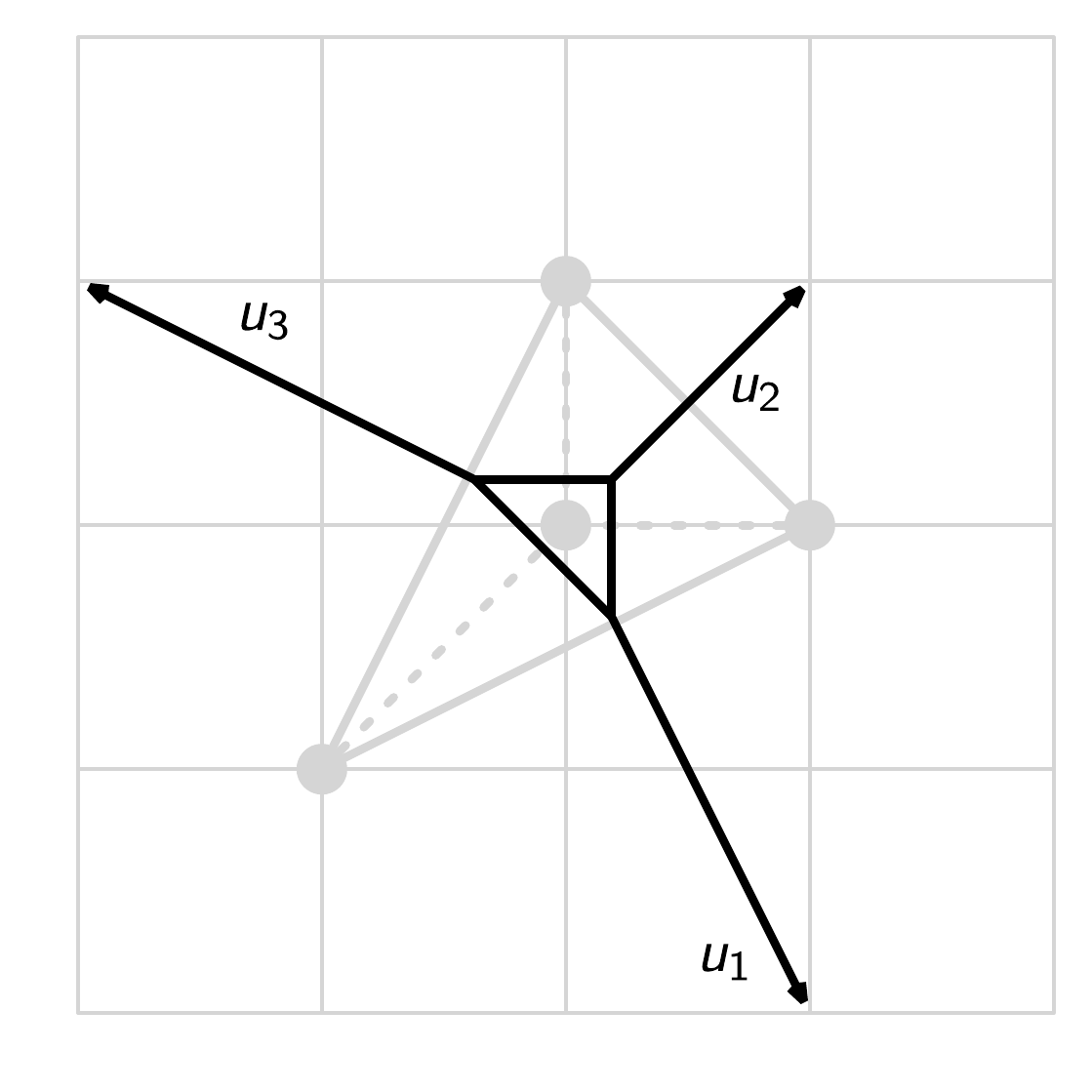}
}
\caption{{\sf {\small
The $(p,q)$-web diagram originating from the toric diagram of $\mathbb{C}^3/\mathbb{Z}_3$. The external legs of the $(p,q)$-web are vectors $u_i$ which contribute to the trial cubic $a$-function.
\label{frchpq}
}}}
\end{center}
\end{figure}
}

When the probed Calabi-Yau 3-fold singularity $\cX$ is toric, it was shown in \cite{Butti:2005vn,Butti:2005ps} that the trial $a$-function takes the form of a cubic function 
\beal{es53a3}
a(R) = \frac{9}{32} \left(2 A + \sum_{(i,j)\in C}  |\left[ u_i, u_j \right]| (R_{i+1} + R_{i+2} + \dots + R_{j} - 1)^3 \right) ~,~
\eea
where $R_i$ are trial R-charges with $\sum_i R_i = 0$ and $\left[ u_i, u_j \right] := \det(u_i,u_j)$. Note that here $u_i$ are the vectors associated to the $(p,q)$-web diagram corresponding to the Calabi-Yau 3-fold toric diagram as illustrated in part (b) of \fref{fdp3pq}.
Moreover, $A$ is the area of the toric diagram scaled so that the unit triangle has area $1$ and $C$ is the set of all the unordered pairs of vectors in the $(p,q)$-web (including non-adjacent ones).

\begin{figure}[h!!!]
\begin{center}
\resizebox{0.8\hsize}{!}{
  \includegraphics[trim=30mm 10mm 0mm 30mm, width=8in]{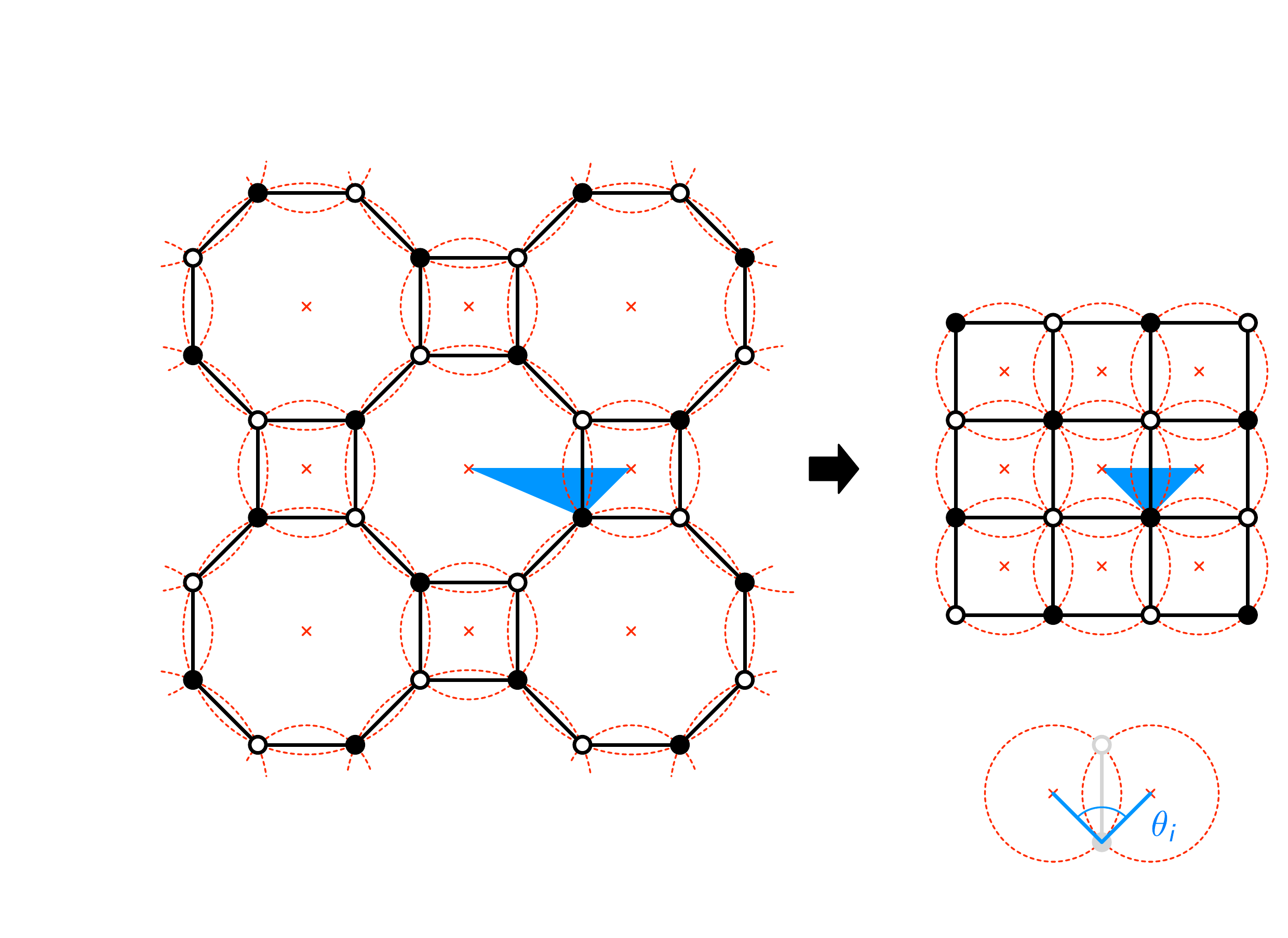}
}
\caption{{\sf {\small
A brane tiling under isoradial embedding links to distinct brane tilings related by Seiberg duality. This figure shows the two brane tilings corresponding to the Hirzebruch zero surface $F_0$. The blue triangle indicates the angle $\theta_i$ corresponding to the R-charge $R_i$.
\label{fisoradial}
}}}
\end{center}
\end{figure}

In the context of brane tilings, $a$-maximization was shown to relate to isoradial embeddings of the bipartite periodic graph on the $2$-torus \cite{Hanany:2005ss}. R-charges of bifundamental chiral matter multiplets $R_i$ were associated to angles in the bipartite graph,
\beal{es53a5}
\theta_i = \pi R_i
\eea
as illustrated in \fref{fisoradial}. The isoradial embedding of the graph determines the correct angles $\theta_i$ which in turn relate to the superconformal R-charges $R_i$. 
 
\subsubsection{Volume Minimization is $a$-maximization \label{svolminamax}}

The AdS/CFT correspondence relates the central charge $a$-function with the volume of the Sasaki-Einstein 5-manifold base $Y$ for $4d$ $\mathcal{N}=1$ superconformal field theories as follows \cite{Gubser:1998vd,Henningson:1998gx},
\beal{es54a1}
a(R;Y) = \frac{\pi^3 N^2}{4 V(R;Y)} ~.~
\eea
Whence we can define the normalized $a$-function
\beal{es54a2}
A(R;Y) \equiv
\frac{a(R; Y)}{a(R; S^5)} 
= \frac{\text{vol}[S^5]}{\text{vol}[Y]} 
= \frac{1}{V(b; Y)} ~.~
\eea
This reciprocal relationship implies that when the $a$-function is maximized along the RG-flow, the corresponding value of the volume function is minimized. This is in line with the fact that the Reeb vector generates the $U(1)_R$ symmetry and that the volume function $V(b_i;Y)$ is dependent only on its components $b_i$. 

Following our review on volume functions in section \sref{svolfct}, one can write a trial volume function $V(b_i;Y)$ dependent on trial Reeb vector components $b_i$ such that the Reeb vector $b$ is always in the interior of the cone.
As shown in \cite{Martelli:2006yb}, the minimized volume $V(b_i^{*};Y)$ identifies the critical Reeb vector $b^{*}$, which under the volume function in \eref{es48a3} results in the volume of the Sasaki-Einstein manifold $Y$. The process of identifying the critical Reeb vector is also known as {\bf $Z$-minimization} \cite{Butti:2005vn,Tachikawa:2005tq}.

\subsubsection{Towards a Classification of Calabi-Yau Volumes \label{sfree}}
We have discussed $a$-maximization and volume minimization in the context of type IIB string theory on $AdS_5 \times Y$ corresponding to $4d$ $\mathcal{N}=1$ superconformal field theories.
As discussed in section \sref{ss2}, toric Calabi-Yau 4-folds with Sasaki-Einstein 7-manifolds as bases were studied in the context of worldvolume theories of M2-branes probing the Calabi-Yau 4-fold singularity. These were some of the first examples where the AdS$_4$ /CFT$_3$ correspondence was studied systematically.
Here, the worldvolume theory on the M2-branes is a $3d$ $\mathcal{N}=2$ superconformal Chern-Simons theory.
For these theories, the role of the $a$-function is played by the supersymmetric free energy on the 3-sphere \cite{Kapustin:2009kz,Jafferis:2010un}
\beal{es55a1F}
F = - \log|Z_{S^3}| ~.~
\eea
The free energy when extremized was shown to determine the superconformal R-charges of the theory in \cite{Closset:2012vg}. Moreover, it was shown in \cite{Herzog:2010hf,Martelli:2011qj} that the free energy can be related in the large $N$ limit to the volume of the Sasaki-Einstein 7-manifold $Y$ as follows
\beal{es55a2F}
F= N^{3/2} \sqrt{\frac{2\pi^6}{27 \text{vol}[Y]}} ~.~
\eea

However, there are several obstacles that do not allow us to make the same conclusions as in the AdS$_5$/CFT$_4$ case. First of all, it is not known what the $3d$ $\mathcal{N}=2$ theory is for instance for a general toric Calabi-Yau 4-fold. Brane tilings provide a unified description of $4d$ $\mathcal{N}=1$ theories and the corresponding toric Calabi-Yau 3-folds \cite{Hanany:2005ve,Franco:2005rj}. 
Such a bridge, following many proposals \cite{Hanany:2008cd,Lee:2006hw}, does not exist for worldvolume theories of M2-branes and arbitrary toric Calabi-Yau 4-folds. 
For some toric Calabi-Yau 4-folds, the corresponding $3d$ $\mathcal{N}=2$ theory are not known.
This is a major obstacle in verifying that the minimum of the volume of the Sasaki-Einstein 7-manifold systematically corresponds to the free energy $F$. 

Moreover, the free energy $F$ does not necessarily take a simple polynomial form as it is the case for the $a$-function as discussed in section \sref{samax}. In \eref{es53a3}, we referred to the result in \cite{Butti:2005vn,Butti:2005ps} that the $a$-function for $4d$ $\mathcal{N}=1$ theories corresponding to toric Calabi-Yau 3-folds is a cubic function in terms of trial R-charges. Naively, one would hope that the free energy $F$ for $3d$ $\mathcal{N}=2$ theories is a quartic function of trial charges. This is not generally the case  and several works \cite{Amariti:2012tj,Lee:2014rca} attempted to correct the conjectured quartic function so as to match its extremal value to the volume of a certain class of Sasaki-Einstein 7-manifolds.

As reviewed in \sref{ss2}, toric Calabi-Yau 4-folds also appeared in the context of worldvolume theories living on probe D1-branes. These theories are $2d$ $(0,2)$ quiver gauge theories realized in terms of a type IIA brane configuration known as brane brick models \cite{Franco:2015tna,Franco:2015tya,Franco:2016qxh}. 
More recently, toric Calabi-Yau 5-folds were studied in the context of probe Euclidean D(-1)-branes \cite{Franco:2016tcm}. The worldvolume theory of the probe D(-1)-branes is a $0d$ $\mathcal{N}=1$ supersymmetric matrix model. 
Both probed toric Calabi-Yau 4-folds and 5-folds have base manifolds whose volumes can be computed via volume minimization. It is an interesting question to ask what role the volume minimum plays in the context of the corresponding $2d$ and $0d$ theories.
 
Computing the volume minima for large sets of toric Calabi-Yau $n$-folds may shed some light to these questions. In the following section, we will precisely do this for toric Calabi-Yau $n$-folds whose toric diagrams are reflexive. By plotting these volume minima against topological quantities of the underlying toric varieties allows us to derive general statements about the overall behavior of volume functions for Calabi-Yaus in various dimensions.


\section{Minimum Volumes for CY${}_{n}$}\label{s:res}\setall

We can now combine the theoretical results obtained thus far for the topological quantities, Hilbert series and volume functions of the reflexive polytopes in various dimensions, in conjunction with the data we have collected.
In this section, we will study the connections between these quantities, derive some bounds and raise intriguing conjectures.

\subsection{Volume Minimum versus Topological Invariants}

In the previous sections, we have discussed various characteristic quantities that can be derived from the compact toric variety $X(\Delta)$ corresponding to a reflexive polytope $\Delta$ and the non-compact Calabi-Yau cone $\cX$ over $X(\Delta)$ and its Sasaki-Einstein base manifold $Y$.
Throughout this section, $\dim_{\IC}\cX = n$, $\dim_{\IC} X(\Delta_{n-1}) = n-1$, and $\dim_{\IR} Y =  2n-1$.

The first characteristic that can be computed is the Euler number for $\widetilde{X(\Delta)}$, which is simply the number of $n$-cones of the fan $\Sigma(\Delta_{n-1})$ by Theorem \ref{betti}.
On the other hand, as was discussed in section \sref{svolmin}, the volume of the Sasaki-Einstein base $Y$ is another important quantity for $\cX$.
The volume function $V(b_i; Y)$ is expressed as a function of Reeb vector components $b_i$ and can be obtained directly from the Hilbert series of the coordinate ring for $\cX$. The Hilbert series in turn is easily obtained from the polytope data by \eqref{es50a1}.
For threefolds, from the point of view of the dual gauge theory, following $a$-maximization and the reciprocal relationship to the volume function $V(b_i;Y)$ discussed in \sref{samax}, the volume of $Y$ minimizes to a critical value when the theory flows to a IR fixed point. As such, the minimized volume matched for theories that flow to the same fixed point -- theories that are Seiberg dual in 4 dimensions.

In order to better understand the role played by the minimum volume for Calabi-Yau varieties in higher dimensions, we propose
to first study the relationship between
\beal{es111a1}
  \fbox{ 
  \begin{minipage}[c][3\height][c]{\dimexpr \textwidth-90 \fboxsep-2\fboxrule\relax}
  $
  V(b_i^*; Y) :=
  \min\limits_{b_i | b_{n} = n} V(b_i; Y)
  $
  \end{minipage}}
\longleftrightarrow
  \fbox{
  \begin{minipage}[c][2\height][c]{\dimexpr \textwidth-120 \fboxsep-6\fboxrule\relax}
  $
  \chi(\widetilde{X(\Delta_{n-1})})
  $
  \end{minipage}}
~.~
\eea
The Euler number $\chi(\widetilde{X(\Delta)})$ is the number of top-dimensional cones as stated in \eqref{es80a15}.

Because the numbers of reflexive polytopes is finite in each dimension, we have:
\begin{lemma}
For all non-compact toric Calabi-Yau $n$-folds $\cX$ with Sasaki-Einstein base $Y$, originating from toric varieties $X(\Delta_{n-1})$ and reflexive polytopes $\Delta_{n-1}$, there is a finite number thereof where
\beal{es500a2}
1/ V(b_i^{*}; Y) = \chi(\widetilde{X(\Delta_{n-1})}) ~.~
\eea
\end{lemma}
What we will attempt to do now is to see whether any reflexive polytopes $\Delta$ satisfies this relation, isolate them and test if the equality furnishes some bound.

\subsubsection{Calabi-Yau 3-fold Cones}
\begin{figure}[h!!!]
\begin{center}
\resizebox{0.7\hsize}{!}{
  \includegraphics[trim=0mm 0mm 0mm 0mm, width=8in]{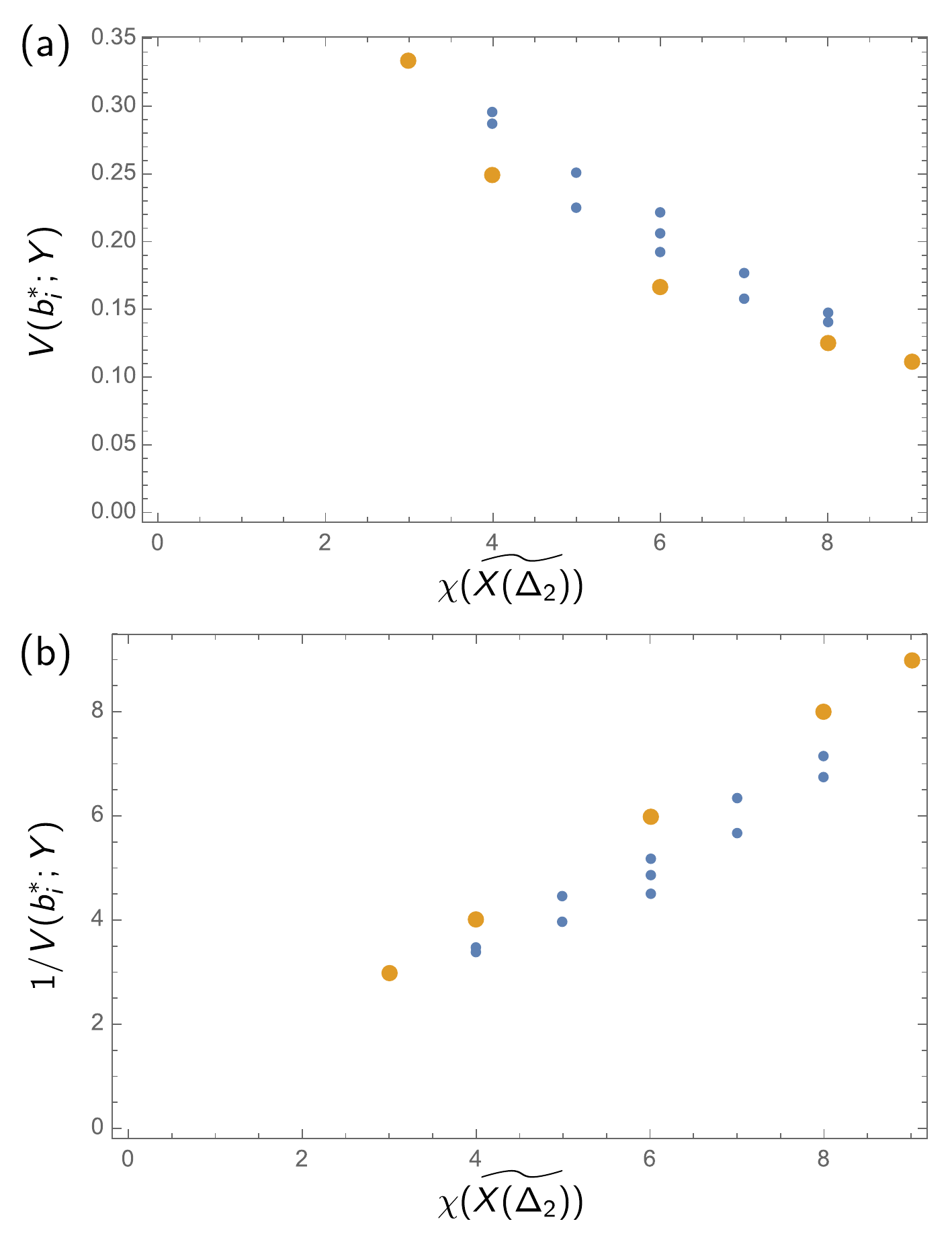}
}
\caption{{\sf {\small
(a) The Euler number of the 16 reflexive toric Calabi-Yau 3-folds $\mathcal{X}$ against the minimum volume $V(b_i^{*};Y)$, and (b) the Euler number against the inverse of the minimum volume. Note that the orange points correspond to the 5 Calabi-Yau 3-folds  which are Abelian orbifolds of $\mathbb{C}^3$ with $1/V(b_i^{*};Y)=\chi(\widetilde{X(\Delta_2)})$.
\label{f:n=2min}}}}
\end{center}
\end{figure}
We begin with the 16 reflexive polygons $\Delta_2$; for these $\cX = \cC_{\IC}(X(\Delta_2))$ are Calabi-Yau 3-folds.
The plot of the volume minimum of $Y$ against the Euler number is shown in \fref{f:n=2min}(a).
We can clearly see from the plot that for some toric Calabi-Yau 3-folds, the minimum volume is inversely related to the corresponding Euler number of $\widetilde{X(\Delta_2)}$.
To illustrate this, in \fref{f:n=2min}(b), we plot the inverse of the minimum volume $1/V(b_i^{*}; Y)$ against the Euler number. By doing so, we note that for 5 toric Calabi-Yau 3-folds, the inverse of the minimum volume is exactly the corresponding Euler number.

\begin{figure}[h!!!]
\begin{center}
\resizebox{0.7\hsize}{!}{
  \includegraphics[trim=0mm 0mm 0mm 0mm, width=8in]{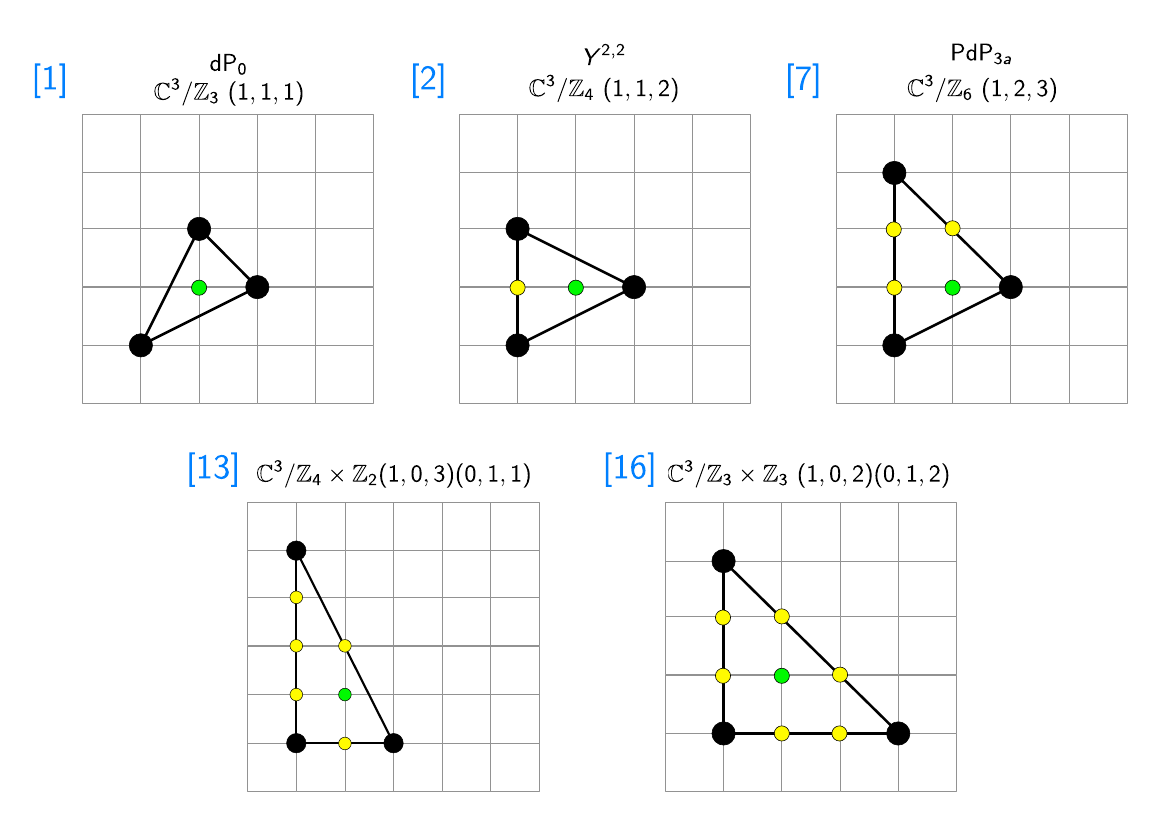}
}
\caption{{\sf {\small
      The 5 toric Calabi-Yau 3-folds whose inverse minimum volume $1/V(b_i^{*};Y)$ is the same as the Euler number $\chi(\widetilde{X(\Delta)})$ of the complete resolution of $X(\Delta)$.
      They are all lattice triangles. We retain the numbering from Figure \ref{f:n=2}.
\label{f:n=2min5}}}}
\end{center}
\end{figure}

A closer look reveals that these 5 toric Calabi-Yau $3$-folds are all Abelian orbifolds of $\mathbb{C}^3$. The corresponding toric diagrams are lattice triangles and are shown in \fref{f:n=2min5}.
We note, from Figure \ref{f:n=2}, that the number of (extremal) vertices of the 16 reflexives range from 3 to 6 and the ones with 3 vertices all correspond to $\cX$ being Abelian orbifolds of $\IC^3$.
For reference, we include the 16 volume functions as well as their critical points and values in Appendix \ref{ap:refcy3}, \tref{tcy3volminA}.

Indeed, this is true in general.
\begin{theorem}\label{orb}
  For $\cX$ being an Abelian orbifold of $\mathbb{C}^{n}$ whose toric diagram is a reflexive polytope, the minimum volume and the Euler number are related by
  \[
  1/ V(b_i^{*}; Y) = \chi(\widetilde{X(\Delta_{n-1})}) \ .
  \]
\end{theorem}
\begin{proof}
  We recall from \eref{es48a2} that the volume of the sphere $S^{2n-1}$ is $\text{vol}[S^{2n-1}] = \frac{2\pi}{(n-1)!}$ and that the minimal volume of the Sasaki-Einstein base $Y$ of $\cX$ is defined as
$V(b_i^*; \cX) \equiv \frac{\text{vol}[Y]}{\text{vol}[S^{2n-1}]}$,
  normalized by $\text{vol}[S^{2n+1}]$.
  In particular, if $Y$ is the sphere $S^{2n-1}$ itself, the normalized volume
  \begin{equation}\label{vCn}
  V(b_i^*; \IC^n) = V(b_i^*; S^{2n-1}) = 1 \ ,
  \end{equation}
  since $\IC^n$ is the real cone over $S^{2n-1}$ (recall that in our convention, we write $V(b;\cX)$ and $V(b;Y)$ inter-changeably).
  
  Moreover, when one takes an Abelian orbifold of the form $\mathbb{C}^{n}/\Gamma$, where the order is $|\Gamma|$, the volume for $S^{2n-1}/\Gamma$ becomes, by construction,
  \beal{es500a16}
  \frac{\text{vol}[S^{2n-1}]}{\text{vol}[S^{2n-1}/\Gamma]} =
  \frac{\cancelto{1}{V(b_i^*; \mathbb{C}^{n})}}{V(b_i^*; \mathbb{C}^{n}/\Gamma)} = |\Gamma| ,
  \eea
  where in the last step, we have used \eqref{vCn}.
  
  Now, under orbifolding, while the volume of the sphere {\em decreases} by a factor of $|\Gamma|$ as above, it is a standard fact of toric geometry \cite{fulton} that the $n$-dimensional volume of the convex lattice polytope $\Delta_n$ that represent the toric diagram of $\cX$ {\em increases} by a factor of $|\Gamma|$.
One can see this, for example, from the toric diagram of $\IC^3$ versus that of $\IC^3/\IZ_3$.
On the other hand, because our FRS triangulation ensures each simplex to be unit volume for the sake of smoothness, the scaling factor $\Gamma$ of the toric diagram is precisely the number $d_{n-1}$ of top-dimensional cones, which we recall from \eqref{es80a15}, is the Euler number:
\begin{equation}\label{dnOrb}
d_{n-1} = |\Gamma| = \chi(\widetilde{X(\Delta_{n-1})})
\end{equation}
Combing \eqref{es500a16} and \eqref{dnOrb} thus gives $1 / V(b_i^*; \mathbb{C}^{n}/\Gamma) = \chi(\widetilde{X(\Delta_{n-1})})$, as required.
\comment{
  \beal{es500a17}
  \frac{
    V_{toric}(\mathbb{C}^{n+1})}{
    V_{toric}(\mathbb{C}^{n+1}/\Gamma)
  } = \frac{1}{|\Gamma|} .
  \eea
On the other hand, from \eqref{es80a15}, the Euler number is given precisely by number $d_n$ of $n$-cones in the fan for the smoothing $\widetilde{X(\Delta)}$ due to FRS triangulation.
  Since $\Delta$ of $\mathbb{C}$ is made of a single simplicial cone, with the toric diagram being a lattice $n$-simplex of unit volume, and the fact that the volume of the toric diagram $V_{toric}(X)=d_n$ under FRS triangulation of $\Delta$, we have following \eref{es500a16}, \eref{es500a17} and \eref{es500a17} that
  \beal{es500a19}
  d_n =
  \frac{
    V_{toric}(\mathbb{C}^{n+1}/\Gamma)}{
    V_{toric}(\mathbb{C}^{n+1})}
  = |\Gamma| = 
  \frac{\text{vol}[S^{2n+1}]}{\text{vol}[S^{2n+1}/\Gamma]}
  = \frac{1}{V(b;\mathbb{C}^{n+1}/\Gamma)}~,~
  \eea
  which implies \eref{es500a2} for Abelian orbifolds of $\mathbb{C}^{n+1}$ that are reflexive toric Calabi-Yau $n$-folds.
}
\qed
\end{proof}

\subsubsection{Calabi-Yau 4-fold Cones}

\begin{figure}[h!!!]
\begin{center}
\resizebox{0.7\hsize}{!}{
  \includegraphics[trim=0mm 0mm 0mm 0mm, width=8in]{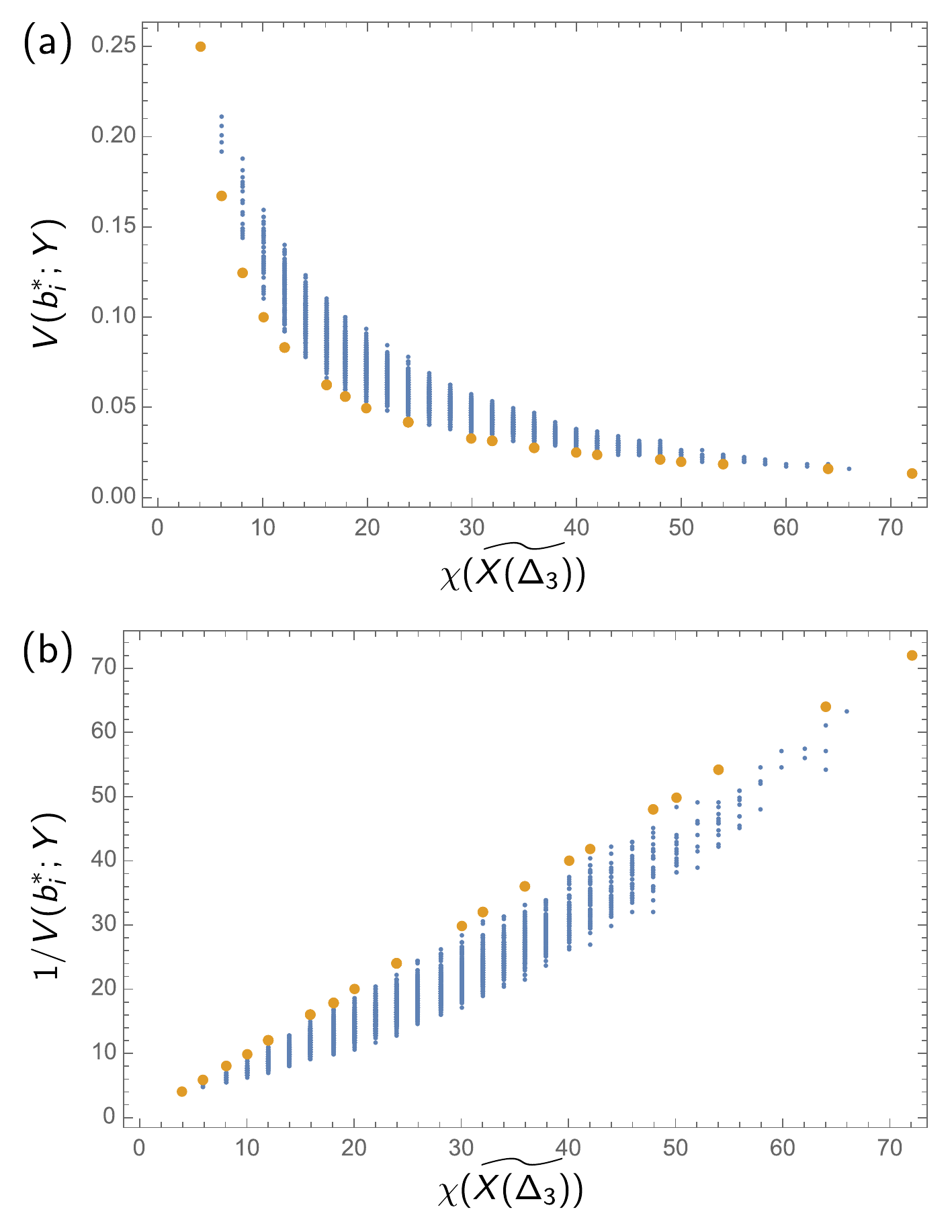}
}
\caption{{\sf {\small
(a) The Euler number of the 4319 reflexive toric Calabi-Yau 4-folds $X$ against the minimum volume $V(b_i^{*};Y)$, and (b) the Euler number against the inverse of the minimum volume. Note that the orange points correspond to the 48 Calabi-Yau 4-folds which are Abelian orbifolds of $\mathbb{C}^4$ with $1/V(b_i^{*};Y)=\chi(\widetilde{X(\Delta_3)})$.
\label{f:n=3min}}}}
\end{center}
\end{figure}

The observations made for reflexive toric Calabi-Yau 3-folds can be extended for toric Calabi-Yau 4-folds $\cX$ whose toric diagrams are reflexive polyhedra.
As stated in Table \ref{tcyno}, there are 4319 reflexive polytopes $\Delta_3$, with number of vertices in the range of 4 to 14.
We computed all minimized volumes and topological quantities for these Calabi-Yau 4-folds.
The volume minima plotted against the Euler number of the corresponding toric varieties is shown in \fref{f:n=3min}(a). In comparison, \fref{f:n=3min}(b) shows the inverse volume against the Euler number.
We can clearly see that for a subset of reflexive toric Calabi-Yau 4-folds, the reciprocal minimum volume is exactly the corresponding integer Euler number.

\begin{figure}[h!!!]
\begin{center}
\resizebox{0.7\hsize}{!}{
  \includegraphics[trim=0mm 0mm 0mm 0mm, width=8in]{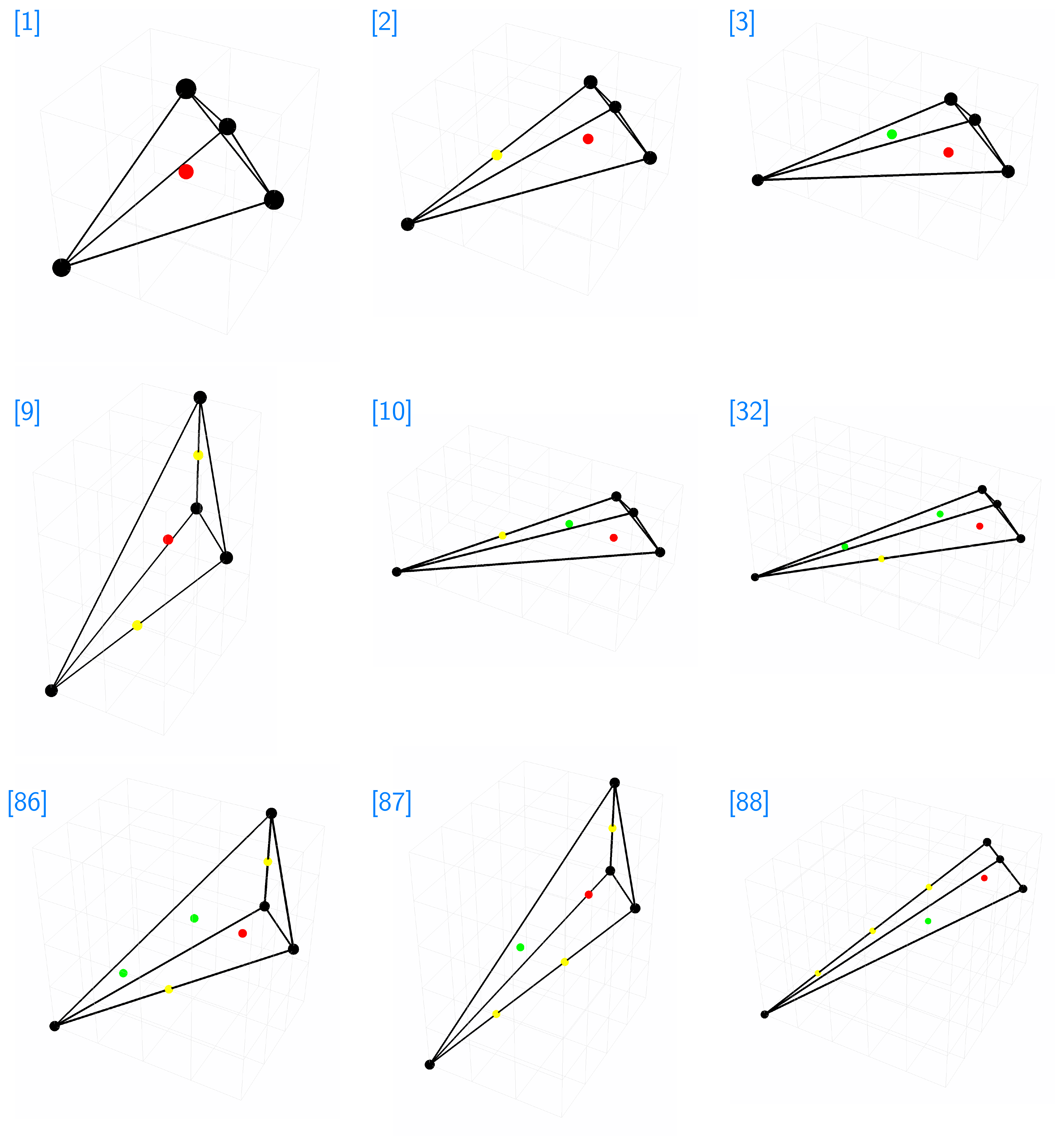}
}
\caption{{\sf {\small
A selection of the 48 toric Calabi-Yau 4-folds which are Abelian orbifolds of $\mathbb{C}^4$. Their inverse minimum volume $1/V(b_i^{*}; Y)$ is the same as the Euler number $\chi(\widetilde{X(\Delta)}$. Their toric diagrams are all lattice tetrahedra and the numbering is in the order of SAGE's database.
\label{f:n=3minset}}}}
\end{center}
\end{figure}

As discussed in Theorem \ref{orb}, Abelian orbifolds of $\mathbb{C}^4$, of which there are 48 out of the 4319, have toric diagrams which are lattice tetrahedra with 4 external vertices.
For these, the reciprocal of the minimum volume of the Sasaki-Einstein 7-manifold is exactly the same as the corresponding integer Euler number of the toric variety.
For reference, we include the 48 volume functions as well as their critical points and values in Appendix \ref{ap:orb}, Tables \ref{tcy4volmin} and \ref{tcy4volmin2}.
\fref{f:n=3minset} shows a selection of 9 of the 48 toric diagrams of these toric Calabi-Yau 4-folds.

\subsubsection{Calabi-Yau 5-fold Cones and a General Bound}

A similar conclusion can be given for toric Calabi-Yau 5-folds. As outlined in \sref{s:toric}, there are 473,800,776 reflexive polytopes $\Delta_4$ each corresponding to a toric Calabi-Yau 5-fold.
It is clearly impossible to compute the minimum volume of each of them at this stage, but, led by the results for toric Calabi-Yau 3-fold and 4-folds, we were able to make selective computations.
From the enormous list we take a sample of about 1000 reflexive $\Delta_4$, making sure that we cover examples for all cases of the number of vertices, which range from 5 to 33, to avoid statistical bias.

For example, a selection (there is a total of 1561 which have 5 vertices, corresponding to our orbifolds) of Abelian orbifolds of $\mathbb{C}^5$ indeed confirms that the relationship $1/ V(b_i^{*}; Y) = \chi(\widetilde{X(\Delta)})$ holds for reflexive Abelian orbifolds.
Furthermore, for a selection of reflexive toric Calabi-Yau 5-folds that are not Abelian orbifolds of $\mathbb{C}^5$, it was shown that $1/V(b_i^{*}; Y) \neq \chi(\widetilde{X(\Delta)})$. The data is shown in \fref{f:n=4min}.

\begin{figure}[h!!!]
\begin{center}
\resizebox{0.7\hsize}{!}{
  \includegraphics[trim=0mm 0mm 0mm 0mm, width=8in]{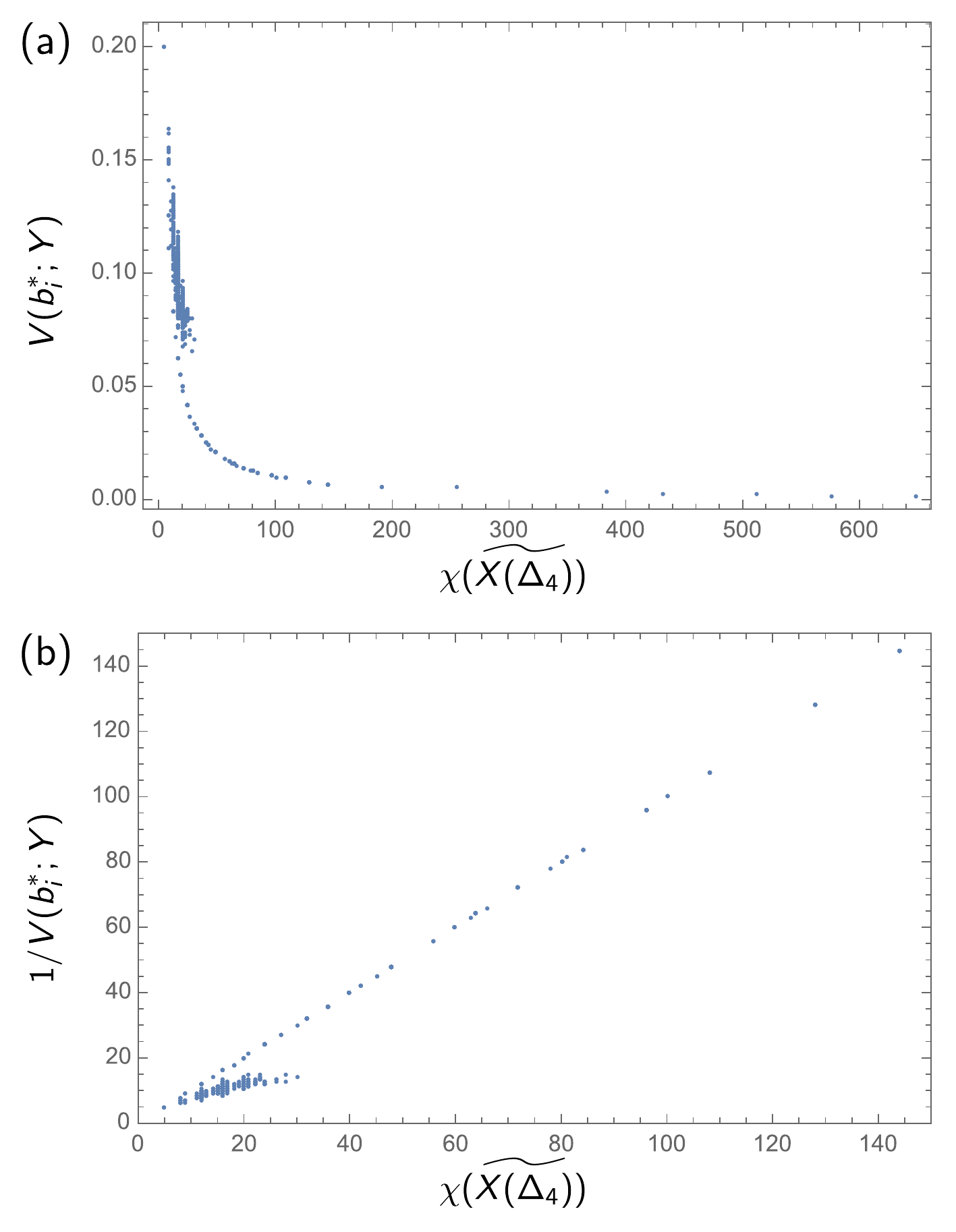}
}
\caption{{\sf {\small
      (a) The Euler number for a selection of reflexive toric Calabi-Yau 5-folds $\cX$ against the minimum volume $V(b_i^{*}; Y)$, and
      (b) the Euler number against the inverse of the minimum volume. Note that for Abelian orbifolds of $\mathbb{C}^5$ which are reflexive toric Calabi-Yau 5-folds, $1/V(b_i^{*};Y)=\chi(\widetilde{X(\Delta_4)})$.
\label{f:n=4min}}}}
\end{center}
\end{figure}

\comment{
  In fact, we conjecture that $1/ V(b_i^{*}; Y) \neq \chi(\widetilde{X(\Delta)})$ represents a bound on the minimum volume of all reflexive toric Calabi-Yau $n$-folds.
  This can be in particular observed for the cases of reflexive toric Calabi-Yau 3-folds and 4-folds and the respective plots in \fref{f:n=2min} and \fref{f:n=3min}. Moreover, a closer look in dimensions $3$ and $4$ reveals that the minimum volume is also bounded above.
We will attempt to analyze the origin of these bounds in the following section. 
}

\begin{figure}[h!!!]
\begin{center}
\resizebox{0.8\hsize}{!}{
  \includegraphics[trim=0mm 0mm 0mm 0mm, width=8in]{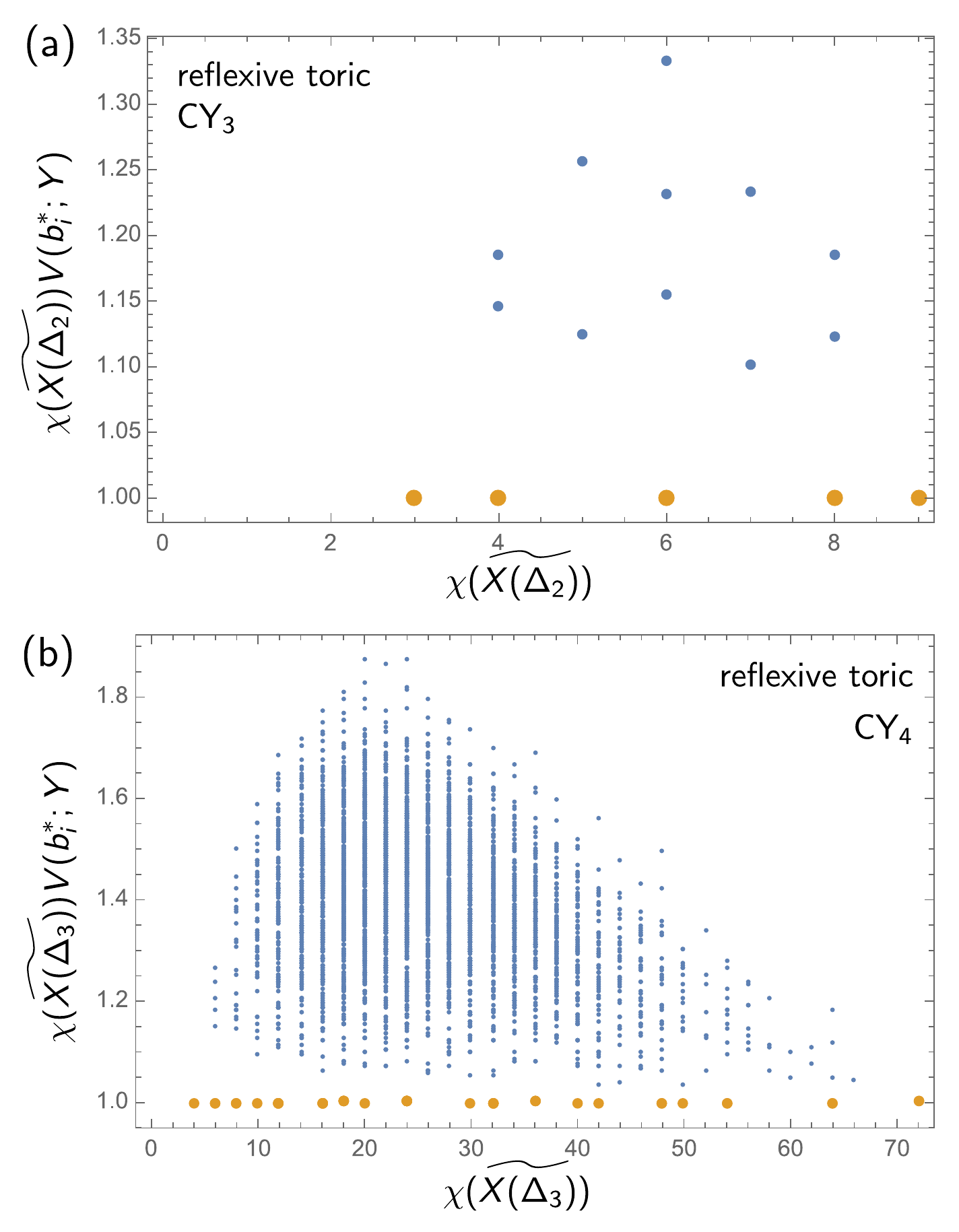}
}
\caption{{\sf {\small
The Euler number $\chi(\widetilde{X(\Delta_n)})$ against $\chi(\widetilde{X(\Delta_n)}) V(b_i^{*};Y)$ for (a) reflexive toric Calabi-Yau 3-folds and (b) 4-folds. Note that for reflexive toric Calabi-Yau, there is a unique maximum for $\chi(\widetilde{X(\Delta_2)}) V(b_i^{*};Y)$ at $\chi(\widetilde{X(\Delta_2)})=6$ for the cone over $\text{dP}_3$, whereas for reflexive toric Calabi-Yau 4-folds the maximum for $\chi(\widetilde{X(\Delta_3)}) V(b_i^{*};Y)$ occurs at $\cX_{1530}$ and $\cX_{2356}$ at $\chi(\widetilde{X(\Delta_3)})=20$ and $\chi(\widetilde{X(\Delta_3)})=24$, respectively.
\label{fCY34peakplots}}}}
\end{center}
\end{figure}

\subsection{Bounds on Minimum Volume \label{slowbound}}

As observed in the previous section, the minimum volume of the Sasaki-Einstein base of reflexive toric Calabi-Yau $3$, $4$ and $5$-folds seems to form a bound when the Calabi-Yau is an Abelian orbifold of $\mathbb{C}^n$. We therefore conjecture that, 
\begin{conjecture}
For toric Calabi-Yau $n$-folds $\cX$ with a toric diagram $\Delta_{n-1}$, the minimum volume $V(b_i^{*}; Y)$ of the Sasaki-Einstein base $Y$ has a lower bound:
\beal{es501a1}
V(b_i^{*}; Y)  \geq \frac{1}{\chi(\widetilde{X(\Delta_{n-1})})} ~,~
\eea
where $\chi(\widetilde{X(\Delta_{n-1})})$ is the Euler number of the completely resolved toric variety $\widetilde{X(\Delta_{n-1})}$. 
The bound is saturated when $\cX$ is an Abelian orbifold of $\mathbb{C}^{n}$.
\end{conjecture}
We remark that this conjecture is expected to hold for any toric Calabi-Yau $n$-fold, not just for the ones which have a reflexive polytope as their toric diagram.

\subsubsection{Identifying the Maximum}
The question arises whether there are any further bounds that can be identified for the minimum volume $V(b_i^{*}; Y)$, especially for the class of toric Calabi-Yau $n$-folds whose toric diagram is a reflexive polytope. 
We draw $V(b_i^{*}; Y) \chi(\widetilde{X(\Delta)})$ versus $\chi(\widetilde{X(\Delta)})$ in parts (a) and (b) in \fref{fCY34peakplots} for all our reflexive toric Calabi-Yau 3-folds and 4-folds, respectively.

These plots respectively show $V(b_i^{*}; Y)$ and $1/V(b_i^{*}; Y)$ against $V(b_i^{*}; Y) \chi(\widetilde{X(\Delta)})$.
An interesting observation from these plots is that specific values of the product $V(b_i^{*}; Y) \chi(\widetilde{X(\Delta)})$ reach a maximum value.

\begin{figure}[h!!]
\begin{center}
\resizebox{0.3\hsize}{!}{
  \includegraphics[trim=0mm 0mm 0mm 0mm, width=8in]{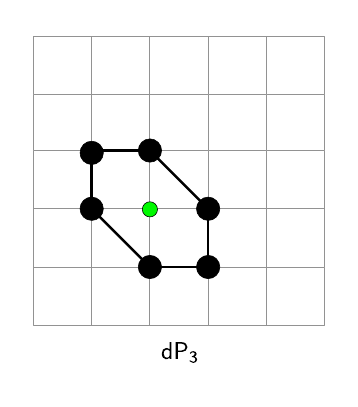}
}
\caption{{\sf {\small
The toric diagram for the cone over $\text{dP}_3$.
\label{fdp3}}}}
\end{center}
\end{figure}

For reflexive toric Calabi-Yau 3-folds, the maximum for $V(b_i^{*}; X) \chi(\widetilde{X(\Delta)})$ is uniquely attained, as can be seen from part (a) of \fref{fCY34peakplots}. The minimum horizontal straight-line at the bottom of the plot corresponds to the 5 Abelian orbifolds of $\mathbb{C}^3$ with reflexive toric diagrams.
We have been able to identify the unique maximum for $V(b_i^{*}; X) \chi(\widetilde{X(\Delta)})$ and the associated $\cX$ corresponds to the cone over $\text{dP}_3$, the del Pezzo surface of degree 6 which is $\IP^2$ blown up at 3 generic points.
The corresponding toric diagram is shown in Figure \ref{fdp3} and the minimum volume and Euler number are
\begin{equation}
  V(b_i^{*}; dP_3) = \frac{2}{9} \ ,
  \chi(\text{dP}_3) = 6
~,~
\end{equation}
giving us
\begin{equation}
\max\limits_{\Delta_2} V(b_i^{*}; Y) \chi(\widetilde{X(\Delta_2))})
=\frac{4}{3}~.~
\end{equation}

For reflexive toric Calabi-Yau 4-folds, the maximum value for $V(b_i^{*}; Y) \chi(\widetilde{X(\Delta_3)})$ is not unique.
There are two out of the 4319 reflexive toric Calabi-Yau 4-folds, which we will call $\cX_{1530}$ and $\cX_{2356}$,\footnote{As always, the subscripts in $\cX_{1530}$ and $\cX_{2356}$ refer to the ordering in the SAGE \cite{sage} database for reflexive toric Calabi-Yau 4-folds.} that have the largest value for $V(b_i^{*}; X) \chi(\widetilde{X(\Delta)})$.
We show the corresponding toric diagrams in Figure \ref{fCY4toricpeaks}.

\begin{figure}[h!!]
\begin{center}
\resizebox{0.8\hsize}{!}{
  \includegraphics[trim=0mm 0mm 0mm 0mm, width=8in]{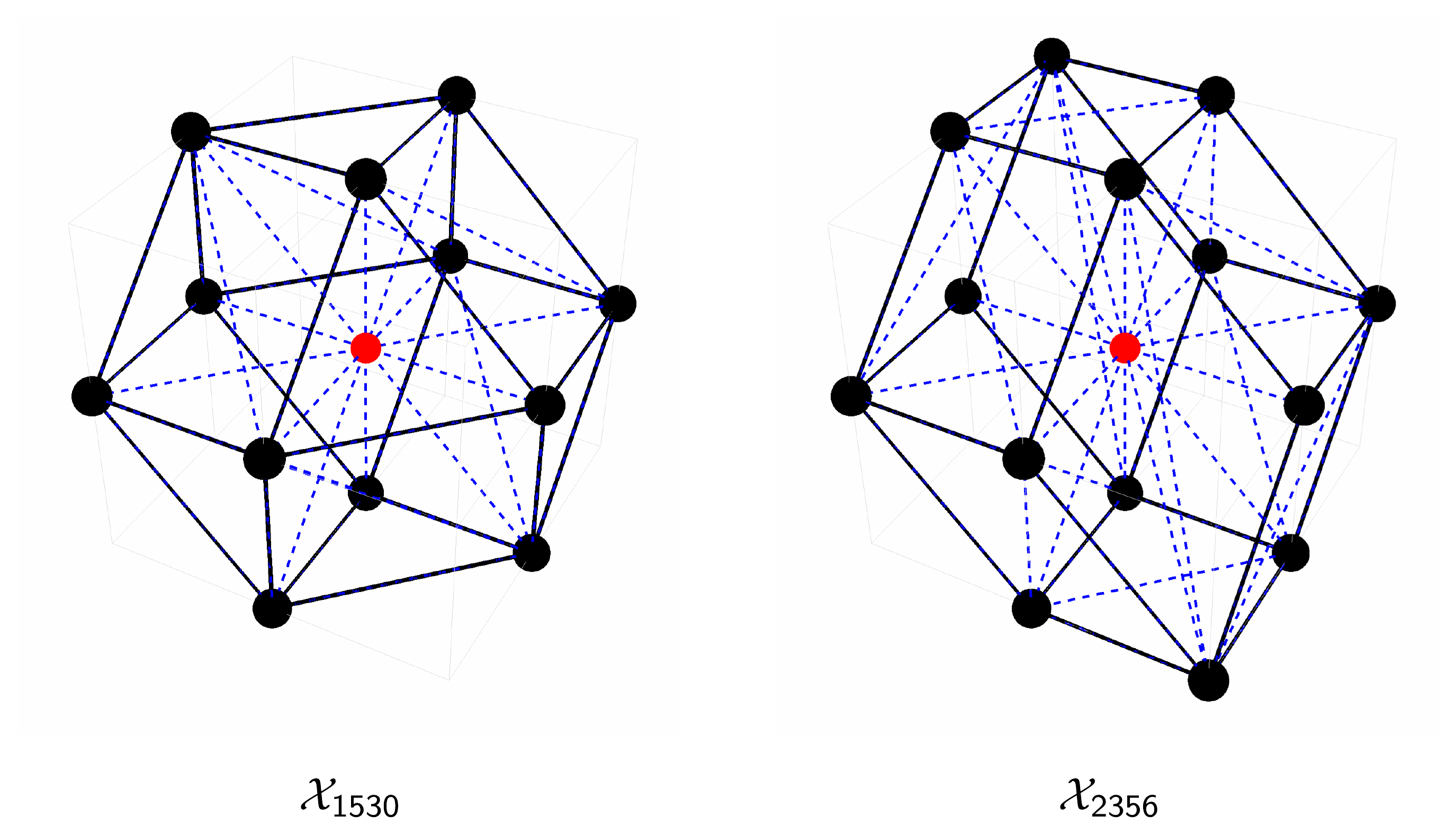}
}
\caption{{\sf {\small
Toric diagrams for reflexive toric Calabi-Yau 4-folds $\cX_{1530}$ and $\cX_{2356}$. The value for $V(b_i^{*}; Y) \chi(\widetilde{X(\Delta)})$ for these is the largest amongst all 4319 reflexive toric Calabi-Yau 4-folds.
\label{fCY4toricpeaks}}}}
\end{center}
\end{figure}

The minimum volumes and Euler numbers are
\beal{es520a2}
&
V(b_i^{*}; Y_{1530}) =\frac{3}{32} 
~,~
\chi(\widetilde{X_{1530}(\Delta)}) = 20
&
\nn\\
&
V(b_i^{*}; Y_{2356}) = \frac{5}{64}
~,~
\chi(\widetilde{X_{2356}(\Delta)}) = 24
&
~,~
\eea
both giving us that
\begin{equation}
  \max\limits_{\Delta_3} V(b_i^{*}; Y) \chi(\widetilde{X(\Delta_3)})
  = \frac{15}{8} = 1.875 \ .
\end{equation}

We currently do not have the computational power to address the maximal situation for $\Delta_4$, but at least we see that the toric diagrams for $\mathcal{C}(\text{dP}_3)$, $\cX_{1530}$ and $\cX_{2356}$, which are shown in Figures \ref{fdp3} and \ref{fCY4toricpeaks}, have hexagons in all their hyperplanes slices. This signifies that they are dP$_3$ fibrations.
Moreover, the maximum value for $V(b_i^*; Y) \chi(\widetilde{X(\Delta)})$ seems to increase in dimension.
We speculate that while the minimum value of $V(b_i^*; Y) \chi(\widetilde{X(\Delta_n)})$ is 1 and is attained by Abelian orbifolds of $\IC^{n+1}$,
\begin{conjecture}
  The maximum value for $V(b_i^*; Y) \chi(\widetilde{X(\Delta_{n-1})})$ for reflexive toric Calabi-Yau $n$-folds is attained by various (not necessarily uniquely) dP$_3$ fibrations.
\end{conjecture}

The maximum value $V(b_i^*; Y) \chi(\widetilde{X(\Delta_{n-1})})$ is characteristic to the set of toric Calabi-Yau $n$-folds with reflexive polytopes as toric diagrams. Its existence for this set of Calabi-Yau $n$-folds is an incentive for us to identify the profile of the envelope in \fref{fCY34peakplots} as well as in plots of the Euler number against the inverse minimum volume in Figures \ref{f:n=2min}, \ref{f:n=3min} and \ref{f:n=4min}. We leave this to appendix \ref{ap:envelop}. 

In the hope to identify other bounds of the minimum volume, we move on to the next section where we identify the relationship between the minimum volume and other topological quantities of the toric variety -- the Chern Numbers.

\subsubsection{Minimum Volume and Chern Numbers}

Let us continue our study to see how $V(b_i^{*};Y)$ behaves relative to the other topological quantities.
For $n=3$, we recall from Proposition \ref{n=2chern} that in addition to $\chi$, there is the Chern number $C = \int c_1(\widetilde{X(\Delta_2)})^2$.
We plot $C$ against both $V(b_i^{*};Y)$ and its reciprocal in parts (a) and (b) of Figure \ref{f:n=3minC}.
Again, the 5 orange points in \fref{f:n=3minC} are the 5 Abelian orbifolds of $\IC^3$.
However, there appears to be a straight-line upper bound for the value of $V(b_i^{*};Y)$, which we numerically identify to be
\begin{equation}
V(b_i^{*};Y) \sim 3^{-3}  \int c_1(\widetilde{X(\Delta_2)})^2 \ .
\end{equation}

\begin{figure}[h!!!]
\begin{center}
\resizebox{0.7\hsize}{!}{
  \includegraphics[trim=0mm 0mm 0mm 0mm, width=8in]{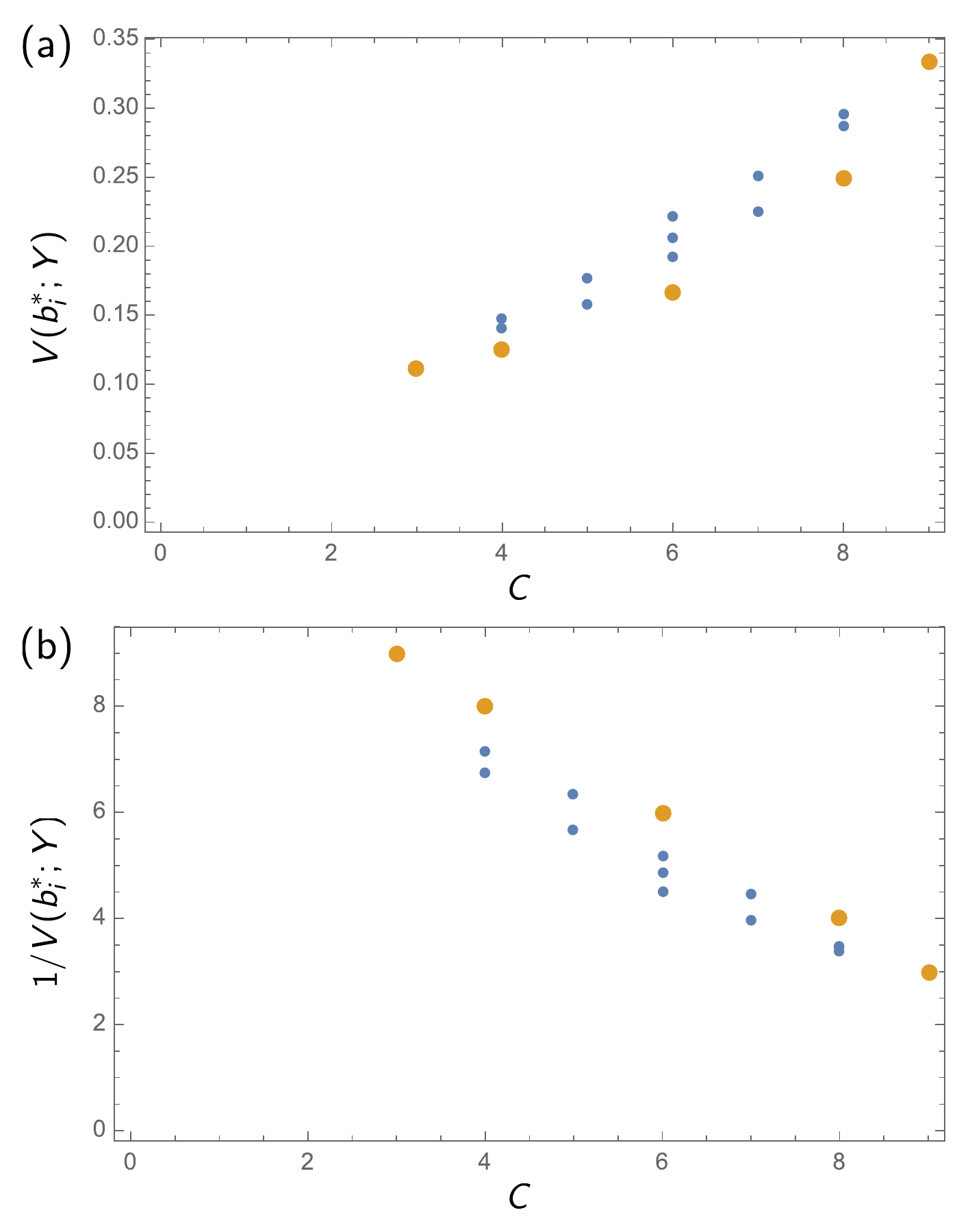}
}
\caption{{\sf {\small
      (a) The Chern number $C = \int c_1^2(\widetilde{X(\Delta_2)})$ of the 16 reflexive toric Calabi-Yau 3-folds $\cX$ against the minimum volume $V(b_i^{*}; Y)$, and (b) $C$ against the reciprocal of $V(b_i^{*}; Y)$.
      Note that the orange points correspond to the 5 Calabi-Yau 3-folds which are Abelian orbifolds of $\mathbb{C}^3$.
      We numerically find that $V(b_i^{*};Y)$ is bounded above by the line $V(b_i^{*}; Y)= 0.037 C \simeq 3^{-3} C $.
\label{f:n=3minC}}}}
\end{center}
\end{figure}

To gain more confidence in this observation, we move on to the 3-dimensional reflexive polytopes and the corresponding toric Calabi-Yau 4-folds.
Here, as discussed in Proposition \ref{n=3chern}, there is the non-trivial topological invariant $C = \int c_1(\widetilde{X(\Delta_2)})^3$.
We plot $C$ against both $V(b_i^{*};Y)$ and its reciprocal in parts (a) and (b) of Figure \ref{f:n=4minC}.
Again, the 48 orange points in \fref{f:n=4minC} are the Abelian orbifolds of $\IC^4$, which are sprinkled throughout the plot.
However, there again appears to be a straight-line upper bound on the value of $V(b_i^{*};Y)$, which we numerically find to be
\begin{equation}
V(b_i^{*};\cX) \sim 4^{-4}  \int c_1(\widetilde{X(\Delta_3)})^3 \ .
\end{equation}

\begin{figure}[h!!!]
\begin{center}
\resizebox{0.7\hsize}{!}{
  \includegraphics[trim=0mm 0mm 0mm 0mm, width=8in]{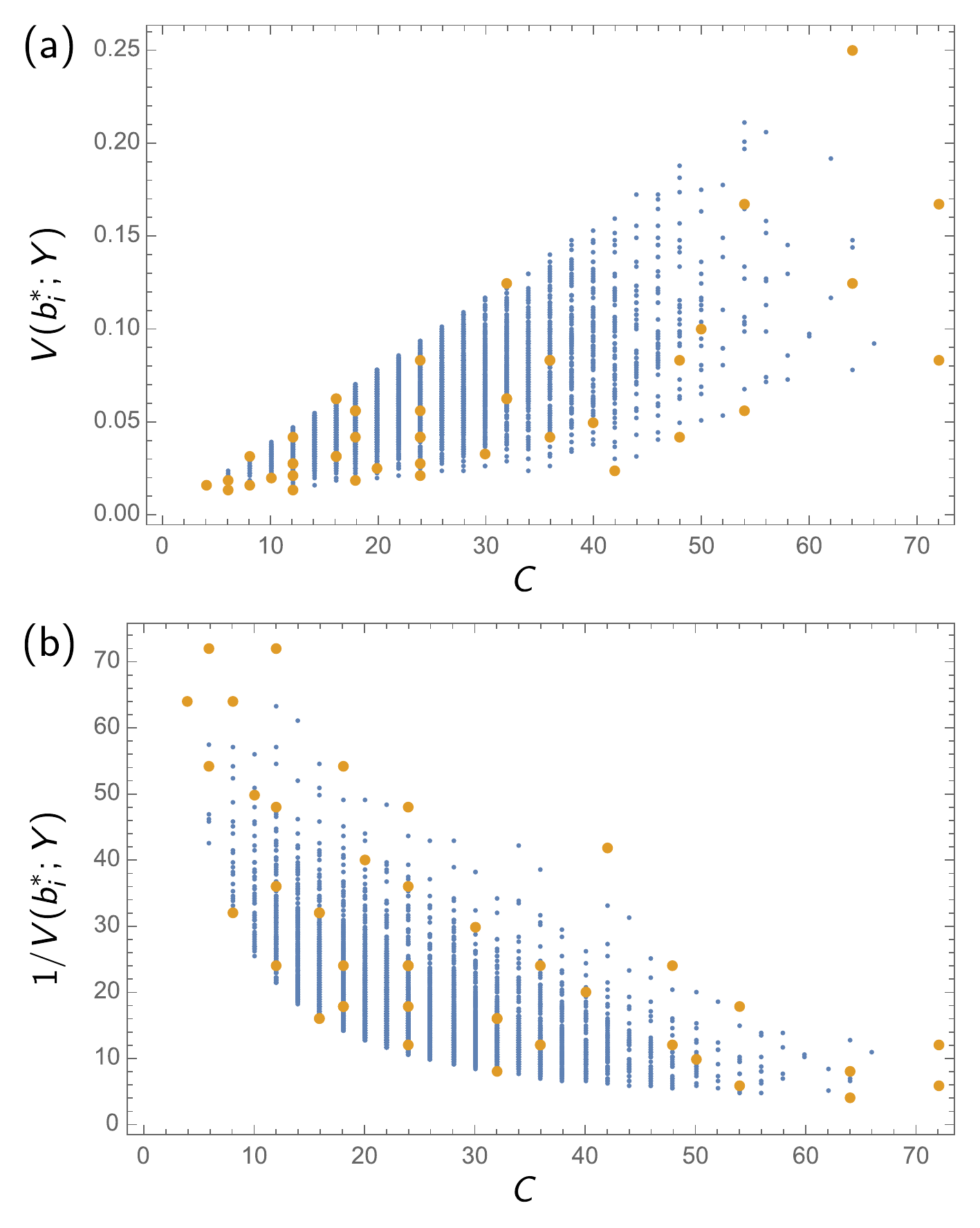}
}
\caption{{\sf {\small
      (a) The Chern number $C = \int c_1^3(\widetilde{X(\Delta_3)})$ of the 4319 reflexive toric Calabi-Yau 4-folds $\cX$ against the minimum volume $V(b_i^{*};\cX)$, and
      (b) $C$ against the reciprocal of $V(b_i^{*};\cX)$.
      Note that the orange points correspond to the 48 Calabi-Yau 4-folds which are Abelian orbifolds of $\mathbb{C}^4$.
      We numerically find that $V(b_i^{*};\cX)$ is bounded above by the line $V(b_i^{*};\cX)= 0.0039 C \simeq 4^{-4} C$.
\label{f:n=4minC}}}}
\end{center}
\end{figure}

We summarize the linear fits of the upper bound for $V(b_i^{*};Y)$ in appendix \ref{ap:envelop2}.
These curious observations, fortified by the abundance of our data, leads us to speculate that indeed while $\chi$ determines a lower bound to $V(b_i^{*};Y)$, the Chern number $\int c_1^n$ determines an upper bound to $V(b_i^{*};Y)$. We note that this observed upper bound of $V(b_i^{*};Y)$ is specific to the class of reflexive toric Calabi-Yau $n$-folds, which we are studying in this paper. As a result, we conjecture

\begin{conjecture}
For toric Calabi-Yau $n$-folds $\cX$ with the toric diagram being a reflexive polytope $\Delta_{n-1}$, the minimum volume $V(b_i^{*};Y)$ of the Sasaki-Einstein base Y lies between the following two bounds,
\beal{es521b1}
\frac{1}{\chi(\widetilde{X(\Delta_{n-1})})} 
\leq
V(b_i^{*};Y) 
\leq
m_n \int c_1(\widetilde{X(\Delta_{n-1})})^{n-1} 
~,~
\eea
where $m_3 \sim 3^{-3}$ and $m_4 \sim 4^{-4}$ (in particular $m_n>m_{n+1}$).
We note that the lower bound is universal for all toric varieties whereas the
upper bound with gradient $m_n\in \mathbb{R}$ is a feature of toric varieties originating from reflexive polytopes.
\end{conjecture}

We emphasize that the lower bound on the minimum volume $V(b_i^{*};Y)$ is conjectured to be universal for any toric Calabi-Yau $n$-fold. The lower bound is saturated when the toric Calabi-Yau $n$-fold is an Abelian orbifold of $\mathbb{C}^n$ as discussed in \sref{slowbound}.

Following section \sref{svolminamax}, for $4d$ $\mathcal{N}=1$ worldvolume theories of probe D3-branes at Calabi-Yau 3-folds singularities, via the AdS/CFT correspondence the minimum volume of the base Sasaki-Einstein manifold $Y$ is inversely proportional to the maximum of the central charge $a$-function of the theory. 
Using \eref{es54a2}, we can therefore rewrite the bound on the normalized minimum volume $V(b_i^*;Y)$ in \eref{es521b1} as a bound on the normalized $a$-function $A(R;Y)$,
\beal{es521b10}
m_3 \int c_1(\widetilde{X(\Delta_2)})^2 
\leq
A(R_i^*;Y) 
\leq
\frac{1}{\chi(\widetilde{X(\Delta_2)})} 
~,~
\eea
where $m_3 \sim 3^{-3}$.
The upper bound on $A(R_i^*;Y)$ is universal for all $4d$ $\mathcal{N}=1$ theories that are worldvolume theories of D3-branes probing toric Calabi-Yau 3-folds, whereas the lower bound is specific to toric Calabi-Yau 3-folds whose toric diagram is a reflexive polygon.\footnote{We note that a similar bound on $a$ was introduced in \cite{Gulotta:2008ef} which agrees with our bounds in \eref{es521b10}.}

The volume bound for Sasaki-Einstein 7-manifolds can be interpreted as a bound on the free energy $F$ of $3d$ $\mathcal{N}=2$ superconformal Chern-Simons matter theories living on the worldvolume of a stack of $N$ M2-branes at Calabi-Yau 4-fold singularities.
As reviewed in section \sref{sfree}, the free energy of the $3d$ theory is related in the large $N$ limit to the volume of the Sasaki-Einstein 7-base of the toric Calabi-Yau 4-fold. 
We rewrite \eref{es55a2} to define a \textit{normalized free energy}
\beal{es521b11}
\mathcal{F}(Y)^2 \equiv \frac{F(Y)^2}{F(S^7)^2} = \frac{\text{vol}[S^7]}{\text{vol}[Y]} = \frac{1}{V(b_i^*;Y)} ~.~
\eea
Then the volume bounds in \eref{es521b1} can be written as 
\beal{es521b12}
m_4 \int c_1(\widetilde{X(\Delta_3)})^3
\leq
\mathcal{F}(Y)^2
\leq
\frac{1}{\chi(\widetilde{X(\Delta_3)})} 
~,~
\eea
where $m_4 \sim 4^{-4}$.

\section{Discussions}

Our work establishes the first comprehensive study of volume minima for large classes of toric Calabi-Yau cones and their Sasaki-Einstein bases in different dimensions.
We have computed the volume minima corresponding to the 16 toric Calabi-Yau 3-folds, 4319 toric Calabi-Yau 4-folds and a subset of the 473,800,776 toric Calabi-Yau 5-folds whose toric diagrams are given by reflexive lattice polytopes $\Delta_n$.
The volume minima were then plotted against topological quantities such as the Chern number $\int c_1(\widetilde{X(\Delta_n)})^{n-1}$ and Euler number $\chi(\widetilde{X(\Delta_{n-1})})$ of the corresponding toric variety $\widetilde{X(\Delta_{n-1})}$.
By doing so, we have been able to identify bounds on the minimum volume in terms of these topological quantities. 
In particular, we have shown that the lower bound of the volume is set by the Euler number $\chi(\widetilde{X(\Delta_{n-1})})$ and the upper bound is determined by the Chern number $\int c_1(\widetilde{X(\Delta_n)})^{n-1}$ of the associated toric Calabi-Yau $n$-fold.
We have noted that the upper bound on the volume is valid for toric varieties coming from reflexive polytopes in dimensions up to $n=4$ and expect it to continue to hold for higher dimensions.

Both volume bounds are expected to hold for toric Calabi-Yau $n$-folds derived from reflexive polytopes in $n-1$ dimensions. 
Via the AdS/CFT correspondence, these volume bounds can be interpreted as bounds on the central charge $a$ of $4d$ $\mathcal{N}=1$ superconformal gauge theories on D3-branes probing a toric Calabi-Yau 3-fold. 
Similarly, we believe that the volume bound in $n=3$ can be interpreted as a bound on the free energy $F$ of a $3d$ $\mathcal{N}=2$ superconformal Chern-Simons matter theory on M2-branes probing a toric Calabi-Yau 4-fold. 
We also expect that these volume bounds in $n=3$ and $n=4$ will give future insights into $2d$ $(0,2)$ theories on D1-branes probing toric Calabi-Yau 4-folds and $0d$ $\mathcal{N}=1$ supersymmetric matrix models on $D(-1)$-branes probing toric Calabi-Yau 5-folds. 

Although not all volume minima for the 473,800,776 toric Calabi-Yau 5-folds were found for this paper, 
our results represent the largest collection of volume minima
for toric Calabi-Yau $n$-folds to date. 
We plan to extend this collection of volume minima to all 473,800,776 toric Calabi-Yau 5-folds and
to higher dimensional Calabi-Yau manifolds. 
Furthermore, in future work, we hope to extend our analysis to non-reflexive (i.e., having more than 1 single internal point) and even non-toric Calabi-Yau cones and to compute the volume minima for the corresponding base manifolds. 

With this paper we intend to initiate a new program on studying volume minima for Calabi-Yau manifolds in different dimensions with the hope to reveal new insights on the dynamics of supersymmetric gauge theories as well as the topology and geometry of Calabi-Yau and Sasaki-Einstein manifolds

\section*{Acknowledgments}
We would like to thank Guido Festuccia, Sebastian Franco, Amihay Hanany, Sangmin Lee, Wenbin Yan and Cumrun Vafa for enjoyable discussions. 
We are grateful to Ross Altman for help on fast triangulations in SAGE. 
R.-K.~S. is supported by the ERC STG grant 639220.
Y.-H.~H. would like to thank the Science and Technology Facilities Council, UK, for grant ST/J00037X/1, the Chinese Ministry of Education, for a Chang-Jiang Chair Professorship at NanKai University as well as the City of Tian-Jin for a Qian-Ren Scholarship, and Merton College, Oxford, for her enduring support.
The work of S.-T.~Y. is supported by NSF grant DMS-1159412, NSF grant PHY-0937443, and NSF grant DMS-0804454.

\appendix
\section{Algorithmic Implementations} \label{appalg}

In this appendix, let us briefly discuss the actualities of the requisite computations.
Fortunately, reflexive polytopes in dimensions 2 and 3 are built into SAGE \cite{sage} (dimension 4 is much more extensive and is still an on-going project \cite{cy3online,Altman:2014bfa}).
There are two methods in SAGE in constructing toric varieties from reflexive polytopes: (1) {\sf CPRFanoToricVariety()} and (2) {\sf ToricVariety()}.
The first is tailored for crepant partial resolutions of Fano toric varieties, which precisely correspond to our reflexive polytopes  \cite{nill}.
In particular, the module performs simplicial resolution of the polytope (with the {\sf CPRFanoToricVariety(Delta\_polar = $\Delta$, make\_simplicial = True)} option) rather efficiently.
This, as we learnt from Theorem \ref{thm:delta1}, suffices to resolve the singularities of our Fano toric variety $X(\Delta)$ up to orbifold singularities and subsequently we can compute the orbifold Euler number in $\IQ$.

On the other hand, method (2) deals with the more general situation, but requires, when a complete smoothing is needed, an FRS triangulation \cite{Altman:2014bfa} of $\Delta$,
which is implemented in SAGE (using its `topcom' package) by
\\
  {\sf PointConfiguration($\Delta$).restrict\_to\_fine\_triangulations(). restrict\_to\_regular\_triangulations().\\
    \qquad restrict\_to\_star\_triangulations(origin).triangulate()}.
  \\
The result is a complex $n$-fold, endowed with all the advantages of being compact and smooth.
The downside, as imaginable, is that this FRS triangulation is very expensive algorithmically.

To circumvent this problem, we use the trick used by \cite{Braun:2011ik,Long:2014fba} which takes advantage of the fact that a reflexive polytope has only a single interior lattice point.
Thus, one removes this point, which is the origin by convention, and performs FRS triangulation to the boundary lattice points (essentially a problem of one dimension less). Then one manually adds lines from the various boundaries triangles to the origin, forming the necessary face cones.
This algorithm usually is many orders of magnitude faster than a naive FRS triangulation.

Finally, as discussed in the text, our concept of {\it regularity} is more stringent than the one in \cite{Altman:2014bfa} and in the SAGE implementation.
This is because what is needed in the latter is to ensure the smoothness of the Calabi-Yau hypersurface in $X(\Delta)$ in a so-called MPCP (maximal projective crepant partial) resolution. In addition, we consider only triangulations with simplices of unit volume (determinant = $\pm1$).


\section{Minimum Volumes for Reflexive CY$_3$}\label{ap:refcy3}

\begin{table}[ht!!]
\resizebox{\hsize}{!}{
\begin{tabular}{c|c|cc|c}
 $\#$ & $V(b_i; \cX)$ & $b_1^{*}$ & $b_2^{*}$ &$V(b_i^{*};Y)$ \\ 
 \hline
1 
& $\frac{9}{\left(b_1-2 b_2-3\right) \left(2 b_1-b_2+3\right) \left(b_1+b_2-3\right)}$ 
& $0$ 
& $0$ 
& $\frac{1}{3}$ 
\\
2 
& $\frac{8}{\left(b_1-3 b_2-3\right) \left(b_1-b_2+3\right) \left(b_1+b_2-3\right)}$ 
& $-1$
& $0$ 
& $\frac{1}{4}$ 
\\
3 
& $-\frac{2 \left(2 b_1+b_2-12\right)}{\left(b_1-3\right) \left(b_1-b_2+3\right)  \left(b_1+b_2-3\right) \left(b_1+2 b_2+3\right)}$ 
& $4-\sqrt{13}$ 
& $0$ 
& $\frac{1}{324} \left(46+13 \sqrt{13}\right)$ 
\\
4 
& $\frac{24}{\left(-b_1-b_2+3\right) \left(b_1-b_2+3\right) \left(-b_1+b_2+3\right) \left(b_1+b_2+3\right)}$ 
& $0$ 
& $0$ 
& $\frac{8}{27}$ 
\\
5 
& $\frac{b_1+4 b_2-21}{\left(b_1+3\right) \left(b_1-2 b_2-3\right) \left(b_2-3\right) \left(b_1+b_2-3\right)}$ 
& $-0.831239$ 
& $0.278298$
& $0.225143$
\\
6 
& $-\frac{b_1^2-2 b_2 b_1+6 b_1+b_2^2+6 b_2-63}{\left(b_1-3\right) \left(3-b_2\right) \left(b_1-b_2-3\right) \left(b_1-b_2+3\right) \left(b_1+b_2+3\right)}$ 
& $\frac{1}{16} \left(57-9 \sqrt{33}\right)$ 
& $\frac{1}{16} \left(57-9 \sqrt{33}\right)$ 
& $\frac{1}{486} \left(59+11 \sqrt{33}\right)$ 
\\
7 
& $\frac{6}{\left(b_1-2 b_2-3\right) \left(b_1-b_2+3\right) \left(b_1+b_2-3\right)}$ 
& $-2$
& $-1$ 
& $\frac{1}{6}$ 
\\
8 
& $\frac{2 \left(b_1+b_2-9\right)}{\left(b_1-b_2-3\right) \left(b_2-3\right) \left(b_1+b_2-3\right) \left(b_1+b_2+3\right)}$ 
& $-\frac{3}{2} \left(-1+\sqrt{3}\right)$
& $\frac{1}{2} \left(3-\sqrt{3}\right)$ 
& $\frac{1}{3 \sqrt{3}}$ 
\\
9 
& $\frac{2 \left(b_2^2+3 b_1-27\right)}{\left(b_1+3\right) \left(b_1-b_2-3\right) \left(b_2-3\right) \left(b_2+3\right) \left(b_1+b_2-3\right)}$ 
& $6-3 \sqrt{5}$ 
& $0$ 
& $\frac{1}{108} \left(11+5 \sqrt{5}\right)$ 
\\
10 
& $-\frac{6 \left(b_1^2+b_2 b_1+b_2^2-27\right)}{\left(b_1-3\right) \left(b_1+3\right) \left(3-b_2\right) \left(b_2+3\right) \left(b_1+b_2-3\right) \left(b_1+b_2+3\right)}$ 
& $0$ 
& $0$ 
& $\frac{2}{9}$ 
\\
11 
& $\frac{b_1+3 b_2-15}{\left(b_1+3\right) \left(b_1-b_2-3\right) \left(b_2-3\right) \left(b_1+b_2-3\right)}$ 
& $-0.902531$ 
& $-0.745444$ 
& $0.157348$
\\
12 
& $-\frac{-b_2 b_1+9 b_1+9 b_2-45}{\left(b_1-3\right) \left(b_1+3\right) \left(3-b_2\right) \left(b_2+3\right) \left(b_1+b_2-3\right)}$ 
& $-0.379079$ 
& $-0.379079$
& $0.176299$
\\
13 
& $\frac{4}{\left(b_1-3\right) \left(b_1-2 b_2-3\right) \left(b_1-b_2+3\right)}$ 
& $-1$ 
& $0$ 
& $\frac{1}{8}$ 
\\
14 
& $\frac{2 \left(b_1+b_2-6\right)}{\left(b_1-3\right) \left(b_2-3\right) \left(b_1+b_2-3\right) \left(b_1+b_2+3\right)}$ 
& $1-\frac{\sqrt{7}}{2}$ 
& $1-\frac{\sqrt{7}}{2}$ 
& $\frac{4}{243} \left(-10+7 \sqrt{7}\right)$ 
\\
15 
& $\frac{12}{\left(b_1-3\right) \left(b_1+3\right) \left(b_1-b_2-3\right) \left(b_1-b_2+3\right)}$ 
& $0$ 
& $0$ 
& $\frac{4}{27}$ 
\\
16 
& $\frac{3}{\left(b_1-3\right) \left(b_1-b_2-3\right) \left(2 b_1-b_2+3\right)}$ 
& $0$ 
& $0$
& $\frac{1}{9}$ 
\\
\end{tabular}
}
\caption{{\sf {\small
     Volume functions $V(b_i;Y)$ with their minima for all 16 reflexive toric Calabi-Yau 3-folds, in the order as indicated in Figure \ref{f:n=2}.
     As throughout the paper, $b_3$ is set to 3.
     Some values are given numerically because of their length, the exact algebraic expressions are given in the text.
 }}
\label{t:voln=2}
\label{tcy3volminA}
}
\end{table}

In this section, we give the volume functions of all 16 Calabi-Yau 3-folds corresponding to the reflexive polygons.
Numbers 1,2,7,13,16 are the 5 Abelian orbifolds of $\IC^3$, while numbers 1,3,4,6,10 correspond to the 5 smooth bases, viz., the del Pezzo surfaces.
In the case of number 5, $b_1^*$ and $b_2^*$ are the roots of
\begin{equation}
 48 + 59 b_1 - 5 b_1^2 - 7 b_1^3 + b_1^4 = 0 \ , \qquad
 18 - 79 b_2 + 55 b_2^2 - 13 b_2^3 + b_2^4 = 0
 \ ,
\end{equation}
and the final minimized volume $V(b_i^*; Y)$ is the unique positive root of
\begin{equation}
 -1 -235 x - 1740 x^2 + 7200 x^3 + 23328 x^4 = 0 \ .
\end{equation}
In the case of number 11, $b_1^*$ and $b_2^*$ are the roots of
\begin{equation}
 18 + 15 b_1 - 9 b_1^2 - 3 b_1^3 + b_1^4 = 0
 \ , \qquad
 -48 -26 b_22 + 42 b_22^2 - 12 b_22^3 + b_2^4 = 0
 \ ,
\end{equation}
and $V(b_i^*; Y)$ is the unique positive root of
\begin{equation}
 -1 -105 x - 84 x^2 + 3808 x^3 + 7776 x^4 = 0
 \ .
\end{equation}
In the case of number 12, $b_1^*=b_2^*$ are roots of
\begin{equation}
 18 + 39 b_1 - 22 b_1^2 + b_1^3 = 0
 \ ,
\end{equation}
and $V(b_i^*; Y)$ is the unique positive root of
\begin{equation}
 -1 - 66545 x + 210924 x^2 + 944784 x^3 = 0
 \ .
\end{equation}

\paragraph{Quadratic Irrationals and Quasi-Regulars.}
The minimum volumes $V(b_i^{*};Y)$ of Sasaki-Einstein manifolds $Y$ are known to be {\em algebraic}. When in particular, they are {\em quadratic irrationals} (i.e., of the form $a+b\sqrt{c}$ with $a,b\in\mathbb{Q}$ and $c\in\mathbb{N}$) \cite{Martelli:2004wu}, we call the Sasaki-Einstein manifolds {\bf quasi-regular} and when further $V(b_i^{*};Y) \in \mathbb{Q}$, $Y$ is called {\bf regular}.

In the case of Sasaki-Einstein 5-manifolds related to reflexive toric Calabi-Yau 3-folds, we can compute the volumes exactly as summarized above. Therefore, we can identify precisely which cases refer to quasi-regular Sasaki-Einstein manifolds as highlighted in \fref{fcy3regular}.

\begin{figure}[H]
\begin{center}
\resizebox{0.6\hsize}{!}{
  \includegraphics[trim=0mm 0mm 0mm 0mm, width=8in]{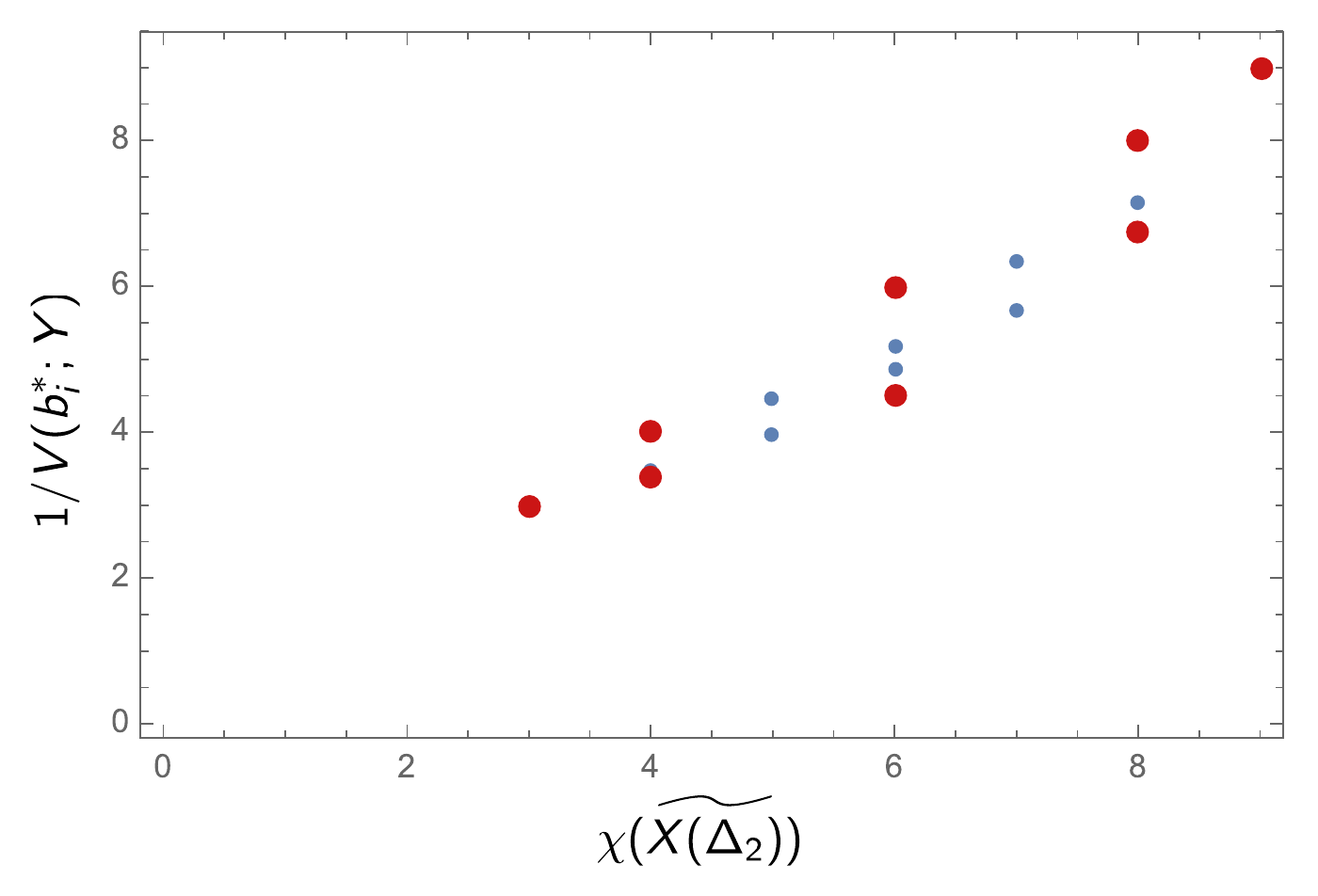}
}
\caption{{\sf {\small
The Euler number of the 16 reflexive toric Calabi-Yau 3-folds $\mathcal{X}$ against the inverse minimum volume $1/V(b_i^{*};Y)$. All quasi-regular cases with rational $V(b_i^{*};Y)$ are highlighted in red.
\label{fcy3regular}}}}
\end{center}
\end{figure}

\section{Minimum Volumes for Abelian Orbifolds of $\mathbb{C}^4$}\label{ap:orb}
In this section, we list the minimum for all the volume functions for the 48 Abelian orbifolds of $\IC^4$ corresponding to reflexive polyhedra $\Delta_3$.
The Reeb vector is here $(b_1,b_2,b_3,b_4)$ with $b_4$ set to 4.

\begin{table}[H]
\resizebox{\hsize}{!}{
\begin{tabular}{c|c|ccc|c}
\# & $V(b_i; \cX)$ & $b_1^{*}$ & $b_2^{*}$ & $b_3^{*}$ & $V(b_i^{*};Y)$
\\
\hline
1& $-\frac{64}{(b_{1}+b_{2}-3 b_{3}-b_{4}) (b_{1}-3b_{2}+b_{3}-b_{4}) (b_{1}+b_{2}+b_{3}-b_{4}) (3b_{1}-b_{2}-b_{3}+b_{4})}$
& $0$ & $0$& $0$
& $1/4$ 
\\
2 & $-\frac{54}{(b_{1}+b_{2}-5 b_{3}-b_{4}) (b_{1}-2 b_{2}+b_{3}-b_{4}) (b_{1}+b_{2}+b_{3}-b_{4}) (2 b_{1}-b_{2}-b_{3}+b_{4})}$
& $-1$ & $-1$& $0$
& $1/6$ 
\\
3 & $-\frac{72}{(b_{1}+b_{2}-5 b_{3}-b_{4}) (b_{1}-5 b_{2}+b_{3}-b_{4}) (b_{1}+b_{2}+b_{3}-b_{4}) (b_{1}-b_{2}-b_{3}+b_{4})}$
& $-2$ & $0$& $0$
& $1/6$ 
\\
9 & $-\frac{32}{(b_{1}+b_{2}-b_{4}) (3 b_{1}-b_{2}+b_{4})(b_{1}+b_{2}-2 b_{3}-b_{4}) (b_{1}-3 b_{2}+2b_{3}-b_{4})}$
& $0$ & $0$& $0$
& $1/8$ 
\\
10 & $-\frac{64}{(b_{1}+b_{2}-7 b_{3}-b_{4}) (b_{1}-3 b_{2}+b_{3}-b_{4}) (b_{1}+b_{2}+b_{3}-b_{4}) (b_{1}-b_{2}-b_{3}+b_{4})}$
& $-3$ & $-1$& $0$
& $1/8$ 
\\
32 & $-\frac{50}{(b_{1}+b_{2}-4 b_{3}-b_{4}) (b_{1}-4b_{2}+b_{3}-b_{4}) (b_{1}+b_{2}+b_{3}-b_{4}) (b_{1}-b_{2}-b_{3}+b_{4})}$
& $-4$ & $-1$& $-1$
& $1/10$ 
\\
86 & $-\frac{36}{(b_{1}+b_{2}-b_{4}) (b_{1}-b_{2}+b_{4}) (b_{1}+b_{2}-3 b_{3}-b_{4}) (b_{1}-5 b_{2}+3 b_{3}-b_{4})}$
& $-2$ & $0$& $0$
& $1/12$ 
\\
87 & $-\frac{24}{(b_{1}+b_{2}-b_{4}) (2 b_{1}-b_{2}+b_{4}) (b_{1}+b_{2}-2 b_{3}-b_{4}) (b_{1}-3 b_{2}+2 b_{3}-b_{4})}$
& $-1$ & $-1$& $-1$
& $1/12$ 
\\
88 & $-\frac{36}{(b_{1}+b_{2}-3 b_{3}-b_{4}) (b_{1}-2 b_{2}+b_{3}-b_{4}) (b_{1}+b_{2}+b_{3}-b_{4}) (2 b_{1}-b_{2}-b_{3}+b_{4})}$
& $-3$ & $-3$& $-2$
& $1/12$ 
\\
89 & $-\frac{72}{(b_{1}+b_{2}-11 b_{3}-b_{4}) (b_{1}-2 b_{2}+b_{3}-b_{4}) (b_{1}+b_{2}+b_{3}-b_{4}) (b_{1}-b_{2}-b_{3}+b_{4})}$
& $-5$ & $-3$& $0$
& $1/12$ 
\\
90 & $-\frac{48}{(b_{1}+b_{2}-5 b_{3}-b_{4}) (b_{1}-3 b_{2}+b_{3}-b_{4}) (b_{1}+b_{2}+b_{3}-b_{4}) (b_{1}-b_{2}-b_{3}+b_{4})}$
& $-5$ & $-2$& $-1$
& $1/12$ 
\\
428 & $-\frac{16}{(b_{1}-b_{4}) (b_{1}-2 b_{2}-b_{4}) (b_{1}-2 b_{3}-b_{4}) (3 b_{1}-2 b_{2}-2 b_{3}+b_{4})}$
& $0$ & $0$& $0$
& $1/16$ 
\\
429 & $-\frac{32}{(b_{1}+b_{2}-b_{4}) (b_{1}-b_{2}+b_{4}) (b_{1}+b_{2}-4 b_{3}-b_{4}) (b_{1}-3 b_{2}+2  b_{3}-b_{4})}$
& $-3$ & $-1$& $0$
& $1/16$ 
\\
430 & $-\frac{16}{(b_{1}+b_{2}-b_{4}) (3 b_{1}-b_{2}+b_{4})  (b_{1}+b_{2}-b_{3}-b_{4}) (b_{1}-3 b_{2}+b_{3}-b_{4})}$
& $0$ & $0$& $0$
& $1/16$ 
\\
431 & $-\frac{32}{(b_{1}-3 b_{2}-b_{4}) (b_{1}+b_{2}-b_{4}) (b_{1}+b_{2}-4 b_{3}-b_{4}) (b_{1}-b_{2}-b_{3}+b_{4})}$
& $-3$ & $-1$& $0$
& $1/16$ 
\\
432 & $-\frac{32}{(b_{1}-7 b_{3}-b_{4}) (b_{1}+b_{3}-b_{4}) (b_{1}-2 b_{2}+b_{3}-b_{4}) (b_{1}-b_{2}-b_{3}+b_{4})}$
& $-4$ & $-2$& $0$
& $1/16$ 
\\
742 & $-\frac{24}{(b_{1}+b_{2}-b_{4}) (b_{1}-b_{2}+b_{4}) (b_{1}+b_{2}-2 b_{3}-b_{4}) (b_{1}-5 b_{2}+2 b_{3}-b_{4})}$
& $-2$ & $0$& $0$
& $1/18$ 
\\
743 & $-\frac{18}{(b_{1}+b_{2}-b_{4}) (2 b_{1}-b_{2}+b_{4}) (b_{1}+b_{2}-2 b_{3}-b_{4}) (b_{1}-2 b_{2}+b_{3}-b_{4})}$
& $-1$ & $-1$& $0$
& $1/18$ 
\\
744 & $\frac{18}{(b_{1}+b_{2}-3 b_{3}-b_{4}) (b_{1}+b_{2}-b_{3}-b_{4}) (b_{1}-2 b_{2}+b_{3}-b_{4}) (-2 b_{1}+b_{2}+b_{3}-b_{4})}$
& $-1$ & $-1$& $0$
& $1/18$ 
\\
745 & $-\frac{18}{(b_{1}-5 b_{3}-b_{4}) (b_{1}+b_{3}-b_{4}) (b_{1}-b_{2}+b_{3}-b_{4}) (2 b_{1}-b_{2}-b_{3}+b_{4})}$
& $-2$ & $-3$& $0$
& $1/18$ 
\\
746 & $-\frac{54}{(b_{1}+b_{2}-8 b_{3}-b_{4}) (b_{1}-2 b_{2}+b_{3}-b_{4}) (b_{1}+b_{2}+b_{3}-b_{4}) (b_{1}-b_{2}-b_{3}+b_{4})}$
& $-8$ & $-5$& $-1$
& $1/18$ 
\\
1115 & $-\frac{40}{(b_{1}+b_{2}-4 b_{3}-b_{4}) (b_{1}-3 b_{2}+b_{3}-b_{4}) (b_{1}+b_{2}+b_{3}-b_{4}) (b_{1}-b_{2}-b_{3}+b_{4})}$
& $-9$ & $-4$& $-3$
& $1/20$ 
\\
1944 & $-\frac{18}{(b_{1}-b_{4}) (b_{1}-3 b_{2}-b_{4}) (b_{1}-3  b_{3}-b_{4}) (b_{1}-b_{2}-b_{3}+b_{4})}$
& $-2$ & $0$& $0$
& $1/24$ 
\\
1945 & $-\frac{24}{(b_{1}+b_{2}-b_{4}) (b_{1}-b_{2}+b_{4})(b_{1}+b_{2}-3 b_{3}-b_{4}) (b_{1}-3 b_{2}+2  b_{3}-b_{4})}$
& $-5$ & $-3$& $-2$
& $1/24$ 
\\
1946 & $-\frac{12}{(b_{1}+b_{2}-b_{4}) (2 b_{1}-b_{2}+b_{4}) (b_{1}+b_{2}-b_{3}-b_{4}) (b_{1}-3 b_{2}+b_{3}-b_{4})}$
& $-1$ & $-1$& $-2$
& $1/24$ 
\\
1947 & $-\frac{36}{(b_{1}-2 b_{2}-b_{4}) (b_{1}+b_{2}-b_{4}) (b_{1}+b_{2}-6 b_{3}-b_{4}) (b_{1}-b_{2}-b_{3}+b_{4})}$
& $-5$ & $-3$& $0$
& $1/24$ 
\\
\end{tabular}
}
\caption{
{\sf {\small
Volume functions $V(b_i;Y)$ with their minima for reflexive toric Calabi-Yau 4-folds that are Abelian orbifolds of $\mathbb{C}^4$. 
\textbf{(Part 1/2)}
}
\label{tcy4volmin}}
}
\end{table}

\begin{table}[H]
\resizebox{\hsize}{!}{
\begin{tabular}{c|c|ccc|c}
\# & $V(b_i; \cX)$ & $b_1^{*}$ & $b_2^{*}$ & $b_3^{*}$ & $V(b_i^{*};Y)$
\\
\hline
1948 & $-\frac{24}{(b_{1}-3 b_{2}-b_{4}) (b_{1}+b_{2}-b_{4}) (b_{1}+b_{2}-3 b_{3}-b_{4}) (b_{1}-b_{2}-b_{3}+b_{4})}$
& $-6$ & $-2$& $-2$
& $1/24$ 
\\
1949 & $-\frac{24}{(b_{1}-5 b_{3}-b_{4}) (b_{1}+b_{3}-b_{4}) (b_{1}-2 b_{2}+b_{3}-b_{4}) (b_{1}-b_{2}-b_{3}+b_{4})}$
& $-7$ & $-4$& $-1$
& $1/24$ 
\\
1950 & $-\frac{48}{(b_{1}+b_{2}-7 b_{3}-b_{4}) (b_{1}-2b_{2}+b_{3}-b_{4}) (b_{1}+b_{2}+b_{3}-b_{4})(b_{1}-b_{2}-b_{3}+b_{4})}$
& $-11$ & $-7$& $-2$
& $1/24$ 
\\
3039 & $-\frac{30}{(b_{1}-2 b_{2}-b_{4}) (b_{1}+b_{2}-b_{4}) (b_{1}+b_{2}-5 b_{3}-b_{4}) (b_{1}-b_{2}-b_{3}+b_{4})}$
& $-7$ & $-4$& $-1$
& $1/30$ 
\\
3313 & $-\frac{16}{(b_{1}+b_{2}-b_{4}) (b_{1}-b_{2}+b_{4}) (b_{1}+b_{2}-2 b_{3}-b_{4}) (b_{1}-3 b_{2}+b_{3}-b_{4})}$
& $-3$ & $-1$& $0$
& $1/32$ 
\\
3314 & $-\frac{8}{(b_{1}-b_{4}) (b_{1}-2 b_{2}-b_{4}) (b_{1}-b_{3}-b_{4}) (3 b_{1}-2 b_{2}-b_{3}+b_{4})}$
& $0$ & $0$& $0$
& $1/32$ 
\\
3315 & $-\frac{16}{(b_{1}-b_{2}+b_{4}) (b_{1}+b_{2}-3 b_{3}-b_{4})(b_{1}+b_{2}-b_{3}-b_{4}) (b_{1}-3 b_{2}+2 b_{3}-b_{4})}$
& $-3$ & $-1$& $0$
& $1/32$ 
\\
3316 & $-\frac{16}{(b_{1}-b_{4}) (b_{1}-2 b_{2}-b_{4}) (b_{1}-4 b_{3}-b_{4}) (b_{1}-b_{2}-b_{3}+b_{4})}$
& $-4$ & $-2$& $0$
& $1/32$ 
\\
3726 & $-\frac{24}{(b_{1}+b_{2}-b_{4}) (b_{1}-b_{2}+b_{4}) (b_{1}+b_{2}-4 b_{3}-b_{4}) (b_{1}-2 b_{2}+b_{3}-b_{4})}$
& $-5$ & $-3$& $0$
& $1/36$ 
\\
3727 & $-\frac{12}{(b_{1}+b_{2}-b_{4}) (b_{1}-b_{2}+b_{4}) (b_{1}+b_{2}-b_{3}-b_{4}) (b_{1}-5 b_{2}+b_{3}-b_{4})}$
& $-2$ & $0$& $0$
& $1/36$ 
\\
3728 & $\frac{12}{(b_{1}-3 b_{3}-b_{4}) (b_{1}+b_{3}-b_{4}) (b_{1}-b_{2}+b_{3}-b_{4}) (-2 b_{1}+b_{2}+b_{3}-b_{4})}$
& $-6$ & $-9$& $-2$
& $1/36$ 
\\
3994 & $-\frac{20}{(b_{1}-4 b_{3}-b_{4}) (b_{1}+b_{3}-b_{4}) (b_{1}-2 b_{2}+b_{3}-b_{4}) (b_{1}-b_{2}-b_{3}+b_{4})}$
& $-13$ & $-8$& $-3$
& $1/40$ 
\\
4081 & $-\frac{42}{(b_{1}+b_{2}-6 b_{3}-b_{4}) (b_{1}-2b_{2}+b_{3}-b_{4}) (b_{1}+b_{2}+b_{3}-b_{4}) (b_{1}-b_{2}-b_{3}+b_{4})}$
& $-20$ & $-13$& $-5$
& $1/42$ 
\\
4229 & $-\frac{12}{(b_{1}-b_{4}) (b_{1}-2 b_{2}-b_{4}) (b_{1}-3 b_{3}-b_{4}) (b_{1}-b_{2}-b_{3}+b_{4})}$
& $-8$ & $-4$& $-2$
& $1/48$ 
\\
4230 & $-\frac{24}{(b_{1}-2 b_{2}-b_{4}) (b_{1}+b_{2}-b_{4}) (b_{1}+b_{2}-4 b_{3}-b_{4}) (b_{1}-b_{2}-b_{3}+b_{4})}$
& $-13$ & $-7$& $-4$
& $1/48$ 
\\
4256 & $-\frac{10}{(b_{1}+b_{2}-b_{4}) (b_{1}-b_{2}+b_{4}) (b_{1}+b_{2}-b_{3}-b_{4}) (b_{1}-4 b_{2}+b_{3}-b_{4})}$
& $-4$ & $-2$& $-5$
& $1/50$ 
\\
4282 & $\frac{6}{(b_{1}-b_{4}) (b_{1}-b_{2}-b_{4}) (b_{1}-2 b_{3}-b_{4}) (-2 b_{1}+b_{2}+b_{3}-b_{4})}$
& $-2$ & $-3$& $-5$
& $1/54$ 
\\
4283 & $-\frac{18}{(b_{1}+b_{2}-b_{4}) (b_{1}-b_{2}+b_{4}) (b_{1}+b_{2}-3 b_{3}-b_{4}) (b_{1}-2  b_{2}+b_{3}-b_{4})}$
& $-8$ & $-6$& $-3$
& $1/54$ 
\\
4312 & $\frac{4}{(b_{1}-b_{4}) (b_{1}-b_{2}-b_{4}) (b_{1}-b_{3}-b_{4}) (-3 b_{1}+b_{2}+b_{3}-b_{4})}$
& $0$ & $0$& $0$
& $1/64$ 
\\
4313 & $-\frac{8}{(b_{1}-b_{4}) (b_{1}-b_{2}+b_{4}) (b_{1}-2 b_{3}-b_{4}) (b_{1}-2 b_{2}+b_{3}-b_{4})}$
& $-4$ & $-2$& $0$
& $1/64$ 
\\
4318 & $-\frac{6}{(b_{1}-b_{4}) (b_{1}-b_{2}+b_{4}) (b_{1}-b_{3}-b_{4}) (b_{1}-3 b_{2}+b_{3}-b_{4})}$
& $-2$ & $0$& $0$
& $1/72$ 
\\
4319 & $-\frac{12}{(b_{1}-b_{2}+b_{4}) (b_{1}+b_{2}-3 b_{3}-b_{4}) (b_{1}+b_{2}-b_{3}-b_{4}) (b_{1}-2 b_{2}+b_{3}-b_{4})}$
& $-5$ & $-3$& $0$
& $1/72$ 
\\
\end{tabular}
}
\caption{
{\sf {\small
Volume functions $V(b_i;Y)$ with their minima for reflexive toric Calabi-Yau 4-folds that are Abelian orbifolds of $\mathbb{C}^4$. 
\textbf{(Part 2/2)}
}
\label{tcy4volmin2}}
}
\end{table}

\section{Enveloping Curves for $V(b_i^{*};Y)$}

In this section, we numerically fit for the enveloping shape for the minimal volume $V(b_i^{*};Y)$ with respect to the topological quantities we have discussed in section \sref{stop}.
\subsection{Envelop for $V(b_i^{*};Y) \chi(\widetilde{X(\Delta)})$}
\label{ap:envelop}
In Figure \ref{fCY34boundaryfits}, we present the plots of $1/V(b_i^{*};Y)$ against $\chi(\widetilde{X(\Delta)})$ for Calabi-Yau cones $\cX$ of dimension 3 and 4.
The orange line represents where $1/V(b_i^{*};Y) = \chi(\widetilde{X(\Delta)})$ corresponding to Abelian orbifolds. We fit the lower enveloping curve in blue.

\begin{figure}[H]
\begin{center}
\resizebox{0.7\hsize}{!}{
  \includegraphics[trim=0mm 0mm 0mm 0mm, width=8in]{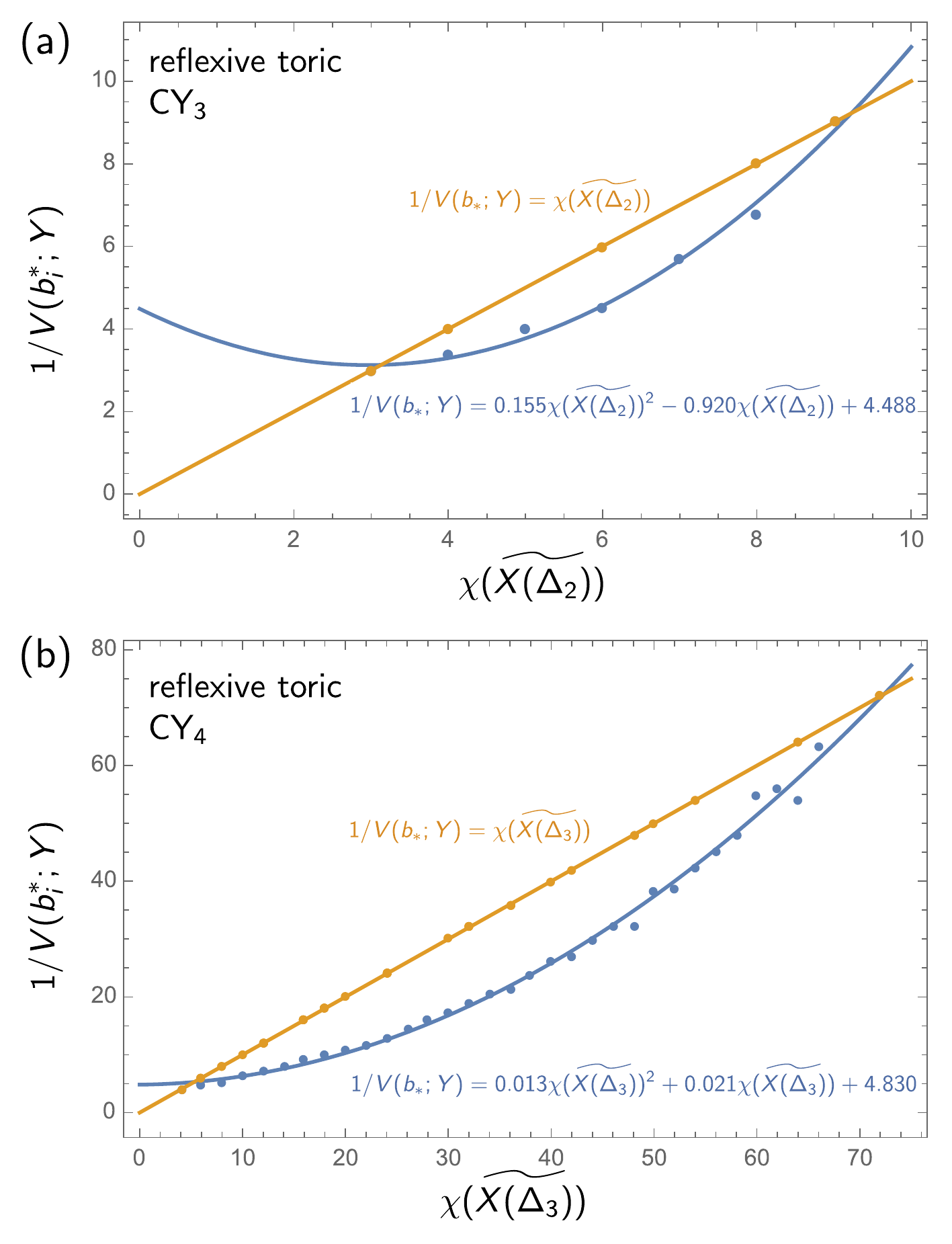}
}
\caption{{\sf {\small
      The reciprocal of the volume minimum has a bound at $1/V(b_i^{*};Y)=\chi(\widetilde{X(\Delta_n)})$ for
      (a) reflexive toric Calabi-Yau 3-folds and (b) 4-folds.
      Cf.~Figure~\ref{fCY34peakplots}.
      The orange straight lines indicate where $V(b_i^{*};Y) \chi(\widetilde{X(\Delta_n)})$ attains the minimum of 1, corresponding to the Abelian orbifolds.
      The blue lines correspond to the other extreme of the profile and have been fitted by regression.
\label{fCY34boundaryfits}}}}
\end{center}
\end{figure}

\subsection{Maximum of $V(b_i^{*};Y)$}
\label{ap:envelop2}
In Figure \ref{f:n=34boundChern}, we plot the minimum volume $V(b_i^{*};Y)$ against the Chern numbers $\int_{\widetilde{X(\Delta)}} c_1^n({\widetilde{X(\Delta_n)}})$ for $n=2,3$, corresponding to Calabi-Yau cones of dimension 3 and 4.
We fit the top envelop, which appears to be a straight line.
\begin{figure}[H]
\begin{center}
\resizebox{0.7\hsize}{!}{
  \includegraphics[trim=0mm 0mm 0mm 0mm, width=8in]{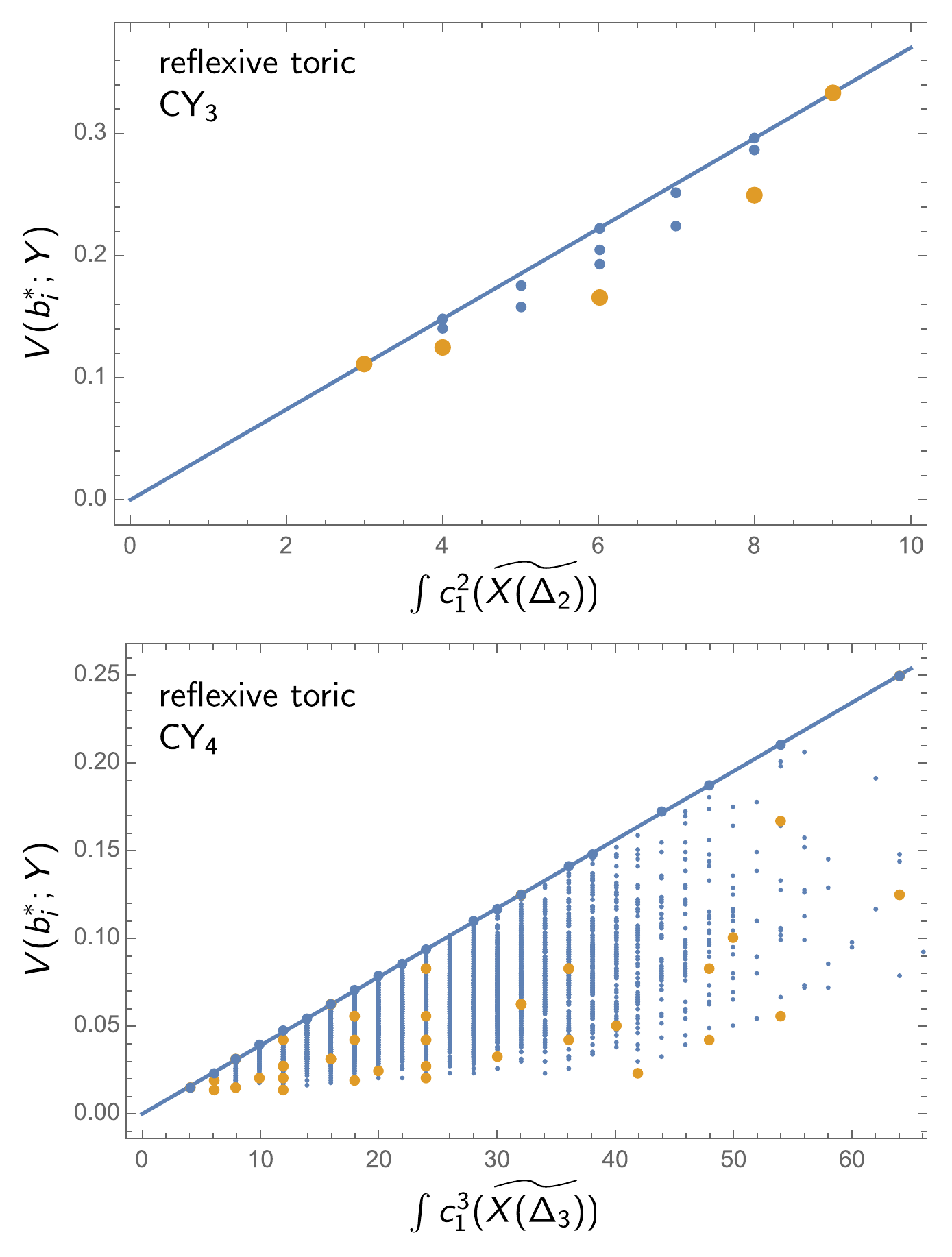}
}
\caption{{\sf {\small
The minimum volume $V(b_i^{*};Y)$ seems bounded above by a linear relation $V(b_i^{*};Y) = m_{n} \int c_1^n$ for $\cX$ being an Calabi-Yau $(n+1)$-cone, where $m_2=0.037$ and $m_3=0.0039$.
\label{f:n=34boundChern}}}}
\end{center}
\end{figure}

\section{Dual Reflexive Polytopes}

\begin{figure}[H]
\begin{center}
\resizebox{\hsize}{!}{
  \includegraphics[trim=0mm 0mm 0mm 0mm, width=8in]{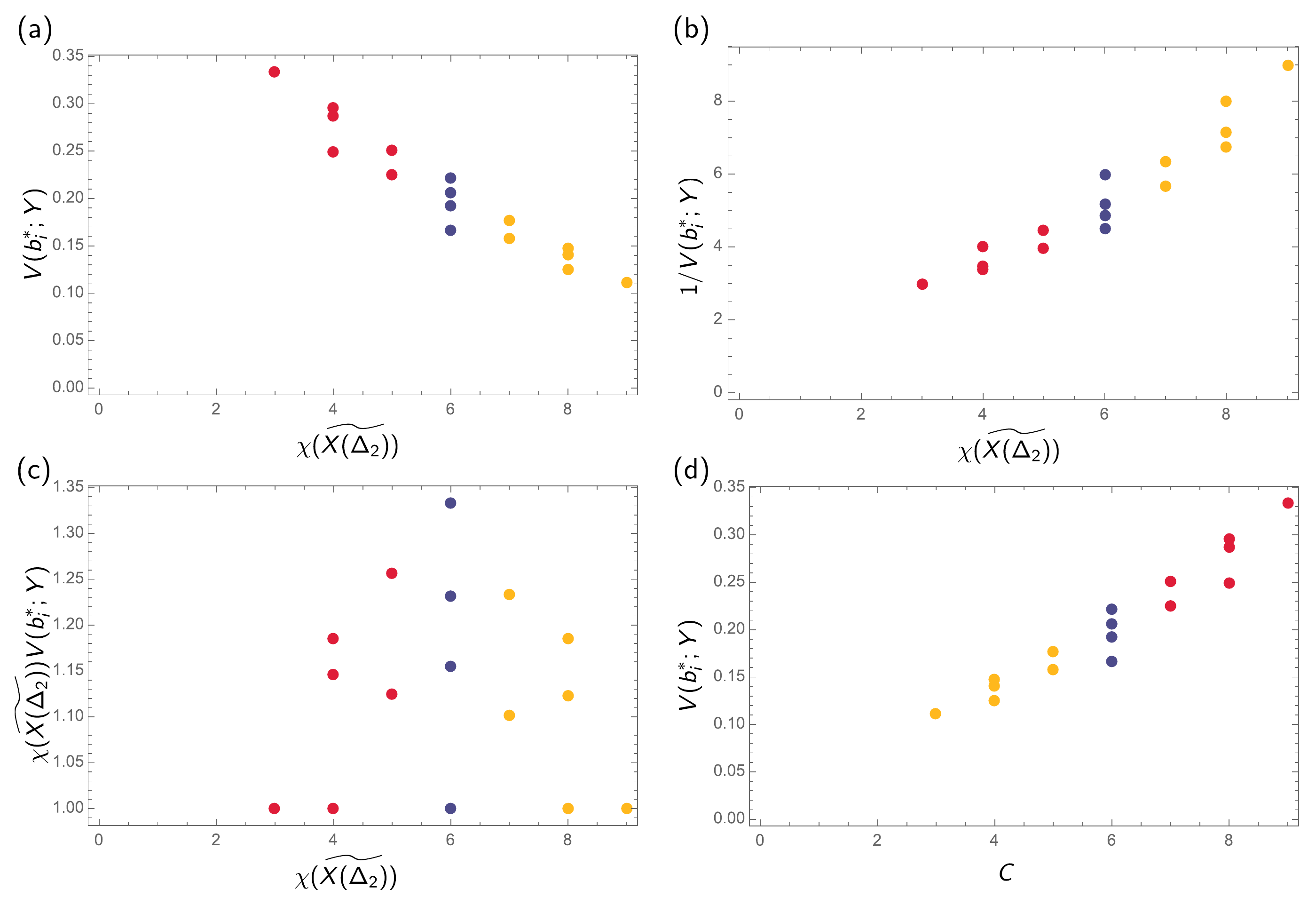}
}
\caption{{\sf {\small
(a) The Euler number  $\chi(\widetilde{X(\Delta_2)})$ against the minimum volume $V(b_i^*;Y)$, (b) the Euler number against the inverse of the minimum volume, (c) the Euler number against $\chi(\widetilde{X(\Delta_2)}) V(b_i^{*};Y)$, and (d) the Chern number $C = \int c_1^2(\widetilde{X(\Delta_2)})$ against the minimum volume for the 16 reflexive toric Calabi-Yau 3-folds $\cX$. Red and yellow points refer to dual reflexive polytopes $\Delta_2$ and the blue points refer to self-dual reflexive polytopes.
\label{fplotcy3refdualdias}}}}
\end{center}
\end{figure}

As discussed in section \sref{sreflexive}, reflexive polytopes $\Delta_n$ have dual polytopes $\Delta_n^\circ$ that are also lattice polytopes. Some reflexive polytopes are self-dual. In the case of the 16 $n=2$ reflexive polygons there are 5 self-dual polytopes while in the case of the 4319 $n=3$ reflexive polytopes there are 79 polytopes which are self-dual.

In \fref{fplotcy3refdualdias} and \fref{fplotcy4refdualdias}, we highlight dual polytopes and self-dual polytopes while plotting different topological quantities of $\widetilde{X(\Delta_n)}$ against the minimum volume. 
 
\begin{figure}[H]
\begin{center}
\resizebox{\hsize}{!}{
  \includegraphics[trim=0mm 0mm 0mm 0mm, width=8in]{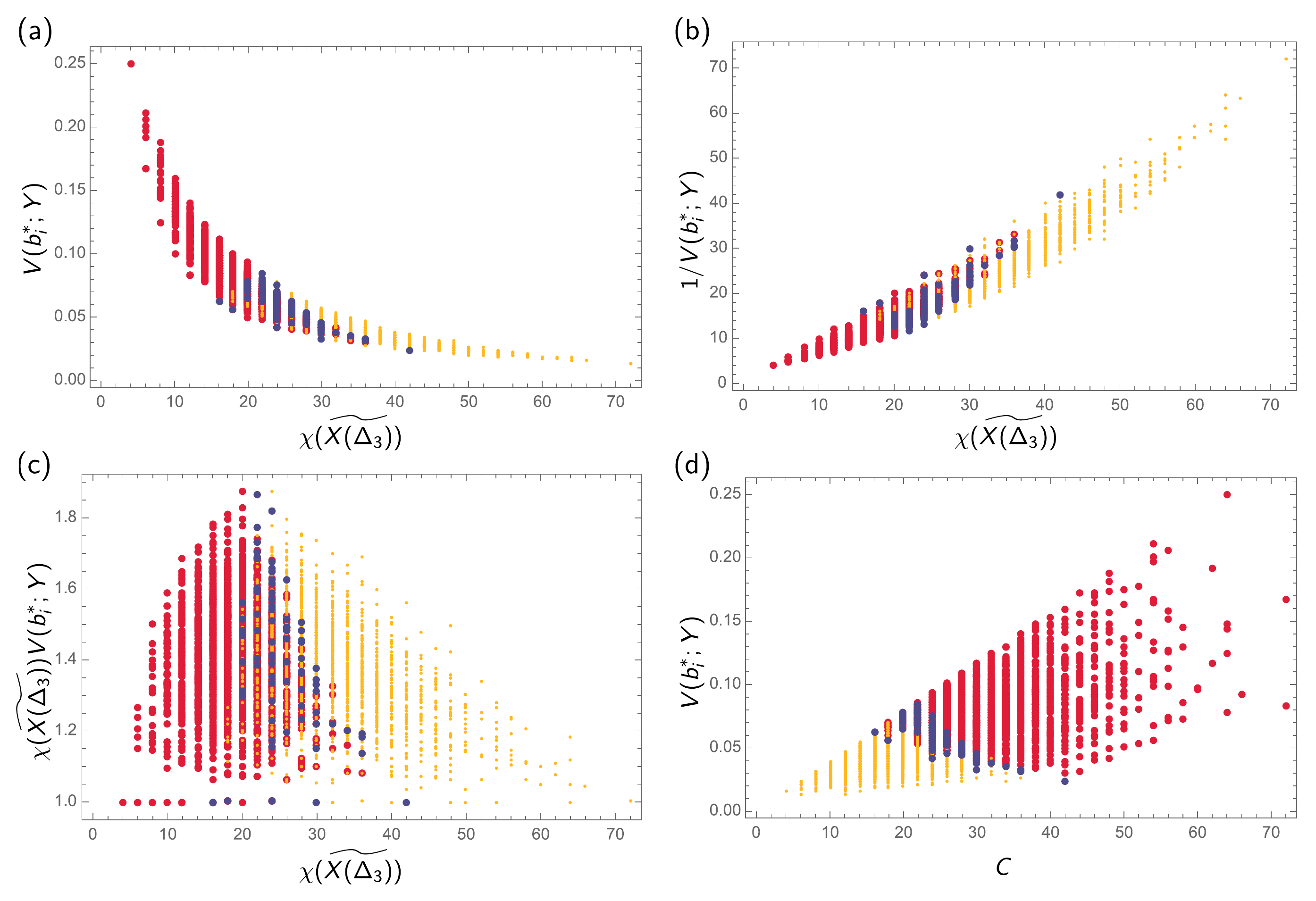}
}
\caption{{\sf {\small
(a) The Euler number  $\chi(\widetilde{X(\Delta_3)})$ against the minimum volume $V(b_i^*;Y)$, (b) the Euler number against the inverse of the minimum volume, (c) the Euler number against $\chi(\widetilde{X(\Delta_3)}) V(b_i^{*};Y)$, and (d) the Chern number $C = \int c_1^3(\widetilde{X(\Delta_3)})$ against the minimum volume for the 4319 reflexive toric Calabi-Yau 4-folds $\cX$. Red and yellow points refer to dual reflexive polytopes $\Delta_3$ and the blue points refer to self-dual reflexive polytopes.
\label{fplotcy4refdualdias}}}}
\end{center}
\end{figure}

\newpage


\bibliographystyle{jhep}
\bibliography{mybib}

\end{document}